\documentclass[manyauthors]{fundam}

\pdfoutput=1

 \usepackage{url} 
\usepackage[ruled,lined]{algorithm2e}
\usepackage{graphicx}

\usepackage[T1]{fontenc}
\usepackage[utf8]{inputenc}
\DeclareUnicodeCharacter{FB01}{fi}
\DeclareUnicodeCharacter{FB02}{fi}

\usepackage{newunicodechar}
\newunicodechar{²}{\ensuremath{{}^2}}

\usepackage[T1,T2A]{fontenc}
\usepackage{textcomp}

\usepackage{symbolic-G}
\graphicspath{{./figures/}}

\def\hlmarking{\ensuremath{M}}
\newcommand{\hlconf}{\ensuremath{C}}
\newcommand{\hlconfs}{\ensuremath{\mathcal{C}}}
\newcommand{\hlext}{\ensuremath{D}}

\usepackage{stmaryrd}

\begin{document}


\setcounter{page}{313}
\publyear{24}
\papernumber{2196}
\volume{192}
\issue{3-4}

\finalVersionForARXIV

\title{Taking Complete Finite Prefixes To High Level, Symbolically\footnote[1]{ This is a revised and extended version of the article \cite{WurdemannCH23} that is published in the Proceedings of PETRI NETS 2023.}}

\author{Nick W\"urdemann\thanks{Address for correspondence: Department of Computing Science,  University
                          of Oldenburg, Oldenburg, Germany}
\\
Department of Computing Science,\\
University of Oldenburg,
Oldenburg, Germany\\
wuerdemann@informatik.uni-oldenburg.de
\and Thomas Chatain\\
Universit\'e Paris-Saclay, INRIA and LMF, CNRS and\\ ENS Paris-Saclay, Gif-sur-Yvette, France\\
thomas.chatain@inria.fr
\and Stefan Haar\\
Universit\'e Paris-Saclay, INRIA and LMF, CNRS and\\ ENS Paris-Saclay, Gif-sur-Yvette, France\\
stefan.haar@inria.fr
\and Lukas Panneke\\
Department of Computing Science,\\
University of Oldenburg,
Oldenburg, Germany\\
lukas.panneke@informatik.uni-oldenburg.de}

\maketitle

\runninghead{N.~W\"urdemann et al.}{Taking Complete Finite Prefixes To High Level, Symbolically}

\begin{abstract}
  Unfoldings are a well known partial-order semantics of P/T Petri nets
  that can be applied to various model checking or verification problems.
  For \emph{high-level} Petri nets, the so-called \emph{symbolic} unfolding generalizes this notion.
  A complete finite prefix of a P/T Petri net's unfolding contains all information to verify, e.g., reachability of markings.
  We unite these two concepts and define complete finite prefixes of the symbolic unfolding of high-level Petri nets.
  For a class of safe high-level Petri nets,
  we generalize the well-known algorithm by Esparza et al.\
  for constructing small such prefixes.
   We evaluate this extended algorithm through a prototype implementation on four novel benchmark families.
  Additionally, we identify a more general class of nets with infinitely many reachable markings, for which an approach with an adapted cut-off criterion extends the complete prefix methodology, in the sense that  the original algorithm cannot be applied to the P/T net represented by a high-level net.
\end{abstract}

\begin{keywords}
Petri Nets, High-level Petri Nets, Unfoldings, Concurrency Theory
\end{keywords}

\section*{Introduction}

Petri nets \cite{Reisig13}, also called P/T (for Place/Transition) Petri nets or low-level Petri nets,
are a well-established formalism for describing distributed systems.
\emph{High-level Petri nets} \cite{Jensen96} (also called \emph{colored Petri nets}) are a concise representation of P/T Petri nets, allowing the places to carry tokens of different colors.
Every high-level Petri net represents a P/T Petri net, here called its \emph{expansion}\footnote{
	The Petri net being represented is commonly referred to as the \emph{unfolding} of the high-level Petri net in the literature. To prevent any potential confusion, we opt for the term \emph{expansion}, as, for instance, in \cite{ChatainF10}.},
where the process of constructing this P/T net is called \emph{expanding} the high-level net.

\emph{Unfoldings} of P/T Petri nets are introduced by Nielsen et al.\ in \cite{NielsenPW81}.
Engelfriet generalizes this concept in \cite{Engelfriet91}
by introducing the notion of \emph{branching processes},
and shows that the unfolding of a net is its maximal branching process.
0n~\cite{McMillan95}, McMillan gives an algorithm to compute a complete finite prefix of the unfolding of a given Petri net.
In a well-known paper \cite{EsparzaRV02}, Esparza, R\"omer, and Vogler improve this algorithm by defining and exploiting a total order on the set of configurations in the unfolding.
We call the improved algorithm the ``ERV-algorithm''.
It leads to a comparably small complete finite prefix of the unfolding.
In \cite{KhomenkoK03}, Khomenko and Koutny describe how to construct the unfolding of the expansion of a high-level Petri net without first expanding it.

High-level representations on the one hand and processes (resp. unfoldings) of P/T Petri nets on the other, at first glance seem to be conflicting concepts; one being a more concise, the other a more detailed description of the net('s behavior).
However, in \cite{EhrigHPBH02}, Ehrig et~al.\ define processes of high-level Petri nets, and in \cite{ChatainJ04},
Chatain and Jard define \emph{symbolic branching processes} and \emph{unfoldings} of high-level Petri nets.
The work on the latter is built upon in \cite{ChatainF10} by Chatain and Fabre, where they consider so-called ``puzzle nets''.
Based on the construction of a symbolic unfolding,
in \cite{ChatainJ06}, complete finite prefixes of safe time Petri nets are constructed,
using time constraints associated with timed processes.
In \cite{Chatain06}, using a simple example, Chatain argues that in general there exists no complete finite prefix of the symbolic unfolding of a high-level Petri net.
However, this is only true for high-level Petri nets with infinitely many reachable markings such that
the number of steps needed to reach them is unbounded,
in which case the same arguments work for P/T Petri nets.

In this paper,
we lift the concepts of complete prefixes and adequate orders to the level of symbolic unfoldings of high-level Petri nets.
We consider the class of \emph{safe} high-level Petri nets
(i.e., in all reachable markings, every place carries at most one token)
that have decidable guards and finitely many reachable markings.
This class generalizes safe P/T Petri nets, and
we obtain a generalized version of the ERV-algorithm creating a complete finite prefix of the symbolic unfolding of such a given high-level Petri net.
Our results are a generalization of~\cite{EsparzaRV02} in the sense that if a P/T~Petri net is viewed as a high-level Petri net,
the new definitions of adequate orders and completeness of prefixes on the symbolic level, as well as the algorithm producing them, all coincide with their P/T counterparts.

We proceed to identify an even more general class of
so-called \emph{symbolically compact} high-level Petri nets;
we drop the assumption of finitely many reachable markings, and instead
assume the existence of a bound on the number of steps needed to reach all reachable markings.
In such a case,
the expansion is possibly not finite, and the original ERV-algorithm from \cite{EsparzaRV02} therefore not applicable.
We adapt the generalized ERV-algorithm by weakening the cut-off criterion to ensure finiteness of the resulting prefix.
Still, in this cut-off criterion we have to compare infinite sets of markings.
We overcome this obstacle by symbolically representing these sets, using the decidability of the guards to decide cut-offs.
 Finally, we present four new benchmark families
for which we report on the results of applying a prototype implementation of the generalized ERV-algorithm.

\subsection*{Distinctions from the Conference Version}
This extended version incorporates numerous textual enhancements compared to our original work in~\cite{WurdemannCH23}.
Apart from that, we made the following changes and additions:
\begin{itemize}
	\item The proofs that were excluded in the conference version have now been integrated into the main body of the paper.
	\item We substituted the running example with a more intricate and compelling one (Fig.~\ref{fig:runex} in Sec.~\ref{sec:hlpns}),
	and discuss it in greater detail.
	Additionally, we present an example for the central concept ``color conflict''.
	\item Sec.~\ref{sec:expansions} has been completely revised.
	\item We introduced a new subsection, found in Sec.~\ref{sec:Nsc-vs-Nf}, where we demonstrate that the generalized ERV-algorithm may not terminate when applied to input nets from $ \Nsc $.
	This further motivates the work from the conference version of finding a new cut-off criterion (Sec.~\ref{sec:AlgoSymbComp}).
	In another new subsection, found in Sec.~\ref{sec:feasibility} we discuss the feasibility of symbolically compact nets
	and provide an outlook into the potential development of a symbolic reachability graph.
	\item We changed the definition of $ \predb $ in Sec.~\ref{sec:constraints} (formerly Section~4.3).
	This allows for a better presentation of Theorem~\ref{thm:constraintsDisj}, and an easier proof.
	\item In a new section, found in Sec.~\ref{sec:impl},
	we report in Sec.~\ref{sec:impl-det} on a new prototype implementation of the generalized ERV-algorithm from Sec.~\ref{sec:genERValgo}.
	In Sec.~\ref{sec:benchmarks} we present four new benchmark families of high-level Petri nets.
	In Sec.~\ref{sec:mode-det} we discuss a property of high-level Petri nets which
	we call \emph{mode-determinism}, leading to an indicator for whether the symbolic unfolding is expected to be faster to construct than the low-level unfolding.
	In Sec.~\ref{sec:exp}, we present the results of applying the implementation to the benchmarks from Sec.~\ref{sec:benchmarks}.
\end{itemize}

\section{High-level Petri nets \& symbolic unfoldings}\label{sec:hlpns}
In  \cite{ChatainJ04}, symbolic unfoldings for high-level Petri nets are introduced.
In Sections~\ref{subsec:hlpns} and \ref{subsec:symbolicBP}, we recall the definitions and formalism for high-level Petri nets and symbolic unfoldings from \cite{ChatainJ04}.

\emph{Multi-sets.}
For a set $ X $, we call a function $ A:X\to\N $ a \emph{multi-set over}~$ X $.
We denote $ x\in A $ if $ A(x)\geq 1 $.
For two multi-sets $ A,A' $ over the same set $ X $, we write $ A\leq A' $ iff $ \forall x\in X: A(x)\leq A'(x) $,
and denote by $ A+A' $ and $ A-A' $ the multi-sets over $ X $
given by $ (A+A')(x)=A(x)+A'(x) $ and $ (A-A')(x)=\max(A(x)-A'(x),0) $.
We use the notation $ \ms{\dots} $ as introduced in \cite{KhomenkoK03}:
elements in a multi-set can be listed explicitly as in $ \ms{x_1,x_1,x_2} $, which describes the multi-set $ A $ with $ A(x_1)=2 $, $ A(x_2)=1 $, and $ A(x)=0 $ for all $ x\in X\setminus\{x_1,x_2\} $.
A multi-set $ A $ is finite if there are finitely many
$ x\in X $ such that $ x\in A $. In such a case, $ \ms{f(x)\with x\in A} $, with $ f(x) $ being an
object constructed from $ x\in X $, denotes the multi-set $ A' $ such that
$A'=\sum_{x\in X} A(x)\cdot f(x) $, where the $ A(x)\cdot y $ is the multi-set containing exactly $ A(x) $ copies of $ y $.

\subsection{High-level Petri nets}\label{subsec:hlpns}

A \emph{(high-level) net structure} is a tuple
$ \hlns=\tup{\hltoks,\hlvars,\hlplaces,\hltranss,\hlflowfunc,\hlguard} $ with
the following components:
$ \hltoks $ and $ \hlvars $ are the sets of \emph{colors} and \emph{variables},
and
$ \hlplaces $~and~$ \hltranss $ are sets of \emph{places} and \emph{transitions}
such that the four sets are pairwise disjoint. The \emph{flow function} is given by $ \hlflowfunc:(\hlplaces\times\hlvars\times\hltranss)\cup(\hltranss\times\hlvars\times\hlplaces)\to \N $.
For $ \hltrans\in\hltranss $, let $ \hlvars(\hltrans)=\{\hlvar\in\hlvars\with\exists\hlplace\in\hlplaces: \tup{\hlplace,\hlvar,\hltrans}\in\hlflowfunc\lor \tup{\hltrans,\hlvar,\hlplace}\in\hlflowfunc\} $.
The function $ \hlguard $ maps each $ \hltrans\in\hltranss $ to a predicate $ \hlguard(\hltrans) $ on $ \hlvars(\hltrans)$, called the \emph{guard} of~$ \hltrans $.
By this, $ \hlguard(\hltrans) $
can contain other (bounded) variables,
but all free variables in $ \hlguard(\hltrans) $ must appear on arcs to or from $ \hltrans $.
A \emph{marking} in $ \hlns $ is a multi-set~$ \hlmarking$ over $ \hlplaces\times\hltoks $,
describing how often each color $ \hltok\in\hltoks $ currently resides on each place $ \hlplace\in\hlplaces $.
A \emph{high-level Petri net} $ \hlpn=\tup{\hlns,\hlmarkings_0} $ is a net structure $ \hlns $ together with a nonempty set $ \hlmarkings_0 $
of \emph{initial markings},
where we assume $ \forall \hlmarking_0,\hlmarking_0'\in\hlmarkings_0:
\ms{\hlplace\with\tup{\hlplace,\hltok}\in\hlmarking_0}=\ms{\hlplace\with\tup{\hlplace,\hltok}\in\hlmarking_0'} $,
i.e., in all initial markings, the same places are marked with \emph{the same number~of~colors}.
We often assume the two sets $ \hltoks $ of colors and $ \hlvars $ of variables to be fixed.
In this case, we denote a high-level net structure (resp.\ high-level Petri net) by
$ \hlns=\tup{\hlplaces,\hltranss,\hlflowfunc,\hlguard} $
(resp.\ $ \hlpn=\tup{\hlplaces,\hltranss,\hlflowfunc,\hlguard,\hlmarkings_0} $).

\medskip
For two nodes $ x, y\in\hlplaces\cup\hltranss $, we write $ x\rightarrow y $, if there exists a variable $ \hlvar $ such that $ \tup{x,\hlvar,y}\in\hlflowfunc $. The reflexive and irreflexive transitive closures of $ \rightarrow $ are denoted respectively by $ \leq $ and $ < $.
For a transition $ \hltrans\in\hltranss $,
we denote by $ \preset{\hltrans}=\mso \tup{\hlplace,\hlvar}\with\tup{\hlplace,\hlvar,\hltrans}\in\hlflowfunc \msc $
and $ \postset{\hltrans}=\mso \tup{\hlplace,\hlvar}\with\tup{\hltrans,\hlvar,\hlplace}\in\hlflowfunc \msc $
the \emph{preset} and \emph{postset} of $ \hltrans $.
A \emph{firing mode} of $ \hltrans $ is a mapping
$ \hlmode:\hlvars(\hltrans)\to\hltoks $ such that $ \hlguard(\hltrans) $ evaluates to $ \true $
under the substitution given by~$ \hlmode $, denoted by $ \hlguard(\hltrans)[\sigma]\equiv\true $.
We then denote
$ \preset{\hltrans,\hlmode}=\mso \tup{\hlplace,\hlmode(\hlvar)}\with\tup{\hlplace,\hlvar}\in\preset{\hltrans} \msc $
and $ \postset{\hltrans,\hlmode}=\mso \tup{\hlplace,\hlmode(\hlvar)}\with\tup{\hlplace,\hlvar}\in\postset{\hltrans} \msc $.
The set of modes of $ \hltrans $ is denoted by~$ \hlmodes(\hltrans) $.
Note that such a ``binding'' of variables to colors is always only local, when firing the respective transition.
$ \hltrans $ can \emph{fire} in such a mode $ \hlmode $ from a marking $ \hlmarking $
if $ \hlmarking\geq\preset{\hltrans,\hlmode} $, denoted by $ \hlmarking[\hltrans,\hlmode\rangle $.
This firing leads to a new marking $ \hlmarking'=(\hlmarking-\preset{\hltrans,\hlmode})+\postset{\hltrans,\hlmode} $,
which is denoted by $ \hlmarking[\hltrans,\hlmode\rangle\hlmarking' $.
We collect in the set $ \reachable(\hlns,\hlmarkings) $ the
markings reachable by firing a sequence of transitions in $ \hlns $ from any marking in a set of markings $ \hlmarkings $.
We say $ \hlns $ resp.\ $ \hlpn $ is \emph{finite} if $ \hlplaces$, $\hltranss $ and $ \hlflowfunc $ are finite.
In this paper, we in particular aim to analyze the \emph{behavior} of high-level Petri nets.
To avoid any issues concerning undecidability regarding the firing relation, we assume that guards are expressed in a decidable logic, with $ \hltoks $ as its domain of discourse.

\medskip
Let
$ \hlns=\tup{\hlplaces,\hltranss,\hlflowfunc,\hlguard} $ and
$ \hlns'=\tup{\hlplaces',\hltranss',\hlflowfunc',\hlguard'} $ be two net structures with the same sets of colors and variables.
A function $ \hlhomom:\hlplaces\cup\hltranss\to\hlplaces'\cup\hltranss' $
is called a \emph{(high-level Petri net) homomorphism},\index{Homomorphism!high-level} if:
\begin{enumerate}[label=\roman*)]
	\item it maps places and transitions in $ \hlns $ into the corresponding sets in $ \hlns' $, i.e.,\\
	$ {\hlhomom(\hlplaces)\subseteq\hlplaces'}$ and $ \hlhomom(\hltranss)\subseteq\hltranss' $;
	\item it is ``compatible'' with the guard, preset, and postset, of transitions, i.e.,\\
	for all
	$  {\hltrans\in\hltranss}$ we have $ \hlguard(\hltrans)=\hlguard'(\hlhomom(\hltrans))$ and $ \preset{\hlhomom(\hltrans)}=\mso\tup{\hlhomom(\hlplace),\hlvar}\with\tup{\hlplace,\hlvar}\in\preset{\hltrans} \msc
	$ and $\postset{\hlhomom(\hltrans)}=\mso\tup{\hlhomom(\hlplace),\hlvar}\with\tup{\hlplace,\hlvar}\in\postset{\hltrans} \msc
	$.
\end{enumerate}
For $ \hlpn=\tup{\hlns,\hlmarkings_0} $ and $ \hlpn'=\tup{\hlns',\hlmarkings_0'} $, the homomorphisms between $ \hlpn $ and~$ \hlpn' $ are the homomorphisms between $ \hlns $ and $ \hlns' $.
Such a homomorphism $ \hlhomom $ is called \emph{initial} if additionally
$ \{\ms{\tup{\hlhomom(\hlplace),\hltok}\with \tup{\hlplace,\hltok}\in\hlmarking_0}\with \hlmarking_0\in\hlmarkings_0 \}=\hlmarkings_0' $ holds,
i.e., the initial markings in $ \hlpn $ are mapped to the initial markings in $ \hlpn' $.

\eject
We define \emph{P/T Petri nets} as high-level Petri nets with singletons $ \hltoks=\{ \bullet \} $ and $ \hlvars=\{ \hlvar_\bullet \} $
for colors and variables, i.e., in a marking, every place holds a number of tokens~$ \bullet $, which is the only value ever assigned to the variable $ \hlvar_\bullet $ on every arc.
The guard of every transition in a P/T Petri net \linebreak is $ \true $.

\begin{example}\label{ex:hlnet}
	\begin{figure}[!htb]
		\centering
		\begin{subfigure}[t]{0.495\textwidth}\vspace{0pt}
			\centering
			\includegraphics[width=0.995\textwidth]{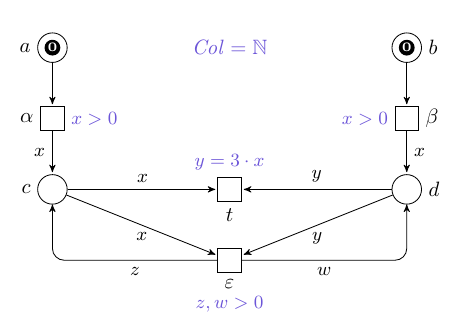}
			\caption{A safe high-level Petri net $ \hlpn $.\label{fig:runEx}}
		\end{subfigure}
		\begin{subfigure}[t]{0.495\textwidth}\vspace{0pt}
			\centering
			\includegraphics[width=0.995\textwidth]{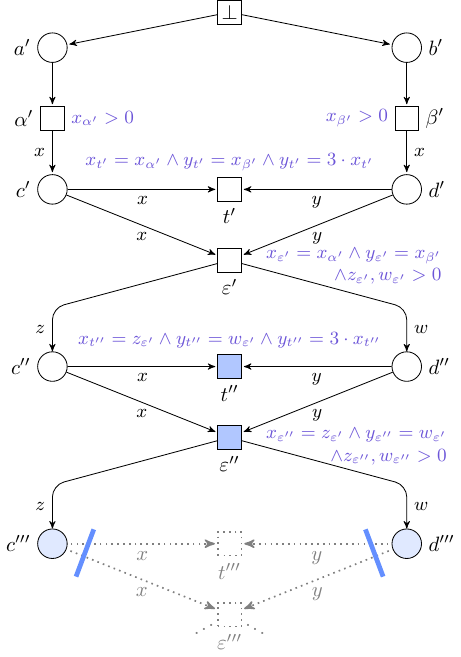}
			\caption{The symbolic unfolding $ \symbunf(\hlpn) $ of the net $ \hlpn $ in \subref{fig:runEx}.}\label{fig:runExUnf}
		\end{subfigure}
		\caption{A safe high-level Petri net $ \hlpn $ in \subref{fig:runEx}, discussed in Example~\ref{ex:hlnet}, and (a prefix of) the infinite symbolic unfolding $ \hlunf(\hlpn) $ in \subref{fig:runExUnf}, discussed in Example~\ref{ex:prefix}.}\label{fig:runex}
	\end{figure}
	Let $ \hltoks=\{ 0,\dots, m\} $ for a fixed $ m $, and $ \hlvars=\{x,y,z,w \} $ be the given sets of colors and variables.
	In Figure~\ref{fig:runEx}, the running example $ \hlpn $ of a high-level Petri net is depicted.
	Places are drawn as circles, and transitions as squares.
	The flow is described by labeled arrows,
	and the guards are written next to the respective transition.
	The set of initial markings is a singleton $ \hlmarkings_0=\{\hlmarking_0 \} $ with $ \hlmarking_0=\ms{ \tup{a,0},\tup{b,0}} $, which is depicted in the net.
	In all our examples we view $ 0 $ as a ``special'' color,
	in the sense that we employ unlabeled arcs as an abbreviation for
	arcs labeled with an additional variable $ x_0 $,
	and the guard of the respective transition having an additional term $ x_0=0 $ in its guard.
	Thus, we handle $ 0 $ as we would the token $ \bullet $ in the P/T case.

	From $ \hlmarking_0 $, both  $ \alpha $ and $ \beta $ can fire,
	taking the color $ 0 $ from place $ a $ resp.\ $ b $ and
	placing a color $ k\in\{1,\dots,m \} $ on place $ c $ resp.\ $ d $
	when firing in mode $ \{x\gets k\} $.
	The mode $ \{ x\gets 0 \} $ is for both transitions excluded by their respective guard.
	When both $ \alpha $ and $ \beta $ fire,
	the net arrives at a marking $ \ms{\tup{c,k},\tup{d,\ell}} $.
	From there, $ \varepsilon $ can fire arbitrarily often,
	always replacing the colors~$ k,\ell $ currently residing on $ c,d $ by any colors $ 0<k',\ell'\leq m $ by firing in mode $ \{ x\gets k,y\gets \ell,z\gets k',w\gets \ell' \} $.
	From every marking $ \ms{\tup{c,k},\tup{d,\ell}} $ satisfying $ \ell=3\cdot k $,
	transition $ t $ can fire, ending the execution.
\end{example}

\subsection{Symbolic branching processes and unfoldings}\label{subsec:symbolicBP}
A high-level net structure
$ \hlns=\tup{\hltoks,\hlvars,\hlplaces,\hltranss,\hlflowfunc,\hlguard} $
is called \emph{ordinary} if there is at most one arc connecting any two nodes in $ \hlns $,
i.e., $ \forall \tup{x,y}\in(\hlplaces\times\hltranss)\cup(\hltranss\times\hlplaces): \sum_{\hlvar\in\hlvars}\hlflowfunc(x,\hlvar,y)\leq 1 $.
For such an ordinary net structure,
analogously to the well-known low-level case, two nodes $ x,y\in\hlplaces\cup\hltranss $ are in \emph{structural conflict}, denoted by $ x\sharp y $, if $ \exists\hlplace\in\hlplaces\,\exists\hltrans,\hltrans'\in\hltranss: {\hltrans\neq\hltrans'}\land
\hlplace\rightarrow\hltrans\land\hlplace\rightarrow\hltrans'\land \hltrans\leq x\land\hltrans'\leq y
$.

\medskip
A \emph{high-level occurrence net} (defined below) is a high-level Petri net with an ordinary net structure that satisfies certain properties.
In such a net, we call the places \emph{conditions} and denote them by $ \hlconds $.
Transition are called \emph{events} (denoted by $ \hlevnts $),
and reachable markings are called \emph{cuts}, where the set of initial cuts is denoted by $ \hlcuts_0 $.
The flow relation is denoted by $ \hlarcs $.

The properties \mbox{\textit{i) -- iii)}} in the definition below are exactly the same as in the low-level case and concern solely the net structure.
Property \textit{iv)} generalizes the corresponding requirement of low-level occurrence nets to the current situation, in which, just as in the low-level case, every condition has at most one event in its preset, and that those conditions having an empty preset constitute the initial cut.
Case \textit{iv.a)} describes the conditions that initially hold a color, at the ``top'' of the net.
Case \textit{iv.b)} on the other hand describes the conditions ``deeper'' in the net,
which initially do not hold a color.

\begin{definition}[High-level occurrence net \cite{ChatainJ04}]
	A \emph{high-level occurrence net} is a  high-level Petri net \label{def:page:hl-onet} $ \hlonet=\tup{\hltoks,\hlvars,\hlconds,\hlevnts,\hlarcs,\hlguard,\hlcuts_0} $ with an ordinary net structure
	$ \tup{\hltoks,\hlvars, \hlconds,\hlevnts,\hlarcs,\hlguard} $
	such that
	\begin{enumerate}[label=\textit{\roman*)}]
		\item No event is in structural self-conflict, i.e.,
		$ \forall \hlevnt\in\hlevnts: \neg(\hlevnt\sharp \hlevnt) $.
		\item No node is its own causal predecessor, i.e.,
		$ \forall x\in\hlconds\cup\hlevnts:{\neg (x<x)} $.
		\item The relation $ < $ is well-founded, i.e., contains no infinite decreasing sequence.
		\item For every $ \hlcond\in\hlconds $, exactly one of the following holds:
		\begin{enumerate}[label=\textit{\alph*)}]
			\item $\forall \hlcut_0\in\hlcuts_0: \sum_{\hltok\in\hltoks}\hlcut_0(\hlcond,\hltok)=1 $ and $ \{ \hlevnt\with\hlevnt\rightarrow \hlcond \}=\emptyset $.\\
			In this case we denote $ \preset{\hlcond}=\tup{\bot,\hlvar^\hlcond} $.
			\item $\forall \hlcut_0\in\hlcuts_0: \sum_{\hltok\in\hltoks}\hlcut_0(\hlcond,\hltok)=0 $ and there exists a unique pair $ \tup{\hlevnt,\hlvar} $ s.t.\ $ \tup{\hlevnt,\hlvar,\hlcond}\in\hlarcs $.
			In this case we denote $ \preset{\hlcond}=\tup{\hlevnt,\hlvar} $
		\end{enumerate}
	\end{enumerate}
\end{definition}

We denote by $ \hlconds_0=\{ \hlcond\in\hlconds\with \exists \hlcut_0\in\hlcuts_0\,\exists \hltok\in\hltoks: \tup{\hlcond,\hltok}\in\hlcut_0 \} $ the conditions from~\textit{iv.a)} occupied in all initial cuts.
$ \bot $ can be seen as a ``special event'' that fires only once to initialize the net, and produces the initial cuts $\hlcut_0\in \hlcuts_0 $ by assigning values to the variables $ \hlvar^\hlcond $ on ``special arcs'' $ \tup{\bot,\hlvar^\hlcond,\hlcond} $ towards the conditions $ \hlcond\in\hlconds_0 $.

In a crucial notation for this paper, we define
in case~\textit{iv.a)} $ \hlcondevnt(\hlcond)=\bot $, and $ \hlcondvar(\hlcond)=\hlvar^\hlcond $,
and
in case~\textit{iv.b)} we identify the event~$ \hlevnt $ by $ \hlcondevnt(\hlcond) $ and the variable $ \hlvar $ by~$ \hlcondvar(\hlcond) $.
By this notation, $\forall\hlcond\in\hlconds: \preset{\hlcond}=\tup{\hlcondevnt(\hlcond),\hlcondvar(\hlcond)} $. We can say that whenever a condition $ \hlcond $ holds a color $ \hltok $,
then it got placed there by firing $ \hlcondevnt(\hlcond) $ in a mode that binds $ \hlcondvar(\hlcond) $
to the color $ \hltok $.

In a high-level occurrence net, we define for every event $ \hlevnt $ the predicates $ \locpred(\hlevnt) $ and $ \pred(\hlevnt) $.
The predicate $ \pred(\hlevnt) $ is satisfiable iff $ \hlevnt $ is not dead,
i.e., there are cuts $ \hlcut_0,\dots,\hlcut_n $ with $ \hlcut_0\in\hlcuts_0 $
and events $ \hlevnt_1,\dots,\hlevnt_n $, s.t.
$ \hlcut_0[\hlevnt_1\rangle\dots[\hlevnt_n\rangle\hlcut_n[\hlevnt\rangle $.
This predicate is obtained by building a conjunction over all
\emph{local predicates} of events $ \hlevnt' $ with $ \hlevnt'\leq\hlevnt $, and
the predicate of the special event~$ \bot $.

\medskip
The local predicate of $ \hlevnt $  is, in its turn, a conjunction of two predicates expressing that (i) the guard of the event $ \hlevnt $ is satisfied, and (ii) that for any $ \tup{\hlcond,\hlvar}\in\preset{\hlevnt} $, the value of the variable $ \hlvar $ coincides with the color that the event~$ \hlcondevnt(\hlcond) $ placed in $ \hlcond $.
To realize this, the variables $ \hlvar\in\hlvars(\hlevnt) $ are instantiated by the index $ \hlevnt $, so that $ \hlvar_\hlevnt $ describes the value assigned to $ \hlvar $ by a mode of $ \hlevnt $.
Having the definition of $ \hlcondevnt(\hlcond) $ and $ \hlcondvar(\hlcond) $ from above in mind, for a condition $ \hlcond $, we abbreviate $ \hlcondvar_{\hlcondevnt}(\hlcond)=\hlcondvar(\hlcond)_{\hlcondevnt(\hlcond)} $.
Formally, we have
\begin{align*}
	\locpred(\hlevnt)\quad&=\quad
	\hlguard(\hlevnt)[ \hlvar\gets\hlvar_\hlevnt ]_{\hlvar\in\hlvars(\hlevnt)}
	\quad\land\quad
	\bigwedge_{\tup{\hlcond,\hlvar}\in\preset{\hlevnt}}\hlvar_\hlevnt=\hlcondvar_\hlcondevnt(\hlcond)
	\\
	\pred(\hlevnt)\quad&=\quad\pred(\bot)\land\bigwedge_{\hlevnt'\leq\hlevnt}\locpred(\hlevnt'),
\end{align*}
where $ \pred(\bot)=\bigvee_{\hlcut_0\in\hlcuts_0}\bigwedge_{\tup{\hlcond,\hltok}\in\hlcut_0}(\hlvar_\bot^\hlcond=\hltok) $
symbolically represents the set of initial cuts.

\medskip
A \emph{symbolic branching process} of a high-level Petri net $ \hlpn
$
is a pair $ \beta=\tup{\hlonet,\hlhomom} $
with an occurrence net $ \hlonet=\tup{\hltoks,\hlvars,\hlconds,\hlevnts,\hlarcs,\hlguard,\hlcuts_0} $ in which $ \pred(\hlevnt) $ is satisfiable for all $\hlevnt\in\hlevnts $,
and an initial homomorphism $ \hlhomom:\hlonet\to\hlpn $ that is injective on events with the same preset, i.e.,
$ \forall \hlevnt,\hlevnt'\in\hlevnts:
(\preset{\hlevnt}=\preset{\hlevnt'}\land \hlhomom(\hlevnt)=\hlhomom(\hlevnt') )
\Rightarrow \hlevnt=\hlevnt' $.

For two symbolic branching processes $ \beta=\tup{\hlonet,\hlhomom} $ and $ \beta'=\tup{\hlonet',\hlhomom'} $  of a high-level Petri net, $ \beta $ is a \emph{prefix} of $ \beta' $ if there
exists an injective initial homomorphism~$ \phi $ from $ \hlonet $ into $ \hlonet' $, such that $ \hlhomom'\circ\phi = \hlhomom $.
In \cite{ChatainJ04} it is stated that for any given high-level Petri net~$ \hlpn $ there exists a unique maximal branching process (maximal w.r.t.\ the prefix relation and unique up to isomorphism).
This branching process is called the \emph{symbolic unfolding}, and denoted by $ \symbunf(\hlpn)=\tup{\hlunf(\hlpn),\symbhomom^\hlpn}
$.
The value $ \symbhomom^\hlpn(x) $ is called the \emph{label} of a node $ x $ in $ \hlunf(\hlpn) $.

\begin{example}\label{ex:prefix}
	Consider again the high-level Petri net $ \hlpn $ from Figure~\ref{fig:runEx}.
	In Figure~\ref{fig:runExUnf} we see (a prefix of) the infinite occurrence net $ \hlunf(\hlpn) $ of the symbolic unfolding $ \symbunf(\hlpn) $.
	We depict the prefix with two instances of each $ t $ and $ \varepsilon $.
	Each node in the unfolding is named after the
	represented place resp.\ transition (i.e., its label),
	equipped with a superscript.
	We include the ``special event'' $ \bot $, that can only fire once, in the drawing.
	The guards of events are omitted, since they have the same guards as their label.
	Instead, the local predicate of each event is written next to it.

	The local predicate of $ \alpha' $, namely $ x_{\alpha'}>0 $
	expresses that the assignment of colors to variables by a mode of $ \alpha' $
	must satisfy the constraint given by the guard of its label $ \alpha $.
	Analogously for $ \beta' $.
	The same is expressed in the local predicate of $ t' $ by $ y_{t'}=3\cdot x_{t'} $,
	coming from the guard $ y=3\cdot x $ of $\symbhomom^\hlpn(t')= t $.
	Additionally, the first part of the conjunction formalizes that,
	since $ \tup{c',x}\in\preset{t'} $,
	the value that a mode of $ t' $ assigns to $ x $
	must be the same that a mode of $ \hlcondevnt(c')=\alpha' $ assigned to $ \hlcondvar(c')=x $.
	This is expressed as $ x_{t'}=x_{\alpha'} $.
	The second part of the conjunction formalizes the same for $ y $ and $ d' $.
	The whole predicate of $ t' $ is then given by
	\begin{equation*}
		\pred(t')=x_{\alpha'}>0\land x_{\beta'}>0
		\land x_{t'}=x_{\alpha'} \land y_{t'}=x_{\beta'}
		\land y_{t'}=3\cdot x_{t'}.
	\end{equation*}
	Since it is satisfiable for example by
	$ \{ x_{\alpha'}\gets 1, x_{\beta'}\gets 3, x_{t'}\gets 1, y_{t'}\gets 3 \} $
	(meaning that $ t $ can fire in mode $ \{ x\gets 1,y\gets 3 \} $ after $ \alpha $ fired in mode $ \{ x\gets 1 \} $ and $ \beta $ fired in mode $ \{ x\gets 3 \} $),
	the node $ t' $ is not dead and an event in the unfolding.

	The blue shading of event $ \varepsilon'' $ and $ t'' $ indicates that they are what we later term \emph{cut-off events},
	which leads to the \emph{complete finite prefix} being marked by the blue thick lines being obtained by Alg.~\ref{alg:complpref}, as described later.
	The unfolding itself is infinite.
\end{example}

As we see in the definition of high-level occurrence nets,
the notion of causality and structural conflict are the same as in the low-level case.
However, a set of events in an occurrence net can also be in what we call \emph{color conflict},
meaning that the conjunction of their predicates is not satisfiable.
In a symbolic branching process,
this means that the constraints  on the values of the ﬁring modes, coming from the guards of the transitions,
prevent joint occurrence of all events from such a set in any \emph{one} run of the net:

\medskip
The nodes in a set $ X \subseteq \hlevnts\cup\hlconds $ are in \emph{color conﬂict} if
\begin{equation*}
	\bigwedge_{\hlevnt\in X\cap\hlevnts}\pred(\hlevnt)\land \bigwedge_{\hlcond\in X\cap\hlconds}\pred(\hlcondevnt(\hlcond))
\end{equation*}
is
\emph{not} satisfiable.
The nodes of $ X $ are \emph{concurrent} if they are \emph{not} in color conflict, and for each $ x,x'\in X' $, neither $ x<x' $ , nor $ x'<x $, nor $ x\sharp x' $ holds.
A set of concurrent conditions is called a \emph{co-set}.

\medskip
Note that while a set  of nodes is defined to be in structural conflict
if and only if two nodes in it are in structural conflict, the same does not hold for color conflict:
it is possible to have a set $ \{ x_1,x_2,x_3 \} $ of nodes
that are in color conflict,
but for which every subset of cardinality 2 is \emph{not} in color conflict.
We demonstrate this on an example.

\begin{example}[Color conflict]
	\begin{figure}[!htb]
		\begin{subfigure}{0.49\linewidth}
			\centering
			\includegraphics[scale=0.9]{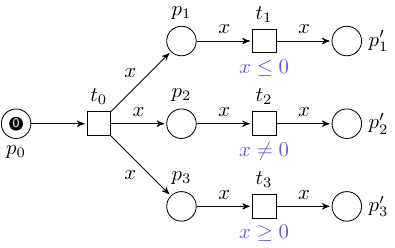}
			\caption{A high-level Petri net with $ \hltoks=\mathbb{Z} $ and $ \hlvars=\{x\} $.\label{fig:colorconflictNet}}
		\end{subfigure}
		\hfill
		\begin{subfigure}{0.49\linewidth}
			\centering
			\includegraphics[scale=0.9]{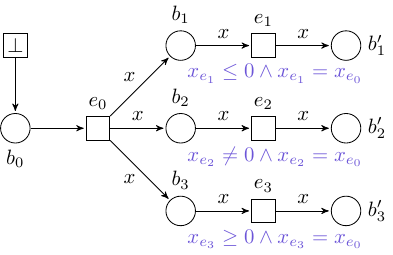}
			\caption{The symbolic unfolding of the net in \subref{fig:colorconflictNet},
				with the events $ \{ \hlevnt_1,\hlevnt_2,\hlevnt_3 \} $ in color conflict.\label{fig:colorconflictUnf}}
		\end{subfigure}
		\caption{Illustration and example of nodes in color conflict.\label{fig:colorconflict}}
	\end{figure}
	In Figure~\ref{fig:colorconflictNet}, a high-level Petri net
	with initial marking $ \ms{\tup{\hlplace_0,0}} $ is depicted.
	The only enabled transition is $ \hltrans_0 $, placing
	the same color $ \hlmode(x)\in\mathbb{Z}=\{ \dots,-1,0,1,2,\dots \} $ on
	each of the three places $ \hlplace_1,\hlplace_2,\hlplace_3 $
	when fired in mode \mbox{$ \hlmode=\{x\gets \hlmode(x)\} $}.
	From each of these places, the color $ \hlmode(x) $ may be taken by a respective transition.
	The three transitions $ \hltrans_1,\hltrans_2,\hltrans_3 $, however,
	each have a guard:
	$ \hlguard(\hltrans_1)=(x\leq 0) $,
	$ \hlguard(\hltrans_2)=(x\neq 0) $, and
	$ \hlguard(\hltrans_3)=(x\geq 0) $.
	Depending on the mode $ \hlmode $ in which $ \hltrans_0 $ fired,
	always two of the three transitions are fireable:
	if $ \hlmode(x)=0 $ then $ \hltrans_1 $ and $ \hltrans_3 $ can both fire (in mode $ \{x\gets 0\} $),
	if $ \hlmode(x)< 0 $ then $ \hltrans_1 $ and $ \hltrans_2 $ can fire,
	and if $ \hlmode(x)> 0 $ then $ \hltrans_2 $ and $ \hltrans_3 $ can fire.

\medskip
	Since the high-level Petri net in Figure~\ref{fig:colorconflictNet} is a high-level occurrence net and all predicates are satisfiable,
	it has the same structure as
	its own symbolic unfolding in Figure~\ref{fig:colorconflictUnf}.
	The set $ \{ \hlcond_1,\hlcond_2,\hlcond_3 \} $ is a co-set,
	since the conditions are neither in conflict, nor causally related,
	and $ \bigwedge_{\hlcond\in\{ \hlcond_1,\hlcond_2,\hlcond_3 \}} \pred(\hlcondevnt(\hlcond))$, which is equivalent to $ \true $,
	is satisfiable, i.e., the conditions are \emph{not} in color conflict.
	Consequently, the set $ \{ \hlevnt_1,\hlevnt_2,\hlevnt_3 \} $ is also not in structural conflict,
	and the events are not causally related.
	However, there now \emph{is} a color conflict between these three events,
	since $ \bigwedge_{\hlevnt\in\{ \hlevnt_1,\hlevnt_2,\hlevnt_3 \}} \pred(\hlevnt) $
	implies $ x_{\hlevnt_0}\leq 0\land x_{\hlevnt_0}\neq0\land x_{\hlevnt_0}\geq 0 $,
	which obviously is not satisfiable.
	In contrast, each of the sets $ \{\hlevnt_i,\hlevnt_j\} $ with $ i,j\in\{1,2,3 \} $, $ i\neq j $ is \emph{not} in color conflict.
	This makes each of the sets $ \{\hlcond_i',\hlcond_j'\} $ a co-set,
	while $ \{ \hlcond_1',\hlcond_2',\hlcond_3' \} $ is not a co-set.
\end{example}

Having employed the notions of conflict, we come to one of the most important definitions when dealing with unfoldings, namely \emph{configurations}.
\begin{definition}[Configuration \cite{ChatainJ04}]
	A \emph{(symbolic) configuration} is a set of high-level events that is free of structural conflict and color conflict, and causally closed.
	The configurations in a symbolic branching process $ \beta $ are collected in~the~set~$ \hlconfs(\beta) $.
\end{definition}

Recall that $ \hlconds_0 $ are the initial conditions occupied in all initial cuts.
For a configuration $ \hlconf $, we define by
$ \confcut(\hlconf):= (\hlconds_0\cup (\hlconf\rightarrow))\setminus (\rightarrow\hlconf)
$
the high-level conditions that are occupied after any execution of $ \hlconf $.
Note that
$ \confcut(\hlconf) $ is a co-set,
and that $ \emptyset $ is a configuration with $ \confcut(\emptyset)=\hlconds_0 $.

Let $ \hlevnt\in\hlevnts $ be a high-level event.
We define the so-called \emph{cone configuration}
$ \cone{\hlevnt}:=\{ \hlevnt'\in\hlevnts\with \hlevnt'\leq \hlevnt \} $.
Additionally, we define the sets $ \hlvars_\hlevnt:=\{ \hlvar_\hlevnt\with \hlvar\in\hlvars(\hlevnt) \} $
and $ \hlvars_\bot:=\{ \hlvar_\bot^\hlcond\with \hlcond\in\hlconds_0 \} $ of indexed variables, and for a set $ \hlevnts'\subseteq\hlevnts\cup\{\bot\} $
we denote $ \hlvars_{\hlevnts'}:= \bigcup_{\hlevnt\in\hlevnts'}\hlvars_\hlevnt $.
Note that, for every event $ \hlevnt $, $ \pred(\hlevnt) $ is a predicate over the variables $ \hlvars_{\cone{\hlevnt}\cup\{\bot\}} $.

\subsection{Properties of the symbolic unfolding.}\label{sec:SymbUnfProps}
Having recalled the definitions and formal language from \cite{ChatainJ04}, we now delve into the novel aspects of this paper.
We state three analogues of well-known properties of the Unfolding of P/T Petri nets for the symbolic unfolding of high-level nets.
These properties are:
\begin{enumerate}[label=(\roman*)]
	\item The cuts in the unfolding represent precisely the reachable markings in the net.
	\item For every transition that can occur in the net, there is an event in the unfolding with corresponding label (and vice versa).
	\item The unfolding is complete in the sense that for any configuration,
	the part of the unfolding that ``lies after'' that configuration
	is the unfolding of the original net with the initial markings being the ones represented by the configurations cut.
\end{enumerate}
The properties are stated in Prop.~\ref{prop:MarkingsRepresented}, Prop.~\ref{prop:unfComplete}, and Prop.~\ref{prop:futureIso}, respectively.

To express these properties, we introduce the notion of \emph{instantiations} of configurations~$ \hlconf $,
choosing a mode for every event in $ \hlconf $ without creating color conflicts.
This is realized by  assigning to each variable $ \hlvar_\hlevnt\in\hlvars_{\hlconf\cup\{\bot\}} $ a value in $ \hltoks $,
such that the above defined predicates evaluate to $ \true $.
For each $ \hlevnt\in\hlconf $,
the assignment of values to the indexed variables in $ \hlvars_\hlevnt $ corresponds to a mode~of~$ \hlevnt $.
\begin{definition}[Instantiation of Configuration]\label{def:inst-of-conf}
	For a given configuration $ \hlconf $,
	an \emph{instantiation of~$ \hlconf $} is a function
	$ \insta:\hlvars_{\hlconf\cup\{\bot\}}\to\hltoks $,
	such that $ \forall \hlevnt\in\hlconf\cup\{\bot\}:\pred({\hlevnt})[\insta]\equiv\true $,
	i.e., it satisfies all predicates in the configuration.
	The set of instantiations of a given configurations $ \hlconf $ is denoted by $ \instas(\hlconf) $.
\end{definition}
Note that, by definition, every configuration $ \hlconf $ has an instantiation $ \insta $.
We denote by
$ \confcut(\hlconf,\insta):=
\{ \tup{\hlcond,\hltok}\with \hlcond\in\confcut(\hlconf)\land \insta(\hlcondvar_\hlcondevnt(\hlcond))=\hltok \}\subseteq\hlconds\times\hltoks  $
the \emph{cut} of an ``instantiated configuration'', and by
$ \confmark(\hlconf,\insta):=\mso \tup{\hlhomom(\hlcond),\hltok}\with \tup{\hlcond,\hltok}\in	\confcut(\hlconf,\insta) \msc $
its \emph{marking}.
We collect both of these in
$ \confcuts(\hlconf):=\{\confcut(\hlconf,\insta)\with\insta\in\instas(\hlconf) \} $
and
$ \confmarks(\hlconf):=\{\confmark(\hlconf,\insta)\with\insta\in\instas(\hlconf) \} $.
Note that in this notation,
for the empty configuration we have
$ \confcuts(\emptyset)=\hlcuts_0 $ and
$ \hlmarkings(\emptyset)=\hlmarkings_0 $.

\begin{proposition}\label{prop:MarkingsRepresented}
	Let $ \hlpn $ be a high-level Petri net
	and $ \symbunf $ its symbolic unfolding.
	Then
	$ \reachable(\hlpn) =
	\{ \confmark(\hlconf,\insta)\with  \hlconf\in\hlconfs(\symbunf), \insta\in\instas(\hlconf) \} $.
\end{proposition}
\begin{proof}
	The proof is an easy induction over the number $ n $ of transitions/events needed to reach a respective marking/cut.
	The induction anchor $ n=0 $ is proved by using that $ \symbhomom $ is an initial homomorphism which gives
	$ \hlmarkings_0
	=\{ \ms{ \tup{\symbhomom(\hlcond),\hltok} \with \tup{\hlcond,\hltok}\in \hlcut_0  }\with\hlcut_0\in\hlcuts_0 \}
	=\{ \ms{ \tup{\symbhomom(\hlcond),\hltok} \with \tup{\hlcond,\hltok}\in \hlcut  }\with\hlcut\in\confcuts(\emptyset) \}
	=\{ \confmark(\emptyset.\insta)\with \insta\in\instas(\emptyset) \}
	$.
	The induction step is realized by Prop.~\ref{prop:unfComplete}.
\end{proof}
\begin{proposition}\label{prop:unfComplete}
	The symbolic unfolding $ \symbunf=\tup{\hlunf,\symbhomom} $ with events $ \hlevnts $ of a high-level Petri net $ \hlpn=\tup{\hlplaces,\hltranss,\hlflowfunc,\hlguard,\hlmarkings_0} $ satisfies
	$ \forall \hlconf\in\hlconfs(\symbunf)\;
	\forall\insta\in\instas(\hlconf)\;
	\forall \hltrans\in\hltranss\;
	\forall\hlmode\in\hlmodes(\hltrans):$
	\begin{equation*}	\confmark(\hlconf,\insta)[\hltrans,\hlmode\rangle
		\ \Leftrightarrow\
		\exists \hlevnt\in\hlevnts: \symbhomom(\hlevnt)=\hltrans \land \confcut(\hlconf,\insta)[\hlevnt,\hlmode\rangle.
	\end{equation*}
\end{proposition}
\begin{proof}
	Let $ \hlunf=\tup{\hlconds,\hlevnts,\hlarcs,\hlguard,\hlcuts_0}$, and let $
	\hlconf\in\hlconfs(\symbunf),
	\insta\in\instas(\hlconf),
	\hltrans\in\hltranss,
	{	\hlmode\in\hlmodes(\hltrans)}$.

	Let
	$ \confmark(\hlconf,\insta)[\hltrans,\hlmode	\rangle $,
	which means
	\begin{equation*}
		\preset{\hltrans,\hlmode}
		\leq \confmark(\hlconf,\insta)
		=\ms{ \tup{\symbhomom(\hlcond),\insta(\hlcondvar_{\hlcondevnt}(\hlcond))} \with \tup{\hlcond,\insta(\hlcondvar_{\hlcondevnt}(\hlcond))}\in\confcut(\hlconf,\insta) } ,
	\end{equation*}
	Let $ \hlconds'\subseteq\confcut(\hlconf) $ be a set of conditions s.t.
	\begin{equation*}
		\preset{\hltrans,\hlmode}=\ms{\tup{\symbhomom(\hlcond),\insta(\hlcondvar_\hlcondevnt(\hlcond))}\with \hlcond\in\hlconds'}.
	\end{equation*}
	Aiming a contradiction, assume there is \emph{no} $ \hlevnt\in\hlevnts$ s.t.\ $\symbhomom(\hlevnt)=\hltrans$ and $ \confcut(\hlconf,\insta)[\hlevnt,\hlmode\rangle $:
	we extend $ \symbunf $ by such an event.	We add to $ \hlevnts $ an event $ \widetilde{\hlevnt} $ with $ \symbhomom(\widetilde{\hlevnt})=\hltrans $ and $ \hlguard(\widetilde{\hlevnt})=\hlguard(\hltrans) $.
	Choose for every $ \hlcond\in\hlconds' $ a variable $ \hlvar^\hlcond\in\hlvars $ s.t.\
	\begin{equation*}
		\ms{(\symbhomom(\hlcond),\hlvar^\hlcond)\with \hlcond\in\hlconds'}=
		\ms{(\hlplace,\hlvar)\with \tup{\hlplace,\hlvar}\in\preset{\hltrans}}\quad(=\preset{\hltrans}).
	\end{equation*}
	We define $ \preset{\widetilde{\hlevnt}}=\ms{\tup{\hlcond,\hlvar^\hlcond}\with \hlcond\in\hlconds'} $.
	Then we have $ \ms{ \tup{\symbhomom(\hlcond),\hlvar}\with \tup{\hlcond,\hlvar}\in\preset{\widetilde{\hlevnt}} }=\preset{\hltrans}=\preset{\symbhomom(\widetilde{\hlevnt})} $.
	For every $ \tup{\hlplace,\hlvar}\in\postset{\hltrans} $,
	we then add
	$ \postset{\hltrans}(\hlplace,\hlvar) $ conditions $ \hlcond $ with $ \symbhomom(\hlcond)=\hlplace $ to $ \hlconds $ and add $ \tup{\widetilde{\hlevnt},\hlvar,\hlcond} $ to $ \hlarcs $.
	We thus get $ \postset{\symbhomom(\widetilde{\hlevnt})}=\ms{ \tup{\symbhomom(\hlcond),\hlvar}\with \tup{\hlcond,\hlvar}\in\postset{\widetilde{\hlevnt}} } $.
	We now created a symbolic branching process bigger than $ \symbunf $, contradicting that $ \symbunf $ is the symbolic unfolding.

\medskip
	Conversely, assume
	$ \exists \hlevnt\in\hlevnts: \symbhomom(\hlevnt)=\hltrans \land \confcut(\hlconf,\insta)[\hlevnt,\hlmode\rangle. $
	Then $ \preset{\hlevnt,\hlmode}\leq \confcut(\hlconf,\insta) $,
	and therefore,
	$ \preset{\hltrans}
	=\ms{ \tup{\symbhomom(\hlcond),\hlvar}\with \tup{\hlcond,\hlvar}\in\preset{\hlevnt} }
	\leq \ms{ \tup{\symbhomom(\hlcond),\hlvar}\with \tup{\hlcond,\hlvar}\in\confcut(\hlconf,\insta) }
	=\confmark(\hlconf,\insta) $,
	meaning
	$ \confmark(\hlconf,\insta)[\hltrans,\hlmode	\rangle $.
\end{proof}

Given a configuration $ \hlconf $ of a symbolic branching process $ \beta=\tup{\hlonet,\hlhomom} $, we define $ {\Uparrow}\hlconf$ as the pair $ \tup{\hlonet',\hlhomom'} $, where $ \hlonet' $ is the unique subnet of $ \hlonet $ whose set of nodes is $\{ x\in\hlconds\cup\hlevnts\with x\notin (\hlconf\cup\rightarrow\hlconf)\land \forall y\in\hlconf:\lnot(y\sharp x) \land( \hlconf\cup\{ x \}\text{ is not in color conflict})  \} $
with the set~$ \confcuts(\hlconf) $ of initial cuts,
and $ \hlhomom' $ is the restriction of $ \hlhomom $ to the nodes of $ \hlonet' $.
The branching process $ {\Uparrow}\hlconf $ is referred to as the \emph{future} of $ \hlconf $.

\begin{proposition}\label{prop:futureIso}
	If $ \beta $ is a symbolic branching process of $ \tup{\hlns,\hlmarkings_0} $ and $ \hlconf $ is a configuration of $ \beta $, then $ {\Uparrow}\hlconf $ is a branching process of $ \tup{\hlns,\confmarks(\hlconf)} $.
	Moreover, if $ \beta $ is the unfolding of $ \tup{\hlns,\hlmarkings_0} $, then $ {\Uparrow}\hlconf $ is the unfolding of $ \tup{\hlns,\confmarks(\hlconf)} $.
\end{proposition}
\begin{proof}
	Let $ {\Uparrow}\hlconf=\tup{\hlonet',\hlhomom'} $ with $ \hlonet'=\tup{\hlconds',\hlevnts',\hlflowfunc',\hlguard',\confcuts(\hlconf)} $.
	To show that $ \hlonet' $ is an occurrence net, we have to show {\it i -- iv} from the Definition~\ref{def:page:hl-onet}.
	{\it i -- iii} are purely structural properties and follow from the fact that $ \hlonet $ is an occurrence net.
	{\it iv} is satisfied since
	$ \forall \hlcond\in\confcut(\hlconf)\; \forall \hlcut\in\hlcuts(\hlconf): \sum_{\hltok\in\hltoks} \hlcut(\hlcond,\hltok)=1 $
	and $ \forall \hlcond\in\hlconds'\setminus\confcut(\hlconf)\; \forall \hlcut\in\hlcuts(\hlconf): \sum_{\hltok\in\hltoks} \hlcut(\hlcond,\hltok)=0 $.
	$ \hlhomom' $ is a homomorphism that is injective on events with the same preset since $ \hlhomom $ is,
	and that $ \hlhomom' $ is initial follows by Prop.~\ref{prop:MarkingsRepresented} and Prop.~\ref{prop:unfComplete}.

	When $ \beta $ is the symbolic unfolding of $ \tup{\hlns,\hlmarkings_0} $,
	then the maximality of $ {{\Uparrow}\hlconf} $ follows from the maximality of $ \beta $,
	making $ {{\Uparrow}\hlconf} $ the symbolic unfolding of $ \tup{\hlns,\hlmarkings(\hlconf)} $.
\end{proof}

\section{Finite \& complete prefixes of symbolic unfoldings}\label{sec:complPref}
We combine ideas from \cite{EsparzaRV02} (computing small finite and complete prefixes of unfoldings) with  results from \cite{ChatainJ04} (symbolic unfoldings of high-level Petri nets) to define and construct
complete finite prefixes of symbolic unfoldings of high-level Petri nets.
We generalize the concepts and the ERV-algorithm from \cite{EsparzaRV02} for safe P/T Petri nets to a class of safe high-level Petri nets, and compare this generalization to the original.
We will see that for P/T nets interpreted as high-level nets, all generalized concepts
(i.e., complete prefixes, adequate orders, cut-off events),
and, as a consequence, the result of the generalized ERV-algorithm, all coincide with their P/T counterparts.

We start by lifting the definition of completeness
to the level of symbolic unfoldings.
Together with Prop.~\ref{prop:MarkingsRepresented} and Prop.~\ref{prop:unfComplete}, this can be seen as a direct translation from the low-level case described, e.g., in \cite{EsparzaRV02}.
\begin{definition}[Complete prefix]\label{def:ComplPref}
	Let $ \beta=\tup{\hlonet,\hlhomom} $ be a prefix of the symbolic unfolding of a high-level Petri net~$ \hlpn $, with events $ \hlevnts' $.
	Then $ \beta $ is called \emph{complete} if for every reachable marking $ \hlmarking $ in $ \hlpn $
	there exists $ \hlconf\in\hlconfs(\beta) $ and $ \insta\in\instas(\hlconf) $ s.t.\
	\begin{enumerate}[label=\roman*),topsep=4pt]
		\item
		$
		\hlmarking=\confmark(\hlconf,\insta)  $,\qquad and
		\item $ \forall \hltrans\in\hltranss\,
		\forall\hlmode\in\hlmodes(\hltrans):\  \hlmarking[\hltrans,\hlmode\rangle\ \Rightarrow\
		\exists\hlevnt\in\hlevnts':
		\hlhomom(\hlevnt)=\hltrans
		\ \land\ \confcut(\hlconf,\insta)[\hlevnt,\hlmode\rangle.
		$
	\end{enumerate}
\end{definition}

We now define the class $ \Nf $ of high-level Petri nets
for which we generalize the construction of finite and complete prefixes of the unfolding of \emph{safe}
P/T Petri nets from~\cite{EsparzaRV02}.
We discuss the properties defining this class, and describe how it generalizes  safe P/T nets.
\begin{definition}[Class $ \Nf $]\label{def:Nf}
	The class $ \Nf $ contains all finite high-level Petri nets $ \hlpn=\tup{\hlplaces,\hltranss,\hlflowfunc,\hlguard,\hlmarkings_0} $
	satisfying the following three properties:
	\begin{enumerate}[label=\textbf{(\arabic*)},leftmargin=*,topsep=4pt]
		\item The net is \emph{safe}, i.e., in every reachable marking there lies at most $ 1 $ color on every place (formally; $ \forall \hlmarking\in\reachable(\hlpn)\,\forall\hlplace\in\hlplaces: \sum_{\hltok\in\hltoks}\hlmarking(\hlplace,\hltok)\leq 1 $).
		\item Guards are written in a decidable theory with the set $ \hltoks $ as its domain of discourse.
		\item The net has \textbf{f}initely many reachable markings (formally; $ |\reachable(\hlpn)|<\infty $).
	\end{enumerate}
\end{definition}
\newcommand{\itmSafety}{\textbf{(1)}}
\newcommand{\itmGuards}{\textbf{(2)}}
\newcommand{\itmFinite}{\textbf{(3)}}

We require the safety property \itmSafety\ for two reasons;
on the one hand, to avoid adding to the already heavy notation.
On the other hand, while we think that a generalization to bounded high-level Petri nets is possible, it comes with all the troubles known from going from safe to $ k $-bounded in the P/T case in \cite{EsparzaRV02}, plus the problems arising from the expressive power of the high-level formalism.
We therefore postpone this generalization to future work.
Note that, under the safety condition, we can w.l.o.g.\ assume
$ \hlns $ to be ordinary (i.e., $ \forall x,y\in\hlplaces\cup\hltranss: \sum_{\hlvar\in\hlvars}\hlflowfunc(x,v,y)\leq 1 $), since
transitions violating this property could never fire.
The finiteness of~$ \hlns $ implies that we can assume $ \hlvars $ to be finite.

While property \itmGuards\ seems very strong,
the goal is an algorithm that generates a complete finite prefix of the symbolic unfolding of a given high-level Petri net.
The definition of symbolic branching processes requires the predicate of every event added to the prefix to be satisfiable,
and the predicates are build from the guards in the given net.
Thus, satisfiability checks in the generation of the prefix seem for now inevitable.
An example for such a theory is Presburger arithmetic~\cite{Presburger30}, which is a first-order theory of the natural numbers with addition.
The guards in the example from Figure~\ref{fig:runEx} are expressible in Presburger arithmetic.

We need property~\itmFinite\ to ensure that the generalized version of the cut-off criterion from \cite{EsparzaRV02}
yields a finite prefix constructed in the generalized ERV-Algorithm.
$ |\reachable(\hlpn)|<\infty $ can be ensured by having a finite set $ \hltoks $ of colors.
In Sec.~\ref{sec:InfinteCase}, we identify a class of high-level Petri nets with infinitely many reachable markings for which the algorithm works with an adapted cut-off criterion.

Under these three assumptions we generalize the finite safe P/T Petri nets considered in \cite{EsparzaRV02}: every such P/T net can be seen as a high-level Petri net with $ \hltoks=\{\bullet \} $ and all guards being $ \true $,
and thus satisfying the three properties above.
Replacing the safety property \itmSafety\ by a respective ``$ k $-bounded property'' would result in a generalization of $ k $-bounded P/T nets.
In Sec.~\ref{sec:expansions},
we compare the result of the generalized ERV-algorithm Alg.~\ref{alg:complpref} applied to a high-level net to
the result of the original ERV-algorithm from \cite{EsparzaRV02} applied to the high-level net's expansion.

For the rest of the section let $ \hlpn=\tup{\hlplaces,\hltranss,\hlflowfunc,\hlguard,\hlmarkings_0}\in\Nf $
with symbolic unfolding $ \symbunf=\tup{\hlunf,\symbhomom}=\tup{\hlconds,\hlevnts,\hlarcs,\hlguard,\hlcuts_0,\symbhomom} $.

\subsection{Generalizing adequate orders and cut-off events}\label{sec:adequateOrders}
We lift the concept of adequate orders on the configurations of an occurrence net
to the level of symbolic unfoldings.
A main property of adequate orders is the preservation by finite \emph{extensions},
which are defined as for P/T-nets (cp.\ \cite{EsparzaRV02}):

Given a configuration $ \hlconf $, we denote $ \hlconf\cup\hlext $ by $ \hlconf{\oplus}\hlext $ if $ \hlconf\cup\hlext $ is a configuration
such that $ \hlconf\cap\hlext=\emptyset $. We say that $ \hlconf{\oplus}\hlext $ is an \emph{extension} of $ \hlconf $, and that $ \hlext $ is a \emph{suffix} of $ \hlconf $.
Obviously, for a configuration $ \hlconf' $, if $ \hlconf\subsetneq\hlconf' $ then there is a nonempty suffix $ \hlext $ of $ \hlconf $ such that $ \hlconf{\oplus}\hlext=\hlconf' $.
For a configuration $\hlconf{\oplus} \hlext $, denote by
$ \hlonet(\hlconf|\hlext)=\tup{ \confcut(\hlconf)\cup{{\rightarrow}\hlext}\cup{\hlext{\rightarrow}}, \hlext,\hlarcs',\confcuts(\hlconf) } $ the occurrence net ``around $ \hlext $'' from $ \confcut(\hlconf) $, where $ \hlarcs' $ is the restriction of $ \hlarcs $ to the nodes of $ \hlonet(\hlconf|\hlext) $.
Note that for every finite configuration $ \hlconf $ with an extension $ \hlconf{\oplus}\hlext $, we have that
$ \hlext $ is a configuration of $ {\Uparrow}\hlconf $.

We abbreviate for a marking
$ \hlmarking $ the fact
$
\exists \insta\in\instas(\hlconf{\oplus}\hlext):
\confmark(\hlconf,\insta|_{\hlvars_{\hlconf\cup\{\bot\}}})=\hlmarking $ by $ 	\hlconf\llbracket\hlmarking\rrbracket\hlext$
to improve readability.
Thus, $ \hlconf\llbracket\hlmarking\rrbracket\hlext $ means that the transitions corresponding to the events in $ \hlext $ can fire from $ \hlmarking\in\confmarks(\hlconf) $.

Since  we consider safe high-level Petri nets, we can relate two cuts representing the same set of places in the following way:
\begin{definition}\label{def:uniquemirror}
	Let $ \hlconf_1, \hlconf_2\in\hlconfs(\symbunf) $ with $ \symbhomom(\confcut(\hlconf_1))=\symbhomom(\confcut(\hlconf_2)) $.
	Then there is a unique bijection $ \mirror: \confcut(\hlconf_1)\to\confcut(\hlconf_2) $
	preserving $ \symbhomom $.
	We call this mapping $ \mirror_{\hlconf_1}^{\hlconf_2} $.
\end{definition}

The now stated Prop.~\ref{prop:futureIsoWeak2}
is a weak version of the arguments
in \cite{EsparzaRV02},
where the Esparza et~al. infer from
the low-level version of Prop.~\ref{prop:futureIso}
that
if the cuts of two low-level configurations represent the same marking in the low-level net,
then their futures are isomorphic, and the respective (unique) isomorphism maps the suffixes of one configuration to the suffixes of the other.
\begin{proposition}\label{prop:futureIsoWeak2}
	Let $ \hlconf_1 $ and $ \hlconf_2 $ be two finite configurations in $ \symbunf $,
	and let $ \hlext $ be a suffix of $ \hlconf_1 $.
	If there is a marking $ \hlmarking\in\confmarks(\hlconf_1)\cap\confmarks(\hlconf_2) $
	s.t.\
	$ \hlconf_1\llbracket\hlmarking\rrbracket\hlext $,
	then there is a unique monomorphism $ \extmonom_{1,\hlext}^{2}:\hlonet(\hlconf_1|\hlext)\to{\Uparrow}\hlconf_2 $ that satisfies
	$ \extmonom_{1,\hlext}^{2}(\confcut(\hlconf_1))=\confcut(\hlconf_2) $
	and preserves the labeling $ \symbhomom $.\\
	For this monomorphism we have that $ \extmonom_{1,\hlext}^{2}(\hlext) $ is a suffix of $ \hlconf_2 $.
\end{proposition}
\textit{Notation.}
For functions $ f:X\to Y $ and $ f':X'\to Y $ with $ X\cap X'=\emptyset $ we define
$ f\uplus f':X\cup X'\to Y $
by mapping $ x $ to $ f(x) $ if $ x\in X $ and to $ f'(x) $ if $ x\in X' $.
\begin{proof}
	By induction over the size $ k=|\hlext| $ of the suffix $ \hlext $.

	\emph{Base case $ k=0 $}.
	This means $ \hlext=\emptyset $.
	Then $ \hlonet(\hlconf_1|\hlext)=\tup{\confcut(\hlconf_1),\emptyset,\emptyset,\confcuts(\hlconf_1)} $.
	Since $ \hlmarking\in\confmarks(\hlconf_1)\cap\confmarks(\hlconf_2) $,
	we know that $ \symbhomom(\confcut(\hlconf_1))=\symbhomom(\confcut(\hlconf_2)) $.
	Since we only consider safe nets,
	$ \extmonom_{1,\hlext}^2 $ is uniquely realized by $ \mirror_{\hlconf_1}^{\hlconf_2}: \confcut(\hlconf_1)\to\confcut(\hlconf_2) $ from Def.~\ref{def:uniquemirror}.

	\emph{Induction step. Let $ k>0 $.}
	Let $ \insta\in\instas(\hlconf_1{\oplus}\hlext)$ s.t.\ $
	\confmark(\hlconf_1,\insta|_{\hlvars_{\hlconf_1\cup\{\bot\}}})=\hlmarking $.
	Let $ \hlevnt\in\Min(\hlext) $.
	Then for $ \hlmode= [\hlvar\gets\insta(\hlvar_\hlevnt)]_{\hlvar\in\hlvars(\hlevnt)}$ we have
	$ \hlmarking[\symbhomom(\hlevnt),\hlmode\rangle $.
	Thus, by Prop.~\ref{prop:unfComplete},
	$ \exists \hlevnt'\in\hlevnts:\symbhomom(\hlevnt')=\symbhomom(\hlevnt)\land \hlconf_2{\oplus}\{\hlevnt'\}\in\hlconfs(\symbunf) $.
	This means $ {\rightarrow}\hlevnt'\subseteq (\hlconds_0\cup(\hlconf_2{\rightarrow}))\setminus({\rightarrow}\hlconf_2) $; else, $ \hlconf_2{\oplus}\{\hlevnt' \} $ would not be a configuration.
	Thus, $ \hlevnt' $ is an event in $ {\Uparrow}\hlconf_2 $.
	Since $ \symbhomom(\hlevnt)=\symbhomom(\hlevnt') $,
	we get by definition of homomorphisms that
	$ \{ \tup{\symbhomom(\hlcond),\hlvar}\with \tup{\hlcond,\hlvar}\in\postset{\hlevnt} \}
	=\{ \tup{\symbhomom(\hlcond),\hlvar}\with \tup{\hlcond,\hlvar}\in\postset{\hlevnt'} \} $.
	The net $ \hlpn $ is safe, therefore we can define the bijection
	$ \mirror_1:(\hlevnt{\rightarrow})\to(\hlevnt'{\rightarrow}) $
	by $ \mirror_1(\hlcond)=\hlcond'\Leftrightarrow \symbhomom(\hlcond)=\symbhomom(\hlcond') $.
	We now define $ \extmonom_1: \hlonet(\hlconf_1|\{\hlevnt\})\to{{\Uparrow}\hlconf_2} $ by
	$ \extmonom_1=\mirror_{\hlconf_1}^{\hlconf_2}\uplus \{\hlevnt\mapsto\hlevnt'\}\uplus\mirror_1 $,
	which is a homomorphism satisfying the claimed conditions.

	Let now $ \hlconf_1'=\hlconf_1\cup\{ \hlevnt \} $,
	$ \hlconf_2'=\hlconf_2\cup\{\extmonom_1(\hlevnt) \} $
	and $ \hlext'=\hlext\setminus\{ \hlevnt \} $.
	We then have for $ \hlmarking' $ given by $ \hlmarking[\symbhomom(\hlevnt),\hlmode\rangle \hlmarking' $
	that $ \hlconf_1'\llbracket\hlmarking'\rrbracket\hlext' $,
	$ \hlmarking'\in\hlmarkings(\hlconf_1')\cap\hlmarkings(\hlconf_2') $,
	and $ |\hlext'|<k $.
	Thus, by the induction hypothesis, we get that there is a unique monomorphism $ \extmonom_2:\hlonet(\hlconf_1'|\hlext')\to{\Uparrow}\hlconf_2' $ satisfying the conditions above.
	Since $ \extmonom_1 $ and $ \extmonom_2 $ coincide on $ \confcut(\hlconf_1') $,
	we can now define $ \extmonom_{1,\hlext}^{2} $ by ``gluing together'' $ \extmonom_1 $ and $ \extmonom_2 $ at $ \confcut(\hlconf_1') $.

	This proves the claim for finite extensions.
	For an infinite extension, every node also contained in a finite extension.
	Due to uniqueness of the homomorphisms, we can define the $ \extmonom_{1,\hlext}^{2} $ in the case of an infinite $ \hlext $ as the union of all homomorphisms of smaller finite extensions.
\end{proof}

Equipped with Prop.~\ref{prop:futureIsoWeak2}, we can now lift the concept of adequate order to the level of symbolic branching processes.
Compared to \cite{McMillan95,EsparzaRV02}, the monomorphism~$ \extmonom_{1,\hlext}^{2} $ defined above  replaces the isomorphism $ I_1^2 $ between $ {\Uparrow}\hlconf_1 $ and $ {\Uparrow}\hlconf_2 $ for two low-level configurations $ \hlconf_1,\hlconf_2 $ representing the same marking.

\begin{definition}[Adequate order]\label{def:adequateOrder}
	A partial order $ \ado $ on the set of finite configurations of the symbolic unfolding of a high-level Petri net is an \emph{adequate order} if:
	\begin{enumerate}[label=\roman*),topsep=4pt]
		\item $ \ado $ is well-founded,
		\item $ \hlconf_1\subset\hlconf_2 $ implies $ \hlconf_1\ado\hlconf_2 $, and
		\item $ \ado $ is preserved by finite extensions in the following way:
		if $ \hlconf_1,\hlconf_2 $ are two finite configurations, and
		$ \hlconf_1{\oplus}\hlext $ is a finite extension of~$ \hlconf_1 $ such that
		there is a marking $ \hlmarking\in\confmarks(\hlconf_1)\cap\confmarks(\hlconf_2) $
		satisfying
		$ \hlconf_1\llbracket\hlmarking\rrbracket\hlext $,
		then the monomorphism
		$ \extmonom_{1,\hlext}^{2} $ from above satisfies
		$ \hlconf_1\ado\hlconf_2 \Rightarrow \hlconf_1{\oplus}\hlext\ado
		\hlconf_2{\oplus} \extmonom_{1,\hlext}^{2}(\hlext) $.
	\end{enumerate}
\end{definition}
In the case of a P/T net
we have $ |\hlmarkings(\hlconf)|=1 $ for every configuration $ \hlconf $,
and therefore, Def.~\ref{def:adequateOrder} coincides with its P/T version \cite{EsparzaRV02}.
We could alternatively generalize the P/T case by replacing  `$\exists\hlmarking\in \confmarks(\hlconf_1)\cap\confmarks(\hlconf_2) $ s.t.\ $ \hlconf_1\llbracket\hlmarking\rrbracket\hlext $' by
`$ \confmarks(\hlconf_1)=\confmarks(\hlconf_2) $',
and
use the isomorphism $ I_1^2 $ between
$ {\Uparrow}\hlconf_1 $ and $ {\Uparrow}\hlconf_2 $
to define preservation by finite extension.
However, in the upcoming generalization of the ERV-algorithm from~\cite{EsparzaRV02},
the generalized cut-off criterion exploits property iii) of adequate orders.
Using `$ \confmarks(\hlconf_1)=\confmarks(\hlconf_2) $' would produce an exponential blowup of the generated prefix's size.
This is circumvented by using `$\exists  \hlmarking\in\confmarks(\hlconf_1)\cap\confmarks(\hlconf_2) $
s.t.\
$ \hlconf_1\llbracket\hlmarking\rrbracket\hlext $',
which however leads to obtaining merely a monomorphism that depends on the considered suffix $ \hlext $, instead of an isomorphism between the futures.
We now show that this monomorphism sufficient:

The upcoming proof that the generalized ERV-algorithm is complete
is structurally analogous to the respective proof in \cite{EsparzaRV02}.
It	uses that, under the conditions of Def.~\ref{def:adequateOrder}~iii),
we also have
$ \hlconf_2\ado\hlconf_1 \Rightarrow \hlconf_2{\oplus} \extmonom_{1,\hlext}^{2}(\hlext)\ado\hlconf_1{\oplus}\hlext$.
This result would directly be obtained if  $ \extmonom_{1,\hlext}^{2} $ was an isomorphism, as $ I_1^2 $ is in the low-level case.
However, a monomorphism is an isomorphism when its codomain is restricted to its range.
This idea is used in the proof of the following proposition,
which states that $ \extmonom_{1,\hlext}^2 $ indeed satisfies the above property.
\begin{proposition}\label{prop:adequateOrderOtherDir}
	Let $ \ado $ be an adequate order.
	Under the conditions of Def.~\ref{def:adequateOrder}~iii)  the monomorphism~$ \extmonom_{1,\hlext}^2 $ also satisfies
	$ \hlconf_2\ado\hlconf_1 \Rightarrow \hlconf_2{\oplus} \extmonom_{1,\hlext}^{2}(\hlext)\ado\hlconf_1{\oplus}\hlext$.
\end{proposition}
\begin{proof}
	Let $ \hlext'=\extmonom_{1,\hlext}^{2}(\hlext) $.
	We first show that
	$ \extmonom_{2,\hlext'}^1(\hlext')=\hlext $.

	Let $ \extmonom_1:\hlonet(\hlconf_1|\hlext)\to\extmonom_{1,\hlext}^2(\hlonet(\hlconf_1|\hlext)) $ be the isomorphism that acts on $ \hlonet(\hlconf_1|\hlext) $ as $ \extmonom_{1,\hlext}^2 $ does,
	and let $ \extmonom_2:\hlonet(\hlconf_2|\hlext')\to \extmonom_{2,\hlext'}^1(\hlonet(\hlconf_2|\hlext')) $ be the isomorphism that acts on $ \hlonet(\hlconf_2|\hlext') $ as $ \extmonom_{2,\hlext'}^1 $ does.
	Since $ \extmonom_1^{-1}:\extmonom_{1,\hlext}^2(\hlonet(\hlconf_1|\hlext))\to \hlonet(\hlconf_1|\hlext) $ and $ \hlonet(\hlconf_1|\hlext)\subset {\Uparrow}\hlconf_1 $,
	and $ \extmonom_1^{-1}(\extmonom_{1,\hlext}^2(\hlext))=\hlext $ is a suffix of~$ \hlconf_1 $,
	we get by Prop.~\ref{prop:futureIsoWeak2} that
	$ \extmonom_1^{-1}=\extmonom_2 $,
	which means $ \extmonom_{2,\hlext'}^1(\hlext')=\hlext $.

	Assume now $ \hlconf_2\ado\hlconf_1 $.
	From the proof of Prop.~\ref{prop:futureIsoWeak2} we see that
	$ \hlconf_2\llbracket\hlmarking\rrbracket\extmonom_{1,\hlext}^2(\hlext) $.
	Thus, we get by the definition of adequate order and the result above that
	$ \hlconf_2{\oplus} \extmonom_{1,\hlext}^{2}(\hlext)\ado
	\hlconf_1{\oplus}\extmonom_{2,\extmonom_{1,\hlext}^2(\hlext)}^1(\extmonom_{1,\hlext}^2(\hlext))
	\allowbreak=\hlconf_1{\oplus}\hlext $.
\end{proof}
In \cite{EsparzaRV02}, Esparza et al.\ discuss three adequate orders
on the configurations of the low-level unfolding.
In particular, they present a \emph{total} adequate order that uses the \emph{Foata normal form} of configurations.
Using such a total order in the algorithm limits the size of the resulting finite and complete prefix:
it contains at most $ |\reachable(\hlpn)| $ non cut-off events.
All three adequate orders presented in \cite{EsparzaRV02}
can be directly lifted to the configurations of the symbolic unfolding by exchanging every low-level term by its high-level counterpart.
The lifted order using the Foata normal form is still a total order.
We include these discussions in App.~\ref{app:adequateOrders}.

We now define cut-off events in a symbolic unfolding.
In the low-level case~\cite{EsparzaRV02},
$ \evnt $ is a cut-off event if there is another event $ \evnt' $ satisfying $\cone{\evnt'} \ado \cone{\evnt}$ and $ \confmark(\cone{\evnt})=\confmark(\cone{\evnt'}) $,
which ensures that the future of $ \evnt $ needs not be considered further.
In the high-level case, we generalize these conditions to high-level events $ \hlevnt $.
However, we do not require the existence of \emph{one}
other high-level event $ \hlevnt' $ with $ \cone{\hlevnt'}\ado\cone{\hlevnt} $
and $ \confmarks(\cone{\hlevnt})=\confmarks(\cone{\hlevnt'}) $.
While this would still be a valid cut-off criterion
and would lead to finite and complete prefixes,
the upper bound on the size of such a prefix would be exponential in the number of markings in the original net.
Instead, we check whether $ \confmarks(\cone{\hlevnt}) $ is contained in the union of \emph{all} $ \confmarks(\cone{\hlevnt'}) $ with $ \cone{\hlevnt'}\ado\cone{\hlevnt} $.
This criterion expresses that we have already seen every marking in $ \confmarks(\cone{\hlevnt}) $ in the prefix $ \beta $ under construction, and therefore need not consider the future of  $ \hlevnt $ any further.
By this, we obtain the same upper bounds as in \cite{EsparzaRV02}, as discussed~later.
\begin{definition}[Cut-off event]\label{def:cut-off}
	Let $ \ado $ be an adequate order on the configurations of the symbolic unfolding of a high-level Petri net. Let $ \beta $ be a prefix of the symbolic unfolding containing a high-level event $ \hlevnt $.
	The high-level event $ \hlevnt $ is a \emph{cut-off} event in $ \beta $ (w.r.t.\ $ \ado $) if
	$ \confmarks(\cone{\hlevnt})\subseteq \bigcup_{\cone{\hlevnt'}\ado\cone{\hlevnt}}\confmarks(\cone{\hlevnt'}). $
\end{definition}
For P/T Petri nets,
this definition corresponds to the cut-off events defined in \cite{EsparzaRV02},
since in this case $ |\hlmarkings(\cone{\hlevnt})|=1 $ for all events $ \hlevnt $.

\subsection{The generalized ERV-algorithm}\label{sec:genERValgo}
We present the algorithm for constructing a finite and complete prefix of the symbolic unfolding of a given high-level Petri net.
It is a generalization of the ERV-algorithm from \cite{EsparzaRV02},
and is structurally equal (and therefore looks very similar).
However, the algorithm is contingent upon the previous section's work of generalizing adequate orders and cut-off events,
which ultimately enables us to adopt this structure.

A crucial concept of the ERV-algorithm is the notion of ``possible extensions'',
i.e., the set of individual events that extend a given prefix of the unfolding.
In Def.~\ref{def:PossExt}, we lift this concept to the high-level formalism.
We do so by isolating the procedure of adding high-level events in the
algorithm from \cite{ChatainJ04} which generates the complete symbolic unfolding of a given high-level Petri net (but does not terminate if the symbolic unfolding is infinite).

We define the data structures similarly to \cite{EsparzaRV02}.
There, an event is given by a tuple $ \evnt=\tup{\trans,\conds'} $ with $ \homom(\evnt)=\trans\in\transs $ and $ \preset{\evnt}=\conds'\subseteq\conds $,
and a condition given by a tuple $ \cond=\tup{\place,\evnt}$ with $ \homom(\cond)=\place\in\places $ and $ \preset{\cond}=\{\evnt \}\subseteq\evnts $.
The finite and complete prefix is a set of such events and transitions.

In the high-level case, we need more information inside the tuples.
A high-level event is given by a tuple
$ \hlevnt=\tup{\hltrans,X,\pred} $ described by
$ \hlhomom(\hlevnt)=\hltrans $,
$ \preset{\hlevnt}=X\subseteq \hlconds\times\hlvars $,
and $ \pred(\hlevnt)=\pred $.
Analogously, a high-level condition is given by a tuple
$ \hlcond=\tup{\hlplace,\tup{\hlevnt,\hlvar},\pred} $, where
$ \hlhomom(\hlcond)=\hlplace $,
$ \preset{\hlcond}=\tup{\hlevnt,\hlvar}\in(\hlevnts\times\hlvars)\cup(\{\bot\}\times\{\hlvar^\hlcond\with\hlcond\in\hlconds_0\}) $,
and $ \pred(\hlcondevnt(\hlcond))=\pred $.

\begin{definition}[Possible extensions]
	\label{def:PossExt}
	Let $ \beta=\tup{\hlonet,\hlhomom} $ be a branching process of a high-level Petri net $ \hlpn $.
	The \emph{possible extensions} $ \mathit{PE}(\beta) $ are the set of
	tuples $\hlevnt= \tup{\hltrans,X,\pred} $ where $ \hltrans $ is a transition of $ \hlpn $, and $ X\subseteq\hlconds\times\hlvars $ satisfying
	\begin{itemize}[topsep=4pt]
     \itemsep=0.8pt
		\item
		$ \{\hlcond\with \tup{\hlcond,\hlvar}\in X\} $ is a co-set, and
		$ \preset{ \hltrans }=\{\tup{\hlhomom(\hlcond),\hlvar}\with \tup{\hlcond,\hlvar}\in X \}$,
		\item $ \pred=
		\locpred\land\big(\bigwedge_{\tup{\hlcond,\hlvar}\in X}\pred(\hlcondevnt(\hlcond))\big)
		$
		is satisfiable,\\
		where
		$ \locpred=\hlguard(\hltrans)[\hlvar\leftarrow\hlvar_\hlevnt]_{\hlvar\in\hlvars(\hlevnt)}
		\land \big(
		\bigwedge_{\tup{\hlcond,\hlvar}\in X}
		\hlvar_\hlevnt= \hlcondvar_\hlcondevnt(\hlcond) \big) $,
		\item $ \beta $ does not contain $ \tup{\hltrans,X,\pred} $.
	\end{itemize}
\end{definition}
Since the notion of co-set in high-level occurrence nets is
achieved by the direct translation from low-level occurrence nets
plus the ``color conflict freedom'',
possible extensions in a prefix $ \beta $
can be found by searching first for sets of conditions that are not in structural conflict as in the low-level case, and then checking whether these sets are in color conflict.

\medskip
Alg.~\ref{alg:complpref} is a generalization of the ERV-Algorithm in \cite{EsparzaRV02} for complete finite prefixes of the low-level unfolding.
The structure is taken from there,
with the only difference being the special initial transition~$ \bot $.
It takes as input a high-level Petri net $ \hlpn\in\Nf $ and assumes a given adequate order~$ \ado $.

\vspace*{-2mm}
\begin{algorithm}[!ht]
	\caption{Generalization of the ERV-Algorithm from \cite{EsparzaRV02} for complete finite prefixes.\label{alg:complpref}}
	\KwData{High-level Petri net $\hlpn=\tup{\hlplaces,\hltranss,\hlflowfunc,\hlguard,\hlmarkings_0}\in\Nf$.}
	\KwResult{A complete finite prefix $ \mathit{Fin}$
		of the symbolic unfolding of~$ \hlpn $.}
	$ \fin \gets \{\bot \} $\;
	$ \pred(\bot)\gets\bigvee_{\hlmarking_0\in\hlmarkings_0}\bigwedge_{\tup{\hlplace,\hltok}\in\hlmarking_0}\hlvar_\bot^{\hlcond_\hlplace}=\hltok $\;
	\ForEach{$ \hlplace\in\hlplaces_0 $}{
		Create a fresh condition $ \hlcond_\hlplace=\tup{\hlplace,\tup{\bot,\hlvar^{\hlcond_\hlplace}},\pred(\bot)} $\;
		$ \fin\gets\fin\cup\{ \hlcond \} $\;
	}
	$ \possextalg \gets \possext(\fin) $\;
	$ \cutoff \gets \emptyset $\;
	\While{$ \possextalg\neq\emptyset $}{
		Pick $ \hlevnt=\tup{\hltrans,X,\pred} $ from $ \possextalg $ such that $ \cone{\hlevnt} $ is minimal w.r.t.\ $ \ado $\;
		\uIf{$ \cone{\hlevnt}\cap\cutoff=\emptyset $}{
			$ \fin\gets\fin\cup\{\hlevnt\} $\;
			\ForEach{$ \tup{\hlplace,\hlvar}\in\postset{\hltrans} $}{
				Create a fresh condition $ \hlcond=\tup{\hlplace,\tup{\hlevnt,\hlvar},\pred} $\;
				$ \fin\gets\fin\cup\{ \hlcond \} $\;
			}
			$ \possextalg\gets\possext(\fin) $\;
			\If{$ \hlevnt $ is a cut-off event of $ \fin $}{
				$ \cutoff\gets\cutoff\cup\{\hlevnt\} $\;
			}
		}
		\Else{
			$ \possextalg\leftarrow\possextalg\setminus\{ \hlevnt \} $
		}
	} \vspace*{-2mm}
   \end{algorithm}

\begin{example}
	Consider the running example $ \hlpn $ from Figure~\ref{fig:runEx}.
	Alg.~\ref{alg:complpref} produces the complete finite prefix
	marked by the blue line in Figure~\ref{fig:runExUnf}.
	Cut-off events are shaded blue.

	Starting with the initial conditions $ a' $ and $ b' $,
	the possible extensions are $ \alpha' $ and $ \beta' $.
	assuming $ \cone{\alpha'}\ado\cone{\beta'} $,
	we first attach $ \alpha' $ and a condition $ c' $
	corresponding to the output place $ c $ of $ \alpha $,
	and then analogously $ \beta' $ and the condition $ d' $.

	For $ \alpha' $ we have $ \confmarks(\cone{\alpha'})=\{ \ms{\tup{c,n},\tup{b,0}}\with n\in\{1,\dots,m\} \} $ and analogously, for $ \beta' $ we have $ \confmarks(\cone{\beta'})=\{ \ms{\tup{a,0},\tup{d,n}}\with n\in\{1,\dots,m\} \} $.
	Since we have not seen these markings before, neither $ \alpha' $ nor $ \beta' $ is a cut-off event.
	Thus, we have the possible extensions $ t' $ and $ \varepsilon' $.
	For $ t' $ we have  $ \confmarks(\cone{t'})=\{ \ms{} \}  $, since no tokens are in the net after firing $ t $.
	However, we have not seen the empty marking $ \ms{} $ before, so formally, $ t' $ is not a cut-off event.

	For $ \varepsilon' $ we have $ \confmarks(\cone{\varepsilon'})=\{ \ms{\tup{c,n},\tup{d,n'}}\with n,n'\in\{1,\dots,m\} \} $.
	Corresponding cuts can be reached in the prefix constructed so far by concurrently firing $ \alpha' $ and $ \beta' $.
	However, no marking $ \ms{\tup{c,n},\tup{d,n'}} $ is represented by a \emph{cone} configuration before $ \varepsilon' $, and therefore $ \varepsilon' $ does \emph{not} satisfy
	$ \confmarks(\cone{\varepsilon'})\subseteq \bigcup_{\cone{\hlevnt'}\ado\cone{\varepsilon'}}\confmarks(\cone{\hlevnt'}) $.
	This means $ \varepsilon' $ is not a cut-off event and we have to proceed with the possible extensions $ t'' $ and $ \varepsilon'' $.

	Since $ \confmarks(\cone{t''})=\{ \ms{} \} = \confmarks(\cone{t'}) $ with $ \cone{t'}\ado\cone{t''} $,
	have that $ t'' $ is a cut-off event.
	This, however, has no impact on the prefix since we cannot continue after $ t'' $ anyway.
	For $ \varepsilon'' $ we have $ \confmarks(\cone{\varepsilon''})=\{ \ms{\tup{c,n},\tup{d,n'}}\with n,n'\in\{1,\dots,m\} \}=\confmarks(\cone{\varepsilon'}) $ with $ \cone{\varepsilon'}\ado\cone{\varepsilon''} $.
	This makes $ \varepsilon'' $ also a cut-off event.
	We therefore have no more possible extensions, and the algorithm terminates.
	In the figure, this is indicated by the blue lines.
\end{example}

We now prove correctness of Alg.~\ref{alg:complpref} analogously to \cite{EsparzaRV02},
by stating two propositions -- one each to show that the prefix is finite and complete, respectively.
The proof structure is also as in \cite{EsparzaRV02},
but adapted to the setting of high-level Petri nets and symbolic unfoldings.

\begin{proposition}\label{prop:FinIsFin}
	After applying Alg.~\ref{alg:complpref} to a high-level Petri net $\hlpn\in\Nf$, the result $ \fin $ is finite.
\end{proposition}
Given an event $ \hlevnt $, define the \emph{depth} of $ \hlevnt $
as the length of the longest chain of events $ \hlevnt_1<\hlevnt_2<\dots<\hlevnt $;
the depth of $ \hlevnt $ is denoted by $ \depth(\hlevnt) $.

\begin{proof}
	As in \cite{EsparzaRV02}, we prove the following results (1) -- (3):

	\begin{enumerate}[label=(\arabic*)]
		\item For every event $ \hlevnt $ of $ \fin $, $ \depth(\hlevnt)\leq |\reachable(\hlpn)|+1 $,
		\item
		For every event $ \hlevnt $ of $ \fin $, the sets $ \preset{\hlevnt} $ and $ \postset{\hlevnt} $ are finite, and
		\item
		For every $ k\geq 0 $, $ \fin $ contains only finitely many events $ \hlevnt $ such that $ \depth(\hlevnt)\leq k $.
	\end{enumerate}
	This works exactly as in \cite{EsparzaRV02},
	with minor adaptations to the generalization of cut-offs in the symbolic unfolding in (1):

	\begin{enumerate}[label=(\arabic*)]
		\item
		Let $ n=|\reachable(\hlpn)| $.
		Every chain of events $ \hlevnt_1<\hlevnt_2<\dots<\hlevnt_{n}<\hlevnt_{n+1}  $
		in the unfolding contains an event $ \hlevnt_i $, $ i>1 $,
		s.t.\ $ \confmarks(\cone{\hlevnt_i})\subseteq\bigcup_{j=1}^{i-1}\confmarks(\cone{\hlevnt_j}) $,
		since, if every $ \confmarks([\hlevnt_j]) $, $ j=1,\dots,n $, contains a marking not contained in
          $ \bigcup_{k=1}^{j-1}\confmarks(\cone{\hlevnt_k}) $,
		then finally $ \bigcup_{j=1}^{n}\confmarks(\cone{\hlevnt_j}) $ contains all $ n $ markings.
		This makes $ \hlevnt_{n+1} $ a cut-off event.
		\item
		By the construction in the algorithm we see that there is a bijection between $ \postset{\hlevnt} $ and $ \postset{\hlhomom(\hlevnt)} $,
        and similarly for $ \preset{\hlevnt} $ and $ \preset{\hlhomom(\hlevnt)} $. The result then follows from the finiteness of $ \hlpn $.
		\item
		By complete induction on $ k $. The base case, $ k = 0 $, is trivial.
		Let $ E_k $ be the set of events of depth at most $ k $.
		We prove that if $ E_k $ is finite then $ E_{k+1} $ is finite.
		By (2) and the induction hypothesis, $ \postset{E_{k}} $ is finite. Since $ \{\hlcond\with \exists\hlvar\in\hlvars:\tup{\hlcond,\hlvar}\in\preset{E_{k+1}}\}\subseteq \{\hlcond\with \exists\hlvar\in\hlvars:\tup{\hlcond,\hlvar}\in\postset{E_{k}}\}$, we
		get by property {\it iv} in the definition of occurrence nets that $ E_{k+1} $ is finite.
	\end{enumerate}
\end{proof}

\begin{proposition}\label{prop:FinIsCompl}
	After applying Alg.~\ref{alg:complpref} to a high-level Petri net $\hlpn\in\Nf$, the result $ \fin $ is complete.
\end{proposition}
The proof of this proposition also has the same general structure as the respective proof in~\cite{EsparzaRV02}.
However here we use the generalizations of adequate order,
possible extensions,
and the cut-off criterion to symbolic branching processes.

\begin{proof}
	We first show that for every reachable marking in $ \hlpn $ there exists a configuration in $ \symbunf $ satisfying a) from the definition of complete prefixes,
	and then show that one of these configurations (a minimal one) also satisfies b).

	\begin{enumerate}[label=(\arabic*),topsep=4pt]
		\item Let $ \hlmarking $ be an arbitrary reachable marking in $ \hlpn $.
		Then by Prop.~\ref{prop:MarkingsRepresented}, we have that there is a $ \hlconf_1\in\hlconfs(\symbunf) $ s.t.\ $ \hlmarking\in\confmarks(\hlconf_1) $.
		Let $ \insta_1\in\instas(\hlconf_1) $ s.t.\ $ \hlmarking=\confmark(\hlconf_1.\insta_1) $.
		If $ \hlconf $ is not a configuration in $ \fin $, then it contains a cut-off event $ \hlevnt_1 $,
		and so $ \hlconf_1=\cone{\hlevnt_1}{\oplus}\hlext $ for some set $ \hlext $	of events.
		Let $ \hlmarking_1=\confmark(\cone{\hlevnt_1}.\insta_1|_{\hlvars_{\cone{\hlevnt_1}\cup\{\bot\}}})\in\confmarks(\cone{\hlevnt_1}) $.
		By the definition of cut-off event, there exists an event $ \hlevnt_2 $ with $ \cone{\hlevnt_2}\ado\cone{\hlevnt_1} $
		and $ \hlmarking_1\in\confmarks(\cone{\hlevnt_2}) $.
		Since we have $C_{1}\llbracket M_{1}\rrbracket D$,
		we get by Prop.~\ref{prop:futureIsoWeak2} that the monomorphism $ \extmonom_1:=\extmonom_{\cone{\hlevnt_1},\hlext}^{\cone{\hlevnt_2}}:\hlonet(\cone{\hlevnt_1}|\hlext)\to{\Uparrow}\cone{\hlevnt_{2}} $ exists and that $ \extmonom_1(\hlext) $ is a suffix of $ \cone{\hlevnt_2} $. By Prop.~\ref{prop:adequateOrderOtherDir} we know
		\begin{equation*}
			\hlconf_2:=\cone{\hlevnt_2}{\oplus}\extmonom_1(\hlext)\ado\cone{\hlevnt_1}{\oplus}\hlext=\hlconf_1 .
		\end{equation*}
		Let $ \insta_2'\in\instas(\cone{\hlevnt_2}) $ s.t.\ $ \hlmarking_1=\confmark(\cone{\hlevnt_2},\insta_2') $.
		Define now $ \insta_2\in\instas(\hlconf_2) $ by $ \insta_2=\insta_2'\uplus\insta_2'' $,
		where $ \insta_2'':\hlvars_{\extmonom_1(\hlext)}\to\hltoks $ is given by
		$ \insta_2''(\hlvar_{\extmonom_1(\hlevnt)})=\insta_1(\hlvar_\hlevnt) $.
		By this construction we get $ \hlmarking=\confmark(\hlconf_2,\insta_2) \in \confmarks(\hlconf_2) $.

		\smallskip
         If $ \hlconf_2 $ is not a configuration of
		$ \fin $, then we can iterate the procedure and find a configuration $ \hlconf_3 $ such that $ \hlconf_3\ado\hlconf_2 $
		and $ \hlmarking\in\confmarks(\hlconf_3) $. The procedure cannot be iterated infinitely often because $ \ado $ is well-founded. Therefore, it terminates in a configuration of $ \fin $.
		\item
		Let now $ \hlconf $ be a minimal configuration w.r.t.\ $ \ado $ s.t.\ $ \hlmarking\in\confmarks(\hlconf) $, and let $ \hltrans\in\hltranss $, $ \hlmode\in\hlmodes(\hltrans) $ s.t.\ $ \hlmarking[\hltrans,\hlmode\rangle $.
		If $ \hlconf $ contains some cut-off event, then we can apply the arguments
		of a) to conclude that $ \fin $ contains a configuration $ \hlconf'\ado\hlconf $ such that $ \hlmarking\in\confmarks(\hlconf') $.
		This contradicts the minimality of $ \hlconf $. So $ \hlconf $ contains no cut-off events.
		Let $ \insta\in\instas(\hlconf) $ s.t.\ $ \hlmarking=\confmark(\hlconf,\insta) $.
		Since $ \preset{\hltrans.\hlmode}\subseteq \hlmarking $,
		we have that there is a co-set $ \hlconds_{\hltrans,\hlmode}\subseteq\confcut(\hlconf) $
		s.t.\ $ \preset{\hltrans,\hlmode}=\{ \tup{\hlhomom(\hlcond),\insta(\hlcondvar_\hlcondevnt(\hlcond))}\with \hlcond\in\hlconds_{\hltrans,
             \hlmode}\}$.

\smallskip
		Let now $
		X:=\{ \tup{\hlcond,\hlvar}\with \hlcond\in\hlconds_{\hltrans,\hlmode}, \tup{\hlhomom(\hlcond),\hlvar}\in\preset{\hltrans}\}$.
		We then have $ \forall \tup{\hlcond,\hlvar}\in X: \hlmode(\hlvar)=\insta(\hlcondvar_\hlcondevnt(\hlcond)) $.

		We now show that
		\begin{equation*}
			\pred:=\hlguard(\hltrans)[\hlvar\leftarrow\hlvar_\hlevnt]_{\hlvar\in\hlvars(\hlevnt)}\land \Big(
			\bigwedge_{\tup{\hlcond,\hlvar}\in X}
			\hlvar_\hlevnt= \hlcondvar_\hlcondevnt(\hlcond) \Big)\land\bigwedge_{(\hlcond,\hlvar)\in X}\pred(\hlcondevnt(\hlcond))
		\end{equation*}  is satisfiable.
		Let $ \insta':=\insta\uplus(\hlmode\circ[\hlvar_\hlevnt\mapsto\hlvar]_{\hlvar\in\hlvars(\hlevnt)})  $.
		Then
		\begin{itemize}
			\item $ \hlguard(\hltrans)[\hlvar\leftarrow\hlvar_\hlevnt]_{\hlvar\in\hlvars(\hlevnt)}[\insta']
			\equiv\hlguard(\hltrans)[\hlmode]\equiv\true
			$, and
			\item
			$ \big(	\bigwedge_{\tup{\hlcond,\hlvar}\in X}
			\hlvar_\hlevnt= \hlcondvar_\hlcondevnt(\hlcond) \big)[\insta']
			\equiv
			\big(\bigwedge_{\tup{\hlcond,\hlvar}\in X}
			\hlmode(\hlvar)=\insta(\hlcondvar_\hlcondevnt(\hlcond)) \big)\equiv \true $, and
			\item $ \bigwedge_{\tup{\hlcond,\hlvar}\in X}\pred(\hlcondevnt(\hlcond))[\insta']
			\equiv
			\bigwedge_{\tup{\hlcond,\hlvar}\in X}\pred(\hlcondevnt(\hlcond))[\insta]\equiv\true$, since $ \insta\in\instas(\hlconf) $.
		\end{itemize}
		Thus, $ \pred[\insta']\equiv\true $.
		Therefore, $ \hlevnt=\tup{\hltrans,X,\pred} $ is a possible extension and added in the execution of the algorithm.
		Then we directly have $ \hlevnt\notin\hlconf $, $ \hlhomom(\hlevnt)=\hltrans $, and with the same arguments as in a), we get $ \hlconf\cup\{ \hlevnt \}\in\hlconfs(\fin) $ and
		$ \insta\uplus(\hlmode\circ[\hlvar_\hlevnt\mapsto\hlvar]_{\hlvar\in\hlvars(\hlevnt)})\in\instas(\hlconf\cup\{ \hlevnt \}) $,
		which means $ \confcut(\hlconf,\insta)[\hlevnt,\hlmode\rangle $.
		Since we chose $ \insta $ independently of $ \hltrans$ and $\hlmode $,
		this concludes the proof.
	\end{enumerate}
\end{proof}
Notice that by this construction, as described in \cite{EsparzaRV02},
we get that if $ \ado $ is a total order, then $ \fin $ contains at most $ |\reachable(\hlpn)| $ non cut-off events.
As mentioned in Sec.~\ref{sec:adequateOrders}, the total adequate order defined in \cite{EsparzaRV02} can be lifted to the configurations in the symbolic unfolding, where it again is total (cp. App.~\ref{app:adequateOrders}).
Thus, we generalized the possibility to construct such a small complete finite prefix by application of Alg.~\ref{alg:complpref} with $ \ado $ being a total adequate order.

\section{High-level versus P/T expansion}\label{sec:expansions}
Every high-level Petri net represents a P/T Petri net (in which the places can only carry a number tokens with color $ \bullet $) with the same behavior, called its expansion.
Markings in a P/T Petri net describe only how many tokens lie on each place. Each transitions only has one possible firing mode that takes and/or places a fixed number of tokens from resp.\ onto each connected place.

\medskip
In this section we state in Lemma~\ref{lem:ExpansionOfComplPref}
that the expansion of a finite complete prefix of the unfolding of a high-level Petri net is a finite and complete prefix of the unfolding of the expanded high-level Petri net. This means the generalization of complete prefixes is ``canonical'', and compatible with the established low-level concepts.
We then compare for our running example the results of
\begin{itemize}[topsep=4pt]
	\item applying the generalized ERV-algorithm Alg.~\ref{alg:complpref} to obtain a complete finite prefix of the symbolic unfolding of a given high-level Petri net, and
	\item first expanding a given high-level Petri net and then applying the ERV-algorithm from \cite{EsparzaRV02} for a complete finite prefix of the (P/T) unfolding.
\end{itemize}

The procedure of constructing the represented P/T Petri net $ \expf(\hlpn) $ (called the \emph{expansion}) of a high-level Petri net $ \hlpn $ is well established (cp., e.g., Chapter~2.4 in \cite{Jensen96}), and we describe it here only briefly;
the places of $ \expf(\hlpn) $ are given by $ \places=\{ \hlplace.\hltok\with\hlplace\in\hlplaces, \hltok\in\hltoks\} $,
and its transitions by $ \transs=\{ \hltrans.\hlmode\with\hltrans\in\hltranss, \hlmode\in\hlmodes(\hltrans) \} $.
There is an arc from $ \hlplace.\hltok $ to $ \hltrans.\hlmode $ iff $ \tup{p,c}\in\preset{\hltrans,\hlmode} $,
and analogously for arcs from transitions to places.
Markings in $ \expf(\hlpn) $ are functions $ \marking:\places\to\N $, describing how often the only color $ \bullet $ lies on each place $ \hlplace.\hltok $.
Every such marking corresponds to a marking~$ \hlmarking $ in the high-level net $ \hlpn $,
with $ \hlmarking(\hlplace,\hltok)=\marking(\hlplace.\hltok) $,
and a transition $ \hltrans $ can fire in mode $ \hlmode $ from $ \hlmarking $ iff
$ \hltrans.\hlmode $ can fire from $ \marking $.
Thus, we say that $ \hlpn $ and $ \expf(\hlpn) $ have the same behavior.
For a finite high-level Petri net~$ \hlpn $,
the expansion $ \expf(\hlpn) $ is finite iff $ \hltoks $ is finite.

The \emph{(low-level) unfolding} of a P/T Petri net $ \pn $ is a tuple $ \llunf(\pn)=\tup{\unf(\pn),\llhomom^\pn} $,
where $ \unf(\pn) $ is an \emph{occurrence net}, and $ \llhomom^\pn:\unf(\pn)\to\pn $ is a \emph{Petri net homomorphism}
such that $ \tup{\unf(\pn),\llhomom^\pn} $ is a \emph{maximum branching process} of $ \pn $.
Since the definition of (low-level) branching processes and homomorphisms is
firstly well established in the literature,
and secondly so similar to their corresponding high-level definitions,
we omit them here and refer the reader for example  to \cite{Engelfriet91,EsparzaRV02}.

\medskip
With the notation of low-level unfoldings of P/T nets
we can define for a high-level occurrence net $ \hlonet $,
the P/T occurrence net
$ \expof(\hlonet):=\unf(\expf(\hlonet)) $
(i.e., the occurrence net from the unfolding $ \llunf(\expf(\hlonet))=\tup{\unf(\expf(\hlonet)),\llhomom^{\expf(\hlonet)}} $).
We abbreviate $ \llhomom^{\expf(\hlonet)} $ by~$ \llhomom^{\hlonet} $.
The  operator $ \expof $ therefore maps high-level occurrence nets to occurrence nets (cf.~\cite{ChatainF10}).
Let now $ \beta=\tup{\hlonet,\hlhomom} $ be a symbolic branching process of $ \hlpn $.
Then we can define the \emph{expanded symbolic branching process}
$ \expof(\beta)=\tup{\expof(\hlonet),\homom} $ of $ \expf(\hlpn) $
with the homomorphism
$ \homom:\expof(\hlonet)\to\expf(\hlpn) $,
defined  \linebreak  by $ \homom(\evnt)=\hltrans.\hlmode\Leftrightarrow \llhomom^\hlonet(\evnt)=\hlevnt.\hlmode\land\hlhomom(\hlevnt)=\hltrans $
and $\homom(\cond)=\hlplace.\hltok\Leftrightarrow \llhomom^\hlonet(\cond)=\hlcond.\hltok\land\hlhomom(\hlcond)=\hlplace $ for events $ \evnt $ resp.\ conditions $ \cond $ in $ \expof(\hlonet) $.
The following diagram serves as an overview:
\begin{center}
	\includegraphics{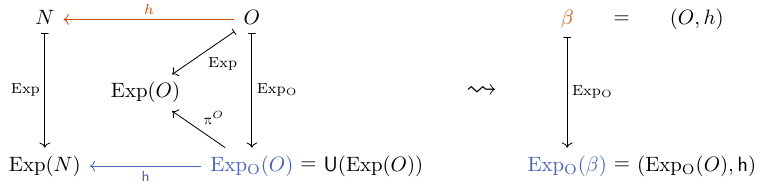}
\end{center}

The following result is shown in \cite{ChatainF10}.
It states that, for a high-level Petri net $ \hlpn $, the unfolding of~$ \hlpn $'s expansion is isomorphic to the expanded symbolic unfolding of $ \hlpn $.
\begin{lemma}[\cite{ChatainF10}, Sec. 4.1]\label{lem:ExpUCom}
	$ \llunf(\expf(\hlpn))\simeq\expof(\symbunf(\hlpn)) $.
\end{lemma}
With this result, we state the following:
\begin{lemma}\label{lem:ExpansionOfComplPref}
	Let $ \hlpn $ be a high-level Petri net and $ \beta $ be a prefix of $ \symbunf(\hlpn) $.
	Then
	$ \beta $~is finite and complete if and only if
	$ \expof(\beta) $ is a finite and complete prefix of $ \llunf(\expf(\hlpn)) $.
\end{lemma}
The proof uses the results from Prop.~\ref{prop:MarkingsRepresented} and Prop.~\ref{prop:unfComplete},
since the definition of completeness on the symbolic level is a direct translation from its P/T analogue.
\begin{proof}
	Let $ \beta=\tup{\hlonet,\hlhomom} $ be finite and complete.
	From Lemma~\ref{lem:ExpUCom} we already know that $ \expof(\hlonet)\subseteq\unf(\expf(\hlpn)) $.
	Since $ \expof(\beta) $ is a branching process of $ \expf(\hlpn) $,
	we see that is a prefix of the unfolding of $ \expf(\hlpn) $.
	Also,
	$ \expof(\beta) $ is obviously finite since $ \hlonet $ is a finite high-level occurrence net.

	We now prove that $ \expof(\beta)=\tup{\expof(\hlonet),\homom} $ is complete.
	Let $ \marking $ be a reachable marking in $ \expf(\hlpn) $.
	Then
	the high-level marking $ \hlmarking $ defined by
	$ \hlmarking(\hlplace,\hltok)=\marking(\hlplace.\hltok) $
	is reachable in $ \hlpn $.
	Thus, since $ \beta $ is complete, there is a configuration $ \hlconf\in\hlconfs(\beta) $
	and an instantiation $ \insta\in\instas(\hlconf) $
	satisfying a) and b) from Def.~\ref{def:ComplPref}.
	This means there is a firing sequence
	$ \hlcut_0[\hlevnt_1,\hlmode_1\rangle\hlcut_1\dots[\hlevnt_n,\hlmode_n\rangle\hlcut_n $
	with $ \{ \hlevnt_1,\dots,\hlevnt_n \}=\hlconf $,
	$ \hlmode_i=\insta\circ [\hlvar\mapsto\hlvar_{\hlevnt_i}]_{\hlvar\in\hlvars(\hlevnt_i)} $,
	and $ \hlcut_n=\confcut(\hlconf,\insta) $
	(meaning $ \hlmarking=\confmark(\hlconf,\insta)=\ms{\tup{\hlhomom(\hlcond),\hltok}\with\tup{\hlcond,\hltok}\in\hlcut_n} $).
	Then, in $ \expf(\hlonet) $,
	the marking $ \ms{\hlcond.\hltok\with\tup{\hlcond,\hltok}\in\hlcut_n} $
	is reachable from the initial marking $ \ms{\hlcond.\hltok\with\tup{\hlcond,\hltok}\in\hlcut_0} $ by the firing sequence
	$ \tup{\hlevnt_1.\hlmode_1,\dots,\hlevnt_n.\hlmode_n} $.
	Thus, there is a configuration $ \conf=\{ \evnt_1,\dots,\evnt_n \} $
	in $ \unf(\expf(\hlonet))=\expof(\hlonet) $
	with $\forall i: \llhomom^{\mathrm{O}}(\evnt_i)= \hlevnt_i.\hlmode_i $.
	Then, by the definition of $ \homom $, we get $ \confmark^{\homom}(\conf):=\ms{\homom(\cond)\with\cond\in\confcut(\conf)} = \ms{\hlhomom(\hlcond).\hltok\with\tup{\hlcond,\hltok}\in\hlcut_n}=\marking $.

	Let now $ \hltrans.\hlmode\in\transs $ s.t.\ $ \marking[\hltrans.\hlmode\rangle $.
	Then $ \hlmarking[\hltrans,\hlmode\rangle $.
	Since $ \hlconf $, $ \insta $ satisfy property b) from Def.~\ref{def:ComplPref},
	we know that $ \exists \hlevnt\in\hlevnts $ s.t.\ $ \hlevnt\notin\hlconf $,
	$ \hlhomom(\hlevnt)=\hltrans $, and $ \hlconf,\insta[\hlevnt,\hlmode\rangle $.
	This means, in $ \expf(\beta) $, we have
	$ \ms{\hlcond.\hltok\with\tup{\hlcond,\hltok}\in\hlcut_n}[\hlevnt.\hlmode\rangle $.
	Thus, there exists an $ \evnt $ in $ \unf(\expf(\beta)) $ such that
	$ \conf[\evnt\rangle $ and $ \llhomom^{\mathrm{O}}(\evnt)=\hlevnt.\hlmode $,
	which again means that $ \homom(\evnt)=\hlhomom(\hlevnt).\hlmode =\hltrans.\hlmode$.
	This proves that $ \expof(\beta) $ is complete.

	The other direction works analogously.
\end{proof}

We can now compare the two complete finite prefixes resulting from the original ERV-algorithm from \cite{EsparzaRV02} applied to $ \expf(\hlpn) $ and the generalized ERV-algorithm Alg.~\ref{alg:complpref} applied to $ \hlpn\in\Nf $.
From the definition of the generalized cut-off criterion
we get that both these prefixes have the same depth.
However, due to the high-level representation, the breadth of the symbolic prefix can be substantially~smaller.
This is the case for our running example:

\begin{example}
	\begin{figure}[!htb]
		\centering
		\includegraphics[scale=1.25]{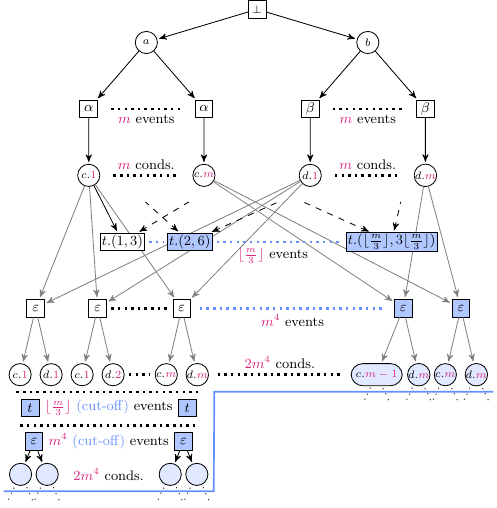}
		\caption{The complete finite prefix of $ \llunf(\expf(\hlpn)) $ of the running example $ \hlpn $ from Figure~\ref{fig:runEx} calculated by the original ERV-Algorithm.
			\label{fig:RunExExp}
		}
	\end{figure}
	Consider again $ \hlpn\in\Nf $ from Figure~\ref{fig:runEx} with $ \hltoks=\{ 0,1,\dots,m \} $ for a fixed $ m>0 $.
	The finite complete prefix of  $ \llunf(\expf(\hlpn)) $ is depicted in Fig.~\ref{fig:RunExExp}.
	Instead of giving each event/condition a distinct name, we indicated the label of each node inside of it.
	For events with label $ \varepsilon $ we even omitted the mode, since it is derivable from the connected conditions.
	Cut-off events and their output conditions are again shaded blue, and the blue line indicates the complete finite prefix resulting from the original ERV-algorithm.

	After firing an instance of $ \alpha $
	and an instance of $ \beta $,
	we arrive at conditions with labels $ c.k $ and $ d.\ell $.
	If these satisfy $ \ell=k\cdot 3 $ then
	we can fire an instance of $ t $, which means we have $ \lfloor\frac{m}{3}\rfloor $ such events.
	Only the first is no cut-off event,
	since their configurations all represent the same (empty) marking.
	This empty marking is however also the reason why we cannot continue even after the first instance of~$ t $.

	For each combination of an output condition of an instance of $ \alpha $ (the $ m $ instances of $ c $) and an output condition of an instance of $ \beta $ (the $ m $ instances of $ d $),
	we have $ m^2 $ possibilities to fire an instance of $ \varepsilon $.
	The reason for this is that in the high-level net $ \hlpn $,
	the output variables are independent of the input variables.
	This leads to $ m^4 $ many $ \varepsilon $-events of depth $ 2 $.

	All except the first $ m^2 $ of those are cut-off events.
	After the non cut-off events, however, we have to repeat the part from above for the ERV-algorithm to terminate.
	All in all the complete finite prefix contains $ 6m^4+4m+ 2\lfloor \frac{m}{3} \rfloor + 2 $ nodes for every fixed $ m $
	in the color class $ \hltoks=\{0,1,\dots,m \} $.
	The complete finite prefix of the \emph{symbolic} unfolding $ \symbunf(\hlpn) $
	that is shown in Figure~\ref{fig:runExUnf},
	on the other hand has the same number of nodes for every $ m $.
\end{example}

Generalizing this example to a family of nets gives the following proposition:
\begin{proposition}
	For every $ n \in\N$, there is a family $ (\hlpn_m^n)_{m\in\N^{>0}}$ of high-level nets  in $\hlpn_m^n\in\Nf $ such that every $ \hlpn_m^n $ has the set of colors $ \hltoks=\{0,\dots,m\} $ and satisfies that
	\begin{itemize}
		\item the complete finite prefix of $ \symbunf(\hlpn_m^n) $ obtained by Alg.~\ref{alg:complpref} has the same number of nodes for every~$ m $,
		\item the number of nodes in the low-level prefix of $ \llunf(\expf(\hlpn_m^n)) $ obtained by the original ERV-algorithm is greater than $ m^n $.
	\end{itemize}
\end{proposition}
In particular, the benchmark family Fork And Join, presented in Sec.~\ref{bm:diamond}
satisfies this property.

\section{Handling infinitely many reachable markings}\label{sec:InfinteCase}

Unfoldings of unbounded P/T Petri nets (i.e., with infinitely many markings)
have been investigated in \cite{AbdullaIN00,DeselJN04},
and in \cite{HerbreteauST07} concurrent well-structured transition systems with infinite state space
are unfolded.
When applying the generalized ERV-algorithm, Alg.~\ref{alg:complpref},
to high-level Petri nets with infinitely many reachable markings
(therefore violating \itmFinite\ from the definition of $ \Nf $),
the proof for finiteness of the resulting prefix does not hold anymore:
the proof of Prop.~\ref{prop:FinIsFin}, step~(1),
is a generalization of the proof of the respective claim in \cite{EsparzaRV02} (which uses the pigeonhole principle).
It is argued that we cannot have $ |\reachable(\hlpn)|+1 $ consecutive events
s.t.\ their cone configurations each generate a marking in the net not seen before,
and we thus have a cut-off event.
When we deal with infinitely many markings, this argument cannot be made.

In this section, we introduce a class $ \Nsc $ of safe high-level nets, called symbolically compact,
that have possibly infinitely many reachable markings (and therefore an infinite expansion),
generalizing the class $ \Nf $.
We then proceed to make adaptions to Alg.~\ref{alg:complpref} (i.e., to the used cut-off criterion),
so that it generates a finite and complete prefix of the symbolic unfolding for any $ \hlpn\in\Nsc $.

The following Lemma precisely describes the finite high-level Petri nets for which a finite and complete prefix of the symbolic unfolding exists.
They are characterized by having a bound on the number of steps needed to arrive at every reachable marking.
For the proof we argue that in the case of such a bound,
the symbolic unfolding up to depth $ n+1 $ is a finite and complete prefix,
and that in the absence of such a bound no depth of a prefix suffices for it to be complete.
\label{sec:SecOfPrefixExistenceCond}
\begin{lemma}\label{lem:PrefixExistenceCond}
	For a finite high-level Petri net $ \hlpn=\tup{\hlns,\hlmarkings_0} $
	there exists a finite and complete prefix of $ \symbunf(\hlpn) $ if and only if there exists a bound $ n\in\N $
	such that every marking in $ \reachable(\hlpn) $ is reachable from a marking in $ \hlmarkings_0 $ by firing at most $ n $ transitions.
\end{lemma}
\begin{proof}
	From Prop.~\ref{prop:MarkingsRepresented} and Prop.~\ref{prop:unfComplete} we see that for a finite high-level Petri net with such a bound $ n $, the prefix of the symbolic unfolding containing exactly the events $ \hlevnt $ with $ \depth(\hlevnt)\leq n+1 $ is complete.
	Finiteness of this prefix follows from the finiteness of the original net and the definition of homomorphism.

	Assume now that no such bound exists, and, for the purpose of contradiction, assume that there is a finite and complete prefix $ \beta $ of $ \symbunf(\hlpn) $.
	Denote $\widetilde{n}= \max\{|\hlconf|\with\hlconf\in\hlconfs(\beta) \}<\infty $.
	Then there exists a marking~$ \hlmarking\in\reachable(\hlpn) $ for which we have to fire at least $ \widetilde{n}+1 $ transitions to reach it.
	Again from Prop.~\ref{prop:MarkingsRepresented} and Prop.~\ref{prop:unfComplete} it follows that a configuration $ \hlconf $ with $ \hlmarking\in\confmarks(\hlconf) $ must contain at least $ \widetilde{n}+1 $ events, contradicting that $ \beta $ is complete.
\end{proof}

\subsection{Symbolically compact high-level Petri nets}
We use the result of Lemma~\ref{lem:PrefixExistenceCond} to define the class $ \Nsc $ of high-level nets
for which we adapt the algorithm for constructing finite and complete prefixes of the symbolic unfolding.
\begin{definition}[Class $ \Nsc $]\label{def:Nsc}
	A finite high-level Petri net $ \hlpn $ is called \emph{\textbf{s}ymbolically \textbf{c}ompact} if it satisfies
	\itmSafety\ and \itmGuards\ from Def.~\ref{def:Nf}, and
	\begin{itemize}[label=\textbf{(3}*\textbf{)},leftmargin=*,topsep=2pt]
		\item There is a bound $ n\in\N $ on the number of transition firings needed to reach all markings in~$ \reachable(\hlpn) $.
	\end{itemize}
	We denote the class containing all symbolically compact high-level Petri nets by $ \Nsc $.
\end{definition}
\newcommand{\itmUniform}{\textbf{(3}*\textbf{)}}

Note that in the case of a (finite, safe) P/T net, property \itmUniform\ is equivalent to~\itmFinite\ (i.e., $ |\reachable(\hlpn)|<\infty $).
However, this is \emph{not} true for all high-level nets~$ \hlpn $:
while $ |\reachable(\hlpn)|<\infty $ still implies \itmUniform\ (meaning $ \Nf\subseteq\Nsc $), the reverse implication does not hold,
as our running example from Figure~\ref{fig:runEx} demonstrates when we change the set of colors to $ \hltoks=\N $: it then still satisfies \itmSafety\ and \itmGuards, with
$ 	\reachable(\hlpn)=\{ \ms{ \tup{a,0},\tup{b,0} }, \ms{}  \}\cup \{ \ms{ \tup{c,n},\tup{b,0} },\ms{ \tup{a,0},\tup{d,n'} },\allowbreak\ms{ \tup{c,n},\tup{d,n'} }\with n,n'\in\N \}. $
So we have infinitely many markings that can all be reached by firing at most two transitions,
meaning the net satisfies~\itmUniform\ and is therefore symbolically compact.

Lemma~\ref{lem:PrefixExistenceCond} implies that the class $ \Nsc $ of symbolically compact nets contains exactly all high-level Petri nets satisfying \itmSafety\ and \itmGuards\ for which a finite and complete prefix of the symbolic unfolding exists (independently of whether the number of reachable markings is finite).
Since the reachable markings of a high-level Petri net and its expansion correspond to each other,
this observation leads to an interesting subclass~$ \Nsc\setminus\Nf $ of symbolically compact high-level Petri nets that have infinitely many reachable markings.
For every net $ \hlpn $ in this subclass
\begin{itemize}
	\item there exists a finite and complete prefix of $ \symbunf(\hlpn) $,\hspace{2.6cm} but
	\item there does \emph{not} exist a finite and complete prefix of $ \symbunf(\expf(\hlpn)) $.
\end{itemize}
In particular, the original ERV-algorithm cannot be applied to $ \expf(\hlpn) $, since this expansion is an infinite net.

An example for such a net is our running example from Figure~\ref{fig:runEx} when we replace the color class $ \hltoks=\{ 0,1,\dots,m \} $ by $ \hltoks=\N $.
Much simpler is the following net, also with $ \hltoks=\N $:
\begin{equation*}
	\includegraphics{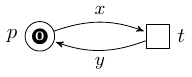}
\end{equation*}
Obviously, every reachable marking $ \ms{\tup{\hlplace,n}} $ with $ n\in\N $ can be reached by firing $ \hltrans $ one time in mode $ \{ x\gets 0,y\gets n \} $, so the net is symbolically compact.
The expansion of this net however is infinite,
and the original ERV-algorithm does not terminate when applied to it.

\subsection{Insufficiency of the cut-off criterion for \texorpdfstring{$ \Nsc $}{}}\label{sec:Nsc-vs-Nf}
Naturally, the question arises whether the generalized ERV-Algorithm, Alg.~\ref{alg:complpref}, also yields a finite and complete prefix of a symbolically compact input net.
For many examples (like the simple one above) this is the case.
However, there are symbolically compact high-level Petri nets for which  Alg.~\ref{alg:complpref} does \emph{not} terminate.

\begin{figure}[!htb]
\vspace*{-2mm}
	\begin{subfigure}[t]{\linewidth}
		\centering
		\includegraphics{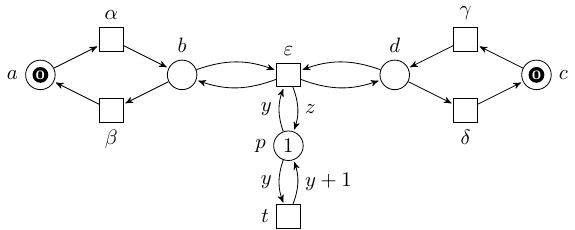}
		\caption{A symbolically compact net with $ \hltoks=\N $.\label{fig:ex-Nsc-vs-Nf}}
	\end{subfigure}\\[5mm]
	\begin{subfigure}[t]{\linewidth}
		\centering
		\includegraphics[width=\textwidth]{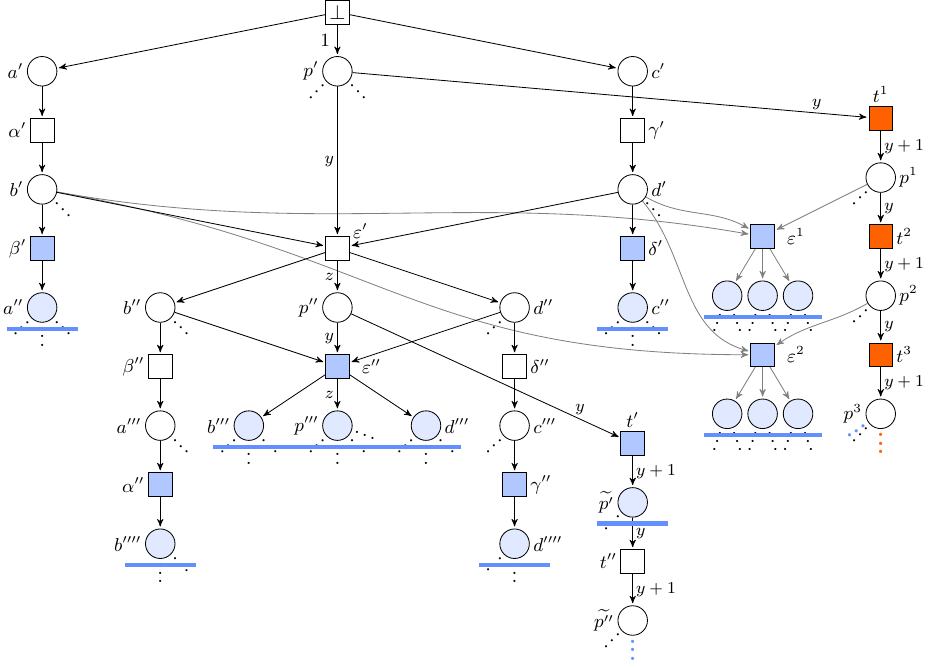}
		\caption{A prefix of the symbolic unfolding of the net in \subref{fig:ex-Nsc-vs-Nf}.\label{fig:ex-Nsc-vs-Nf-unf}}
	\end{subfigure}\vspace*{-3mm}
	\caption{A symbolically compact net in \subref{fig:ex-Nsc-vs-Nf} where Alg.~\ref{alg:complpref}, trying to build a complete finite prefix of the symbolic unfolding shown in \subref{fig:ex-Nsc-vs-Nf-unf}, does not terminate.}\vspace*{-6mm}
\end{figure}

\medskip
The criterion for nontermination of Alg.~\ref{alg:complpref} is that in the unfolding of the net, there is an
infinite sequence of cone configurations $ \cone{\hlevnt_1},\cone{\hlevnt_2},\dots $
with $ \cone{\hlevnt_1}\ado\cone{\hlevnt_2}\ado\dots $ such that
\begin{itemize}
	\item
	$ \forall i\in\N_{>0}: \{ \hlevnt_1,\dots,\hlevnt_i \}\in\hlconfs(\symbunf) $,
	i.e., $ \tup{\hlevnt_1,\hlevnt_2,\dots} $ is a fireable sequence in the symbolic unfolding,
	\item $ \forall i\in\N_{>0}:  \confmarks(\cone{\hlevnt_i})\not\subseteq \bigcup_{\cone{\hlevnt'}\ado\cone{\hlevnt_i}}\confmarks(\cone{\hlevnt'}) $,
	i.e., no event $ \hlevnt_i $ is a cut-off event.
\end{itemize}
Note that in the second condition,
the $ \hlevnt' $ in the union are arbitrary events in the unfolding, and not restricted to the sequence $ \tup{\hlevnt_1,\hlevnt_2,\dots} $.
An example for such a net is shown in Figure~\ref{fig:ex-Nsc-vs-Nf}:

The set of colors is given by $ \hltoks=\N $.
Initially there is a token $ 0 $ in each of the places~$ a $ and $ c $.
The token on $ a $ can cycle between $ a $ and $ b $ by transitions~$ \alpha $ and $ \beta $.
Analogously, the other token can cycle between $ c $ and $ d $ by $ \gamma $ and $ \delta $.
Additionally, in the inital marking, there is a color~$ 1 $ on place~$ p $.
This number can be increased by $ 1 $ by firing $ t $.
Thus, every number~$ n $ can be placed on $ p $ by $ n-1 $ firings of $ t $.
When, however, the two cycling tokens of color $ 0 $ are on places $ b $ and $ d $,
an arbitrary number $ n $ can be placed directly on $ p $ by firing $ \varepsilon $.
The net therefore is symbolically compact.

Examine now the unfolding in Figure~\ref{fig:ex-Nsc-vs-Nf-unf}.
Cut-off events and their output conditions are again shaded blue.
For a cleaner presentation do not write the local predicate next to each event.
For the event $ \varepsilon' $ we have
$ \confmarks(\cone{\varepsilon'})=\{ \ms{ (b,0),\tup{d,0},(p,n) }\with n\in\N \} $.
This means by firing no event, only $ \beta'' $, only $ \delta'' $, or both $ \beta'' $ and $ \delta'' $ from the corresponding cut,
we can represent every reachable marking in the net.

The sequence $ \cone{t^1},\cone{t^2},\dots $ of cone configurations,
with the corresponding events shaded orange,
now satisfies the criterion from above:
The first condition
is obviously satisfied.
The sequence of events corresponds to firing $ t $ infinitely often,
always increasing the number on $ p $ by $ 1 $.
The cones $ \cone{t^i} $ are the only cone configurations
where the cuts represent markings with \emph{no tokens} on the two places $ b $ and $ d $.
For all other cones in the unfolding, there is a $ 0 $ on $ b $ and/or a $ 0 $ on $ d $.
Thus, no event $ t^i $ is a cut-off event.
This means if Alg.~\ref{alg:complpref} is applied to the net in Figure~\ref{fig:ex-Nsc-vs-Nf} it does not terminate, building a prefix containing every $ t^i $ with $ i\in\N^+ $.

We have now shown that there are nets in $ \Nsc\setminus\Nf $ where Alg.~\ref{alg:complpref}
does not terminate due to the insufficient cut-off criterion.
A compelling avenue for further research lies in exploring the existence of subclasses $ \hlpncls $ between $ \Nf $ and $ \Nsc $ (i.e., $ \Nf\subset \hlpncls \subset \Nsc $) for which the cut-off criterion suffices and the algorithm terminates.

\subsection{The finite prefix algorithm for symbolically compact nets}\label{sec:AlgoSymbComp}
As previously discussed, the argument that states the existence of one event in a chain of $ |\reachable(\hlpn)|+1 $ consecutive events, such that every marking represented by its cone configuration is contained in the union of all markings represented by previous cone configurations, cannot be applied in the case of an infinite number of reachable markings. Consequently, Alg.~\ref{alg:complpref} may not terminate when applied to a net in $ \Nsc\setminus\Nf $. However, condition \itmUniform\ guarantees that every marking reached by a cone configuration $ \cone{\hlevnt} $ with depth $ >n $ can be reached by a configuration
$ \hlconf $ containing no more than $ n $ events.

For the algorithm to terminate, we need to adjust the cut-off criterion since we do not know whether~$ \hlconf $ is also a cone configuration, as demanded in Def.~\ref{def:cut-off}. Therefore, we define \emph{cut-off* events}, that generalize cut-off events. They only require that every marking in $ \confmarks(\cone{\hlevnt}) $ has been observed in a set $ \confmarks(\hlconf) $ for \emph{any} configuration $ \hlconf\prec \cone{\hlevnt} $, rather than just considering cone configurations:
\begin{definition}[Cut-off* event]
	Under the assumptions of Def.~\ref{def:cut-off},
	the high-level event $ \hlevnt $ is a \emph{cut-off*} event (w.r.t.\ $ \ado $) if
	$ \confmarks(\cone{\hlevnt})\subseteq \bigcup_{\hlconf\ado\cone{\hlevnt}}\confmarks(\hlconf). $
\end{definition}
We additionally assume that the used adequate order satisfies  $ |\hlconf_1|<|\hlconf_2|\Rightarrow \hlconf_1\ado\hlconf_2 $,
so that every event with depth $ >n $ will be a cut-off event.
Since all adequate orders discussed in \cite{EsparzaRV02} satisfy this this property (cp. App.~\ref{app:adequateOrders}), this is a reasonable requirement.
This adaption and assumption now lead to:
\begin{theorem}
	Assume a given adequate order $ \prec $ to satisfy $ |\hlconf_1|<|\hlconf_2|\Rightarrow {\hlconf_1\prec \hlconf_2} $.
	When replacing in Alg.~\ref{alg:complpref} the term ``cut-off event'' by ``cut-off* event'',
	it terminates for any input net $ \hlpn\in\Nsc $, and generates a complete finite prefix of $ \symbunf(\hlpn) $.
\end{theorem}
\begin{proof}
	The properties of symbolic unfoldings that we stated in Sec.~\ref{sec:SymbUnfProps}
	are independent on the class of high-level nets.
	Def.~\ref{def:uniquemirror} only uses that the considered net is safe, and so do Prop.~\ref{prop:futureIsoWeak2} and Prop.~\ref{prop:adequateOrderOtherDir}.
	We therefore only have to check that the correctness proof for the algorithm still holds.
	In the proof of Prop.~\ref{prop:FinIsFin} ($ \fin $ is finite),
	the steps~(2) and (3) are independent of the used cut-off criterion.
	In step~(1), however, it is shown that the depth of events never exceeds $ |\reachable(\hlpn)| +1$. This is not applicable when $ |\reachable(\hlpn)|=\infty $, as argued above.
	Instead	we show:

	\begin{itemize}[label=(1*),leftmargin=*,topsep=2pt]
		\item
		For every event $ \hlevnt $ of $ \fin $, $ \depth(\hlevnt)\leq n+1 $,
		where $ n $ is the bound on the number of transitions needed to reach all markings in $ \reachable(\hlpn) $.
	\end{itemize}

\noindent 	In the proof of Prop.~\ref{prop:FinIsCompl},
	the cut-off criterion is used to show (by an infinite descent approach),
	for any marking $ \hlmarking\in\reachable(\hlpn) $
	the existence of a minimal configuration $ \hlconf\in\fin $ with $ \hlmarking\in\confmarks(\hlconf) $.
	Due to the similarity of cut-off and cut-off*,
	this proof can easily be adapted to work as before:

	Assume that at some point during the algorithm,
	we reach a state $ (\hlconds',\hlevnts',\hlarcs',\hlguard',\hlcuts_0') $ of the prefix under construction,
	such that
	there occurs a chain of events $ \hlevnt_1<\hlevnt_2<\dots<\hlevnt_{n+1} $.
	We prove that $ \hlevnt_{n+1} $ must be a cut-off* event.
	Let $ \hlmarking\in\confmarks(\cone{\hlevnt_{n+1}}) $.
	Then, by definition of $ \Nsc $,
	$ \hlmarking $ can be reached by firing at most $ n $ transitions.
	Accordingly, from Prop.~\ref{prop:unfComplete},
	we get that there is a configuration $ \hlconf\in\symbunf $
	containing at most $ n $ events
	such that $ \hlmarking\in\confmarks(\hlconf) $.
	As in the proof of Prop.~\ref{prop:FinIsCompl},
	we can now follow that there is a configuration $ \widetilde{\hlconf}\in\hlconfs(\symbunf) $
	such that $ \hlmarking\in\confmarks(\widetilde{\hlconf}) $ and
	$ \widetilde{\hlconf}\prec\hlconf $, that contains no cut-off event
	and is therefore in $ \fin $.
	Since $ |\hlconf|\leq n < n+1\leq |\cone{\hlevnt_{n+1}}| $,
	we follow $ \widetilde{\hlconf}\prec \cone{\hlevnt_{n+1}} $.
	So we have that $ \forall \hlmarking\in\confmarks(\cone{\hlevnt_{n+1}})\, \exists\widetilde{\hlconf}\prec\cone{\hlevnt_{n+1}} $, which means that $ \hlevnt_{n+1} $ is a cut-off* event.
	This proves that $ \fin $ is finite.

	It remains to show termination.
	In the case of nets in $ \Nf $,
	every object is finite, which, together with Prop.~\ref{prop:FinIsFin}, leads to termination of the algorithm.
	For nets in $ \Nsc\setminus\Nf $, however, there is at least one event $ \hlevnt $ in $ \fin $ s.t.\ $ |\confmarks(\cone{\hlevnt})|=\infty $.
	Thus, we have to show that we can check the cut-off* criterion in finite time.
	This follows from Cor.~\ref{cor:CheckCutOffs} in the next section,
	which is dedicated to symbolically representing markings generated by configurations.
\end{proof}

\subsection{Feasibility of symbolically compact nets and cut-off*}\label{sec:feasibility}
To check the cut-off* criterion of an event added to a prefix of the unfolding,
we have to compare the set of markings represented by the cut of the event's cone configuration to
\emph{all} markings represented by cuts of smaller configurations.
This means that we possibly have to store the whole state space.

This realization gives rise to two questions.
Firstly, how do we manage the storage of an infinite number of markings?
This query is addressed in Sec.~\ref{sec:constraints}, where we demonstrate how to symbolically represent the markings represented by a configuration's cut
and how to check the cut-off* criterion within finite time.
The prototype implementation outlined in Sec.~\ref{sec:impl-det} utilizes these methods for the $ \Nf $-case.

The second question that arises asks how the complete finite prefix resulting from the generalized ERV-algorithm with the cut-off* criterion relates to a reachability graph -- both in terms of size and computation time.
However, as symbolically compact nets
possibly have an infinite expansion, the reachability graph can be infinitely broad. Thus,
at present this method provides a more viable solution compared to calculating the
infinite reachability graph. However, we give an outlook on how a (finite) symbolic
reachability tree of a symbolically compact net could possibly be constructed.
\paragraph{Outlook: Symbolic Reachability Trees of Symbolically Compact Nets.}
The idea of a symbolic reachability tree has been realized for algebraic Petri nets in \cite{Schmidt1995} by Karsten Wolf.
However, in contrast to this work, we think that for the class of symbolically compact nets
we can build a symbolic reachability tree that is complete.

The idea is to gradually extend for every subset $ \hlplaces' $ of places a formula $ \reachable_{\hlplaces'} $
that symbolically describes all reachable markings that we have seen so far and have colors on exactly all places in~$ \hlplaces' $.
Initially, all formulae are $ \false $, except for $ \reachable_{\hlplaces_0} $, where $ \hlplaces_0 $ are the initially marked places.
$ \reachable_{\hlplaces_0} $ symbolically represents the set of initial markings.

The symbolic reachability tree is then constructed by starting with a root $ n_0 $ labeled with $ f_{n_0}=\reachable_{\hlplaces_0} $
representing the set of initial markings.
For every transition $ t $, we can determine whether $ t $ can fire in any mode from any marking represented by $ f_{n_0} $
by a satisfiability check.
If $ t $ can fire, we add a new node $ n' $, and label it with a formula $ f' $ that symbolically represents all markings reached from firing $ t $ in any mode from any marking in $ \hlmarkings_0 $.
In all these markings, there are colors on the same set of places $ \hlplaces' $.
If $ f'\Rightarrow\reachable_{\hlplaces'} $ then we end this branch.
We then extend $ \reachable_{\hlplaces'} $ to $ \reachable_{\hlplaces'}\lor f' $.

By repeating this procedure in breadth-first order, we build a tree that symbolically represents all reachable markings.
This tree should correspond
to the complete finite prefix of the symbolic unfolding of the net to which you
added a shared resource (in form of a new place) that every transition consumes and
recreates. This makes the system sequential
without changing the sequential semantics.
We give here only the idea of the tree, and not a formal definition.
In future work we want to further investigate on this idea.
We can then compare the complete finite prefix of the symbolic unfolding to the symbolic reachability tree.

\section{Checking cut-offs symbolically}\label{sec:constraints}

We show how to check whether a high-level event $ \hlevnt $ is a cut-off* event symbolically in finite time.
By definition, this means checking whether
$
\confmarks(\cone{\hlevnt})\subseteq \bigcup_{\hlconf\ado\cone{\hlevnt}}\confmarks(\hlconf)$.
However, since the cut of a configuration can represent infinitely many markings, when applying the adapted algorithm we cannot simply store the set $ \confmarks(\hlconf) $ for every $ \hlconf\in\hlconfs(\fin) $.
Instead, we now define constraints that symbolically describe the markings represented by a configuration's cut.
Checking the inclusion above then reduces to checking an implication of these constraints.
Since we consider high-level Petri nets with guards written in a decidable theory,
such implications can be checked in finite time.

At the end we see that this method can be easily adapted to symbolically check whether,
in a prefix of the symbolic unfolding of a net $ \hlpn\in\Nf $,
an event $ \hlevnt $ is a cut-off event in the sense of Def.~\ref{def:cut-off}.
This method is also used in the implementation described in Sec.~\ref{sec:impl-det}.

For the rest of this section, let $\hlpn=\tup{\hlplaces,\hltranss,\hlflowfunc,\hlguard,\hlmarkings_0}\in\Nsc $
with symbolic unfolding $ \symbunf(\hlpn)=\tup{\hlunf,\symbhomom}=\tup{\hlconds,\hlevnts,\hlarcs,\hlguard,\hlcuts_0,\symbhomom} $.

\medskip
We first define for every condition $ \hlcond $ a new predicate $ \predb(\hlcond) $ by
\begin{equation*}
	\predb({\hlcond}):=\pred({\hlcondevnt(\hlcond)})\land (\symbhomom(\hlcond)=\hlcondvar_\hlcondevnt(\hlcond)).
\end{equation*}

This predicate now has (in an abuse of notation) an extra variable, named after its label $ \symbhomom(\hlcond) $.
The remaining variables in $ \predb({\hlcond}) $ are coming from $ \pred({\hlcondevnt(\hlcond)}) $ and given by $ \hlvars_{\cone{\hlcondevnt(\hlcond)}\cup\{\bot\}} $.
As we know, $ \pred({\hlcondevnt(\hlcond)}) $ evaluates to $ \true $ under an assignment
$ \insta: \hlvars_{\cone{\hlcondevnt(\hlcond)}\cup\{\bot\}}\to\hltoks $ if and only if
a concurrent execution of $ \cone{\hlcondevnt(\hlcond)} $ with the assigned modes is possible (i.e., under every instantiation of $ \cone{\hlcondevnt(\hlcond)} $).
In such an execution, $ \insta(\hlcondvar_\hlcondevnt(\hlcond))\in\hltoks $ is placed on~$ \hlcond $.
The predicate $ \predb({\hlcond}) $ therefore can only be true if $ \symbhomom(\hlcond) $ is assigned a color that can be placed on $ \hlcond $.

\medskip
We now define for a co-set $ \hlconds'\subseteq\hlconds $ of high-level conditions
the constraint on~$ \hlconds' $ as an expression with free variables~$ \symbhomom(\hlconds')=\{\symbhomom(\hlcond)\with\hlcond\in\hlconds'\} $,
describing which color combinations can lie on the places represented by the high-level conditions.
We build the conjunction over all predicates $ \predb(\hlcond) $ for $ \hlcond\in\hlconds' $
and quantify over all appearing variables $ \hlvar_\hlevnt $:
the \emph{constraint on $ \hlconds' $} is defined by
\begin{equation*}
	\constr{\hlconds'}:=
	\exists_{\bigcup_{\hlcond\in\hlconds'}\hlvars_{\cone{\hlcondevnt(\hlcond)}\cup\{\bot\}}}:
	\bigwedge_{\hlcond\in\hlconds'}\predb({\hlcond}).
\end{equation*}
We denote by $ \asss(\hlconds') $ the set of variable assignments $ \ass:\symbhomom(\hlconds')\to\hltoks $
that satisfy $ \constr{\hlconds'}[\ass]\equiv\true $.

\medskip
For a configuration $ \hlconf $,
we have that $ \hlconds'=\confcut(\hlconf) $ is a co-set, $ \symbhomom(\hlconds')=\symbhomom(\confcut(\hlconf)) $
describes the set of places occupied in every marking in $ \confmarks(\hlconf) $.
Note that in this case,
we have $ \bigcup_{\hlcond\in\confcut(\hlconf)}\hlvars_{\cone{\hlcondevnt(\hlcond)}}=\hlvars_\hlconf $,
i.e., the bounded variables in $ \constr{\confcut(\hlconf)} $ are exactly the variables appearing in predicates in~$ \hlconf $.

For every instantiation $ \insta $ of $ \hlconf $ we
define a variable assignment $  \ass_\insta:\symbhomom(\confcut(\hlconf))\to\hltoks $ by setting
$ \forall\symbhomom(\hlcond)\in\symbhomom(\confcut(\hlconf)): \ass_\insta(\hlcond)=\insta(\hlcondvar_\hlcondevnt(\hlcond)) $.
Instantiations of a configuration and the constraint on its cut are related as follows.
\begin{lemma}\label{lem:constraintsVSinsts}
	Let $ \hlconf\in\hlconfs(\symbunf(\hlpn)) $. Then
	$ \asss(\confcut(\hlconf))=
	\{ \ass_\insta \with \insta\in\instas(\hlconf)\} $.
\end{lemma}
\begin{proof}
	The proof follows by construction of $ \predb $ and~$ \ass_\insta $:
	Let $ \ass\in\asss(\confcut(\hlconf)) $.
	Then $ \true\equiv\constr{\confcut(\hlconf)}[\ass] $.
	Thus, there exists $\insta:\hlvars_{\hlconf\cup\{\bot\}}\to\hltoks$ s.t.\
	$
	\big(\bigwedge_{\hlcond\in\confcut(\hlconf)}\predb({\hlcond})\big)[\ass][\insta]\equiv \true$
	and therefore
	\begin{align}\label{eq:constraintsVSinsts}
		\Big(\bigwedge_{\hlcond\in\confcut(\hlconf)}\pred(\hlcondevnt({\hlcond}))[\insta]\Big)
		\ \land\
		\Big(\bigwedge_{\hlcond\in\confcut(\hlconf)}\ass(\symbhomom(\hlcond))=\insta(\hlcondvar_\hlcondevnt(\hlcond))\Big)\quad&\equiv\quad \true.
	\end{align}
	From
	the inductive definition of $ \pred $ then follows that
	$ \forall \hlevnt\in\hlconf\cup\{\bot\}:\pred({\hlevnt})[\insta]\equiv\true $.
	Thus, $ \insta $
	is an instantiation of $ \hlconf $,
	and $ \ass_\insta=\ass $, as shown by the posterior conjunction in \eqref{eq:constraintsVSinsts}.

\medskip
	Let on the other hand $ \insta\in\instas(\confcut(\hlconf)) $.
	Then directly, by the definition of $ \predb(\hlcond) $ and $ \ass_\insta $,
	we get
	$
	\big(
	\bigwedge_{\hlcond\in\confcut(\hlconf)}\predb({\hlcond})
	\big)
	[\ass_\insta][\insta]\equiv\true
	$
	and by the definition of $ \constr{\confcut(\hlconf)} $ that
	$ \constr{\confcut(\hlconf)}[\ass_\insta]\equiv\true $, i.e., $ \ass_\insta\in\asss(\confcut(\hlconf))$.
\end{proof}

From the definition of $ \confcuts(\hlconf) $ and $ \hlmarkings(\hlconf) $ we get:
\begin{corollary}\label{cor:CutAndMarkChar}
	Let $ \hlconf\in\hlconfs(\symbunf(\hlpn)) $. Then
	$ \confcuts(\hlconf)
	=\{\{ \tup{\hlcond,\ass(\symbhomom(\hlcond))}\with \hlcond\in\confcut(\hlconf)
	\}\with \ass\in\asss(\confcut(\hlconf)) \} $
	and
	$ \confmarks(\hlconf)
	=\{ \ms{\tup{\symbhomom(\hlcond),\ass(\symbhomom(\hlcond))}\with \hlcond\in\confcut(\hlconf)}\with\ass\in\asss(\confcut(\hlconf)) \} $.
\end{corollary}
We now show how to check
whether an event is a cut-off* event
via the constraints defined above.
For that, we first look at general configurations in Theorem~\ref{thm:constraintsDisj},
and then explicitly apply this result to cone configurations $ \cone{\hlevnt} $ in Corollary~\ref{cor:CheckCutOffs}.

\begin{theorem}\label{thm:constraintsDisj}
	Let $ \hlconf,\hlconf_1,\dots,\hlconf_n $ be finite configurations in the symbolic unfolding of a safe high-level Petri net s.t.\ $\forall 1\leq i\leq n: \symbhomom(\confcut(\hlconf))=\symbhomom(\confcut(\hlconf_i)) $.
	Then
	\begin{equation*}
		\confmarks(\hlconf)\subseteq \bigcup_{i=1}^n \confmarks(\hlconf_i)
		\quad\text{if and only if}\quad
		\constr{\confcut(\hlconf)}\Rightarrow	\bigvee_{i=1}^n\constr{\confcut(\hlconf_i)}.
	\end{equation*}
\end{theorem}
\begin{proof}
	Assume
	$ \confmarks(\hlconf)\subseteq \bigcup_{i=1}^n \confmarks(\hlconf_i) $ and let $ \ass\in\asss(\confcut(\hlconf)) $.
	We have that
	$ \hlmarking_\ass:= \{ \tup{\symbhomom(\hlcond),\ass(\symbhomom(\hlcond))} \with \hlcond\in\confcut(\hlconf) \}\in\confmarks(\hlconf) $ by Cor.~\ref{cor:CutAndMarkChar}.
	Thus, $ \exists 1\leq i \leq n: \hlmarking_\ass\in\confmarks(\hlconf_i) $.
	This, again by Cor.~\ref{cor:CutAndMarkChar}, means
	$ \exists \ass_i\in\asss(\confcut(\hlconf_i)): $
	\begin{equation*}
		\hlmarking_\ass
		=\{ \tup{\symbhomom(\hlcond'),\ass_i(\symbhomom(\hlcond'))}\with \hlcond'\in\confcut(\hlconf_i) \}
	\end{equation*}
	This shows that $ \ass=\ass_i $.
	Thus,
	$ \constr{\confcut(\hlconf_i)}[\ass]\equiv\true $,
	which proves the implication.

\medskip
	Assume on the other hand $ \constr{\confcut(\hlconf)}\Rightarrow 	\bigvee_{i=1}^n\constr{\confcut(\hlconf_i)} $.
	Let $ \hlmarking\in\confmarks(\hlconf) $.
	Then $ \exists \ass\in\asss(\confcut(\hlconf)):
	\hlmarking=\{ \tup{\symbhomom(\hlcond),\ass(\symbhomom(\hlcond))}\with \hlcond\in\confcut(\hlconf) \} $.
	Thus, $ {\exists 1\leq i\leq n}: \constr{\confcut(\hlconf_i)}[\ass]\equiv\true $.
	Let $ \ass_i=\ass $.
	Then $ \ass_i\in\asss(\confcut(\hlconf_i)) $, and
	$ \hlmarking=\{ \tup{\symbhomom(\hlcond'),\ass_i(\symbhomom(\hlcond'))} \with \hlcond'\in\confcut(\hlconf_i)\}\in\confmarks(\hlconf_i) $,
	which completes the proof.
\end{proof}
The following Corollary now gives us a characterization of cut-off* events in a symbolic branching process.
It follows from Theorem~\ref{thm:constraintsDisj} together with the facts that
$ \confmarks(\hlconf_1)\cap\confmarks(\hlconf_2)\neq\emptyset\Rightarrow\symbhomom(\confcut(\hlconf_1))=\symbhomom(\confcut(\hlconf_2)) $,
and that $ {{\ado}\cone{\hlevnt}} $ is finite.
\begin{corollary}\label{cor:CheckCutOffs}
	Let $ \beta $ be a symbolic branching process
	and $ \hlevnt $ an event in $ \beta $.
	Then $ \hlevnt $ is a cut-off* event in $ \beta $ if and only if
	\begin{equation*}
		\constr{\confcut(\cone{\hlevnt})}\Rightarrow
		\bigvee_{\substack{\hlconf\ado\cone{\hlevnt}\\
				\hlhomom(\confcut(\hlconf))=\hlhomom(\confcut(\cone{\hlevnt}))}}
		\constr{\confcut(\hlconf)}.
	\end{equation*}
\end{corollary}
Thus, we showed how to decide for any event $ \hlevnt $ added to a prefix of the unfolding whether it is a cut-off* event, namely, by checking the above implication in Cor.~\ref{cor:CheckCutOffs}.
Note that we can also check whether $ \hlevnt $ is a \emph{cut-off} event (w.r.t.\ Def.~\ref{def:cut-off}) by the implication in Cor.~\ref{cor:CheckCutOffs} when we replace all occurrences of ``$ \hlconf $'' by  ``$ \cone{\hlevnt'} $'' .

\section{Implementation and experimental results}\label{sec:impl}
In this section, we delve into the implementation details of the generalized ERV-algorithm and discuss the results of our experiments.
We give a concise overview of the technical decisions made during implementation
and provide an evaluation of its performance across four novel benchmark families.
We identify a property called ``mode-determinism''
that offers an indicator for whether it is faster to construct (a complete finite prefix of) the symbolic unfolding or the low-level unfolding.

\subsection{Implementation specifics}\label{sec:impl-det}
Other tools designed for generating (prefixes of) P/T Petri net unfoldings include \textsc{Mole} \cite{MOLE}, \textsc{Cunf} \cite{CUNF,RodriguezS13}, and \textsc{Punf} \cite{PUNF}.
However, as these tools are specifically optimized for their intended purpose and do not cater to high-level Petri nets, we opted not to integrate the new algorithms into any of these frameworks.
Furthermore, we refrain from conducting a speed comparison between our implementation and the aforementioned tools. The objective of Section~\ref{sec:impl} is to provide a comparison between two approaches: calculating a respective complete finite prefix of the low-level or the symbolic unfolding.

\medskip
We have devised a prototype implementation called \emph{\textsc{ColorUnfolder}} \cite{ColorUnfolder}
written in the Java programming language.
It serves a dual purpose as an implementation
of the low-level approach as a base for comparisons
and the novel symbolic approach.
It can calculate a finite complete prefix of the low-level unfolding for a given high-level Petri net, combining the concepts from \cite{EsparzaRV02} (complete finite prefixes) and \cite{KhomenkoK03} (generating the low-level unfolding without expansion).
Additionally, it is capable of computing a complete finite prefix of the net's symbolic unfolding,
utilizing a modified version of Alg.~\ref{alg:complpref}.
Since we want to compare the low-level with the high-level case,
we restricted ourselves to nets from the class $ \Nf $ to
guarantee that the low-level unfolding exists.

Both, the generalized (Alg.~\ref{alg:complpref}) and the original (\cite{EsparzaRV02}) ERV-algorithm create possible extensions that are structurally dependent cut-off events,
whereas in the implementation a cut-off event never triggers the calculation of possible extension.
With the same idea, conditions in the postset of cut-off events are never considered for finding co-sets.
This leaves the finite complete prefix unmodified,
as it only eliminates unnecessary work.

\medskip
More importantly, the tool operates on a modified unfolding
since an implementation using the predicates defined here turned out to be very slow.
It rewrites the predicates in the unfolding and modifies arc labels
to drastically reduce the number of variables.
We achieve this by associating with each condition just one
variable, the \emph{internal variable}.
Every condition then has this internal variable on all outgoing edges, as on the ingoing edge. We choose the variable name such that it uniquely
identifies the event that chooses the color of the token. We say an event \emph{chooses}
a color if its transition has a variable on an outgoing edge that is on no ingoing
edge. This is in contrast
to the case where the variable on the outgoing edge was also on an ingoing edge.
In that case the event only forwards the color choice of a previous event.
For example, in the unfolding from Figure~\ref{fig:colorconflictUnf},
\textsc{ColorUnfolder} replaces  the four variables by a single variable.

As a consequence of this renaming, we might need to modify guards in order
to preserve the behavior of the original net. A transition in the original net then can
have two ingoing edges with the same variable, but the corresponding event
has distinct internal variables in those positions. In that case we add a guard
that requires equality of those internal variables.
After a finite complete prefix of the modified unfolding is found,
the result can be easily transformed into the expected result with barely any overhead.

\medskip
This optimization yields a significant speed up.
However, when working on the symbolic unfolding,
in our experiments
still more than 99 percent of the time is spent evaluating the satisfiability of predicates
to identify cut-off events using Cor.~\ref{cor:CheckCutOffs},
and to detect when to discard event candidates because of a color conflict.
For this task we chose the \emph{\textsc{cvc5}} SMT solver~\cite{cvc5}.
It performed best in the relevant category
(non-linear arithmetic with equality and quantifiers)
of the Satisfiability Modulo Theories Competition 2023 (SMT-COMP 2023)\footnote{
	\url{https://smt-comp.github.io/2023/results/equality-nonlineararith-single-query}}.

\subsection{Benchmark families}\label{sec:benchmarks}

In this section, we present four new benchmark families
on which we tested the calculation (resp.\ verification on) of the symbolic unfolding
and compared it to the calculation of (resp.\ verification on) the low-level unfolding.

\subsubsection{Fork And Join}\label{bm:diamond}

The simplest of our benchmark families is called \emph{Fork And Join}.
In the initial marking, a token lies on place $ p_0 $.
A transition $ t $ takes this token from $ p $
and places an arbitrary color on each of its output places.
A transition $ \varepsilon $ then takes these colors from all places, ending the nets execution.
We have two parameters: the first parameter, $ m\in\N $, determines the set of colors $ \hltoks=\{ 0,\dots,m \} $.
The second parameter, $ n\in\N $, determines the number of output places of $ t $.
Fig.~\ref{fig:Independent-Diamonds} shows the independent diamonds for $ n=2 $ in \subref{fig:Independent-Diamond-2} and for $ n=4 $ in \subref{fig:Independent-Diamond-4}.
\begin{figure}[!htb]
\vspace*{-2mm}
	\centering
	\begin{subfigure}{0.48\textwidth}
		\centering
		\includegraphics[width=0.5\textwidth]{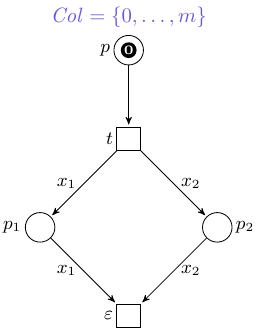}
		\caption{$ n=2 $.\label{fig:Independent-Diamond-2}}
	\end{subfigure}
	\begin{subfigure}{0.48\textwidth}
		\centering
		\includegraphics[width=0.49\textwidth]{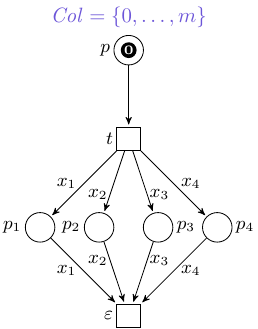}
		\caption{$ n=4 $.\label{fig:Independent-Diamond-4}}
	\end{subfigure}
	\caption{Fork And Join for $ n=2 $ in \subref{fig:Independent-Diamond-2} and $ n=4 $ in \subref{fig:Independent-Diamond-4}.\label{fig:Independent-Diamonds}}\vspace*{-2mm}
\end{figure}

\medskip
The symbolic unfolding of a Fork And Join has $ n+3 $ nodes as it is structurally equal to the net itself.
The low-level unfolding of the expansion has $ (n+2)(m+1)^n+1 $ nodes (since $ t $ is fireable in $ (m+1)^n $ modes), cp.~App.~\ref{app:symb-unfs}.

\begin{figure}[!b]
	\centering
	\scalebox{0.9}{\includegraphics[]{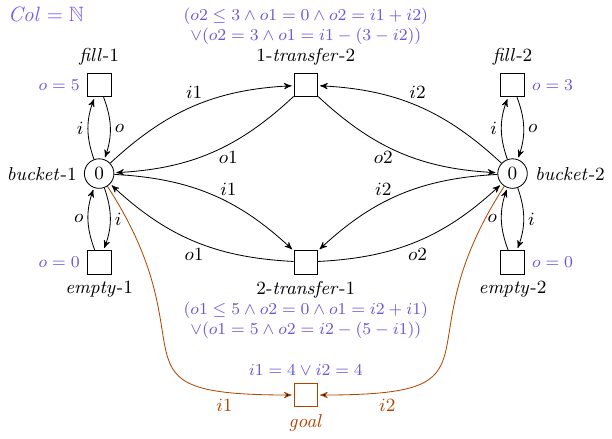} } \vspace*{-2mm}
	\caption{The Water Pouring Puzzle with 2 buckets holding 5 resp.\ 3~liters.\label{fig:buckets}}\vspace*{-6mm}
\end{figure}

\subsubsection{The water pouring puzzle}\label{bm:buckets}

This benchmark family generalizes the following logic puzzle (cf., e.g., \cite{AtwoodP76}):
\begin{quotation}
	\noindent``You have an infinite supply of water and two buckets.
	One holds 5 liters, the other holds 3 liters.
	Measure exactly 4 liters of water in one bucket.''
\end{quotation}
In our generalization we have two parameters.
The first parameter is a finite list $ n=[n_1,\dots,n_k] $ of natural numbers.
Each entry $ n_i $ represents an available bucket holding $ n_i $ liters.
The second parameter, $ m\in\N $ is the amount of water that should be measured.

\medskip
Fig.~\ref{fig:buckets} shows the high-level Petri net corresponding to the parameters
$ n=[3,5] $ and $ m=4 $, corresponding to the puzzle above.
Independently of the parameters, we have $ \hltoks=\N $.
The current fill level of each bucket $ i $ is represented by a place $ \mathit{bucket} $-$ i $,
with an initial color $ 0 $,
and two attached transitions  $ \mathit{fill} $-$ i $ and $ \mathit{empty} $-$ i $,
that, when fired, replace the color on $ \mathit{bucket} $-$ i $ by $ n_i $ or $ 0 $, respectively.
Additionally, for each pair of buckets $ i,j $, there are two transitions $ i $-$ \mathit{transfer} $-$ j $
and $ j $-$ \mathit{transfer} $-$ i $ that transfer as much water as possible from one bucket to the other without overflowing it.
When at least one bucket contains $ m $ liters, the goal transition can fire.
We include a prefix the symbolic unfolding of the net from Fig.~\ref{fig:buckets} in App.~\ref{app:symb-unfs}.

\subsubsection{Hobbits And Orcs}\label{bm:hob-orc}

The \emph{Hobbits And Orcs} problem\footnote{The problem is also known as ``Missionaries and Cannibals'', and is a variation of the ``Jealous Husbands'' problem.} is another logic puzzle (in particular one of many ``river crossing problems'', cf., e.g., \cite{JeffriesPRA1977,PressmanS1989}) and goes as follows:
\begin{quotation}
	\noindent``Three Hobbits and three Orcs must cross a river using a boat which can carry at most two passengers. For both river banks, if there are Hobbits present on the bank, they cannot be outnumbered by orcs (if they were, the Orcs would attack the Hobbits). The boat cannot cross the river by itself with no one on board.''
\end{quotation}

We generalize this problem by introducing two parameters.
The first parameter, $ m\in \N $ is the number of both Hobbits and Orcs.
We always have equally many of the two parties.
The second parameter, $ n\in\N $ is the number of passengers the boat can carry.
We additionally assume that also on the boat, if there are Hobbits present, they cannot be outnumbered by Orcs.

\begin{figure}[!ht]
	\vspace*{4mm}
	\centering
	\begin{subfigure}{0.82\textwidth}\centering
		\includegraphics[width=\textwidth]{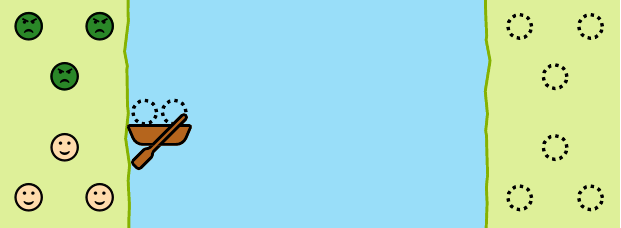}
		\caption{An illustration of the ``Hobbits And Orcs'' problem.\footnotemark\label{fig:hob-orc-ill}}
	\end{subfigure}\\[4mm]
	\begin{subfigure}{0.82\textwidth}\centering
		\includegraphics[width=\textwidth]{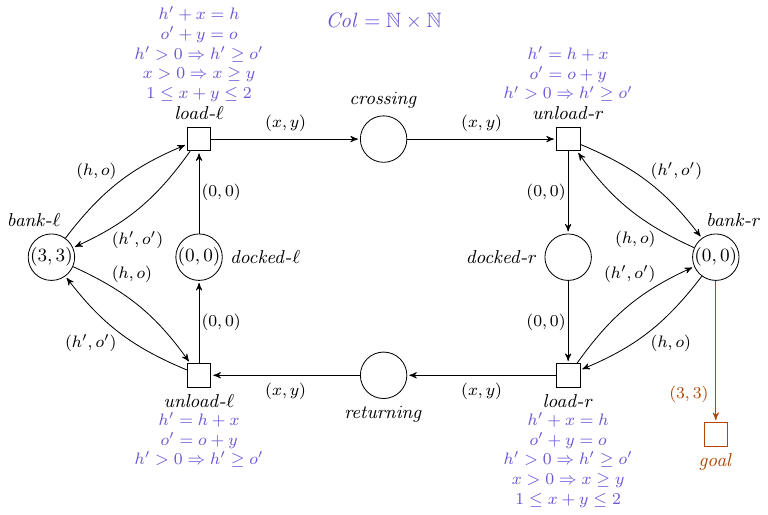}
		\caption{The problem modeled as a high-level Petri net.\label{fig:hob-orc-pn}}
	\end{subfigure}
	\caption[The Hobbits And Orcs problem with 3 Hobbits, 3 Orcs, and a boat fitting 2 passengers.]{The Hobbits And Orcs problem with 3 Hobbits, 3 Orcs, and a boat fitting 2 passengers.\label{fig:hob-orc}}\vspace*{-2mm}
\end{figure}

Figure~\ref{fig:hob-orc-ill} shows an illustration of the original puzzle presented above, with three of both, Hobbits and Orcs, and a boat fitting two passengers.
Figure~\ref{fig:hob-orc-pn} shows the corresponding high-level Petri game.
The colors are given by $ \hltoks=\N\times\N $,
where a tuples describes the number of Hobbits and Orcs at a location --
the left bank, the boat, or the right bank.
We start with three Hobbits and three Orcs on the left bank, indicated by the tuple $ (3,3) $ on the place $ \operatorname{\mathit{bank-\ell}}$.
The four center places describe the current state of the boat,
being either empty and docked on a bank, or loaded and on the river.
Initially, there is a tuple $ (0,0) $ on $ \operatorname{\mathit{docked-\ell}} $, indicating an empty boat the left bank.
Via the transitions $ \mathit{load} $ and $ \mathit{unload} $ left and right,
the boat can be loaded or unloaded, with the guards ensuring all the conditions from the riddle
regarding the number of Hobbits and Orcs on both banks and on the boat.
When all creatures are on the right bank, the goal transition can fire, ending the net's execution.

\subsubsection{Mastermind}\label{bm:mastermind}

\begin{figure}[t]
	\centering
	\includegraphics[]{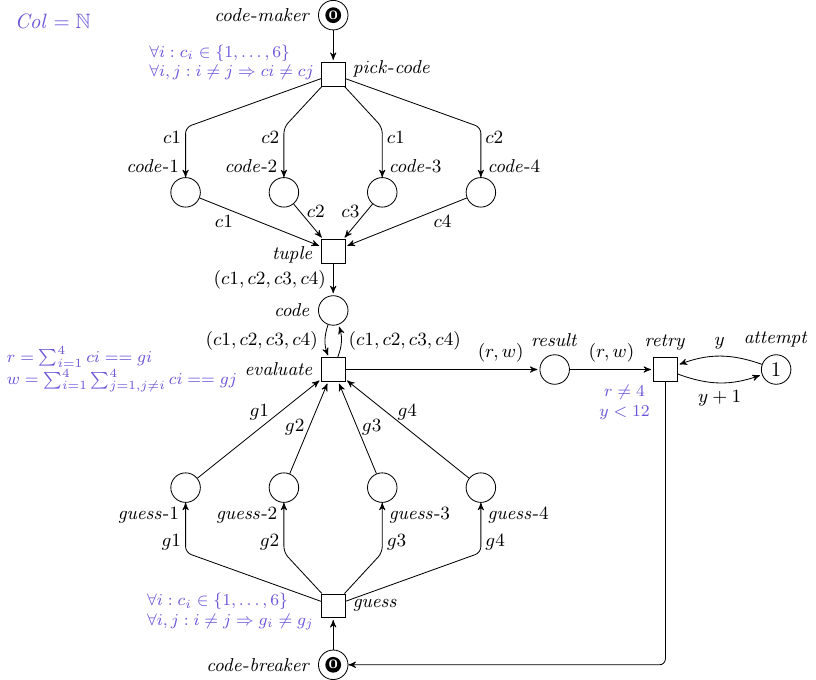}
	\caption{The Mastermind game with a code of length 4, 6 possible colors, and 12 attempts.\label{fig:mastermind}}
\end{figure}
The last benchmark family models a generalization of the classic code-breaking game \emph{Mastermind} developed in the early Seventies and completely solved in 1993 \cite{Knuth76,KoyamaL93}.
The game is played between two players.
The \emph{code maker} secretly chooses an ordered, four digit color code,
with six available colors.
The \emph{code breaker} then guesses the code.
The code maker evaluates the guess by a number of \emph{red pins} indicating how many colors in the guess are in the correct position,
and a number of \emph{white pins} indicating how many colors in the guess appear at a different position in the code.
Using this knowledge, the code breaker makes the next guess, with up to twelve attempts.\footnotetext{The "Boat" icon used in Figure~\ref{fig:hob-orc-ill} is by DinosoftLabs from Noun Project, \url{https://thenounproject.com/dinosoftlab/}.}

In our generalization we have three parameters.
The first parameter, $ m\in\N $, describes the number of available colors.
The second parameter, $ n\in\N $, describes the length of the code.
The third parameter, $ k\in\N $, describes the number of guesses the code breaker can make.
To simplify the net, we restricted the allowed codes to not contain any color twice.

\medskip
Fig.~\ref{fig:mastermind} shows the high-level Petri net with $ m=6 $ available colors, a code of length $ n=4 $, and $ k=12 $ possible attempts, i.e., the scenario described above.
The code maker places one color on each of the four places $ \mathit{code}$-$\mathit{i} $ by firing $ \mathit{pick}$-$\mathit{code} $.
Transition $ \mathit{tuple} $ puts these colors into a tuple on place $ \mathit{code} $.
The purpose of this (for the model unimportant) transition is to make the symbolic unfolding (cp.~App.~\ref{app:symb-unfs}) resemble an actual game board of Mastermind.
The code breaker, analogously to the code maker, concurrently guesses a code via transition $ \mathit{guess} $.
Transition $ \mathit{evaluate} $ compares this guess to the code and places the result,
i.e., the corresponding number of red and white pins, on $ \mathit{result} $.
From there, the code breaker either wins if the guessed code was correct,
or resets with transition $ \mathit{retry} $,
provided that he has another attempt left.

\subsection{Mode-deterministic high-level Petri nets}\label{sec:mode-det}

In the experiments presented in the next section we identified
an important indicator for whether the symbolic approach for a finite complete prefix presented in this paper is expected to outperform,
the complete finite prefix of the low-level unfolding, combining the concepts of \cite{EsparzaRV02} and \cite{KhomenkoK03}.

The identified net property is that in every reachable marking of $ \hlpn $,
every transition can fire in at most one mode.
We call a high-level Petri net with this property \emph{mode-deterministic}, formally:
\begin{definition}\label{def:mode-deterministic}
	A high-level Petri net $ \hlpn $ with transitions $ \hltranss $ is called \emph{mode-deterministic}
	iff
	\begin{equation*}
		\forall \hlmarking\in\reachable(\hlpn)\,
		\forall \hltrans\in\hltranss\,
		\exists^{\leq 1}\hlmode\in\hlmodes(\hltrans) :\hlmarking\fires{\hltrans,\hlmode}.
	\end{equation*}
\end{definition}
In the case of a mode-deterministic net $ \hlpn $, the \emph{skeleton} of $ \hlpn $'s symbolic unfolding
(essentially, the core structure of the high-level occurrence net, devoid of arc labels and guards, and interpreted as a P/T Petri net)
is equivalent in structure to the low-level unfolding of $ \hlpn $'s expansion.
This implies that the high-level abstraction does not offer any computational advantage in the unfolding process.
To the best of our knowledge, this property has not been studied elsewhere.

\medskip
We borrow terminology from ``regular'' determinism and say that a high-level Petri
net is \emph{mode-nondeterministic} if it does not satisfy the above property, and implicitly
describe by a high or low ``degree'' of mode-nondeterminism that there are many resp.\ few
transition-mode combinations making the net mode-nondeterministic.

An illustrative example of a family of mode-deterministic nets is the benchmark family Water Pouring Puzzle introduced in Sec.~\ref{bm:buckets}.
Although this property may not be immediately apparent, it is true that in every state of the system, all transitions can fire in at most one mode.
Specifically, transitions $ \operatorname{\mathit{fill-i}} $ and $ \operatorname{\mathit{fill-i}} $  replace the color on $ \operatorname{\mathit{bucket-i}} $ with a predetermined one.
Every transition $ \operatorname{\mathit{i-transfer-j}} $ emulates the transfer of \emph{all} water from bucket $ i $ that can fit into bucket $ j $.

A condition that guarantees $ \hlpn $ to be a mode-deterministic net is as follows:
``For any conceivable assignment of variables on input arcs, there exists at most one possible mode that completes this assignment to all variables.''
The nets in the Water Pouring Puzzle family fulfill this criterion,
since all output variables ($ o $) are either predetermined or derived from the input variables ($ i $).

An example for a highly mode-nondeterministic nets is the benchmark family Fork And Join from Sec.~\ref{bm:diamond}.
There are only two transitions in the net, but from the initial marking, $ t $ can fire in \emph{every} of its $ m^n $ modes.
This structural pattern of Fork And Join can also be observed in the Mastermind benchmark family (Sec.~\ref{bm:mastermind}),
making these nets also mode-nondeterministic.

The Hobbits and Orcs family (Sec.~\ref{bm:hob-orc}), on the other hand, falls in between and is notably contingent on the parameters involved.
When there are only two seats on the boat, there are merely five \emph{potential} modes of $ \operatorname{\mathit{load-\ell}} $ from the initial marking (reflecting the possible ways to occupy one or both seats with two types of creatures).
However, in the scenario where there are $ n $ seats on the boat, along with $ m $ Hobbits and Orcs, and $ m>n $, the total number of possibilities is $ \sum_{i=1}^{n}i+1=\frac{1}{2}n(n+3) $. It's worth noting that approximately half of these possibilities would result in more orcs than hobbits on the boat, rendering them non-modes.
Hence, given the condition $ m>n $, the ``degree'' of mode-determinism remains unaffected by an increase in $ m $, but solely varies with the parameter $ n $.

\medskip
In the subsequent section, we empirically validate the hypothesis that the level of mode-determinism is inversely proportional to the benefit gained from employing symbolic unfolding compared to low-level unfolding.

\subsection{Experiments and results}\label{sec:exp}

We now present the experimental results of applying our implementation to the four benchmark families presented above.
\textsc{ColorUnfolder} can check the reachability of a (set of) marking(s)
by adding a respective transition to the net.
It then executes Alg.~\ref{alg:complpref} but stops when an instance of
this transition is added to the prefix under construction, which means that the respective
(set of) marking(s) is/are reachable. When the algorithm terminates without such an
instance, the completeness of the prefix implies that the marking is not reachable.

For the nets from the Fork And Join benchmark family, the complete unfolding (being its own smallest finite and complete prefix) is calculated.
For the benchmark families Water Pouring Puzzle and Hobbits And Orcs,
the reachability of the goal state of the riddle is checked.
This is done by adding the respective $ \mathit{goal} $ transition
that is depicted in each of the two figures Fig.~\ref{fig:buckets} and Fig.~\ref{fig:hob-orc}. This transition is not part of the input net.
Finally, for nets in the Mastermind benchmark family,
reachability of a marking with the result of $ n-1 $ red pins and $ 1 $ white pin is checked.
Such a marking is never reachable, so always the complete (but finite) unfolding is calculated.

The tasks described above are checked twice, once using the symbolic unfolding and once using the low-level unfolding as described in Sec.~\ref{sec:impl-det}.
The experiments are calculated
with commit 124b1735 of
\textsc{ColorUnfolder} \cite{ColorUnfolder}
on an otherwise idle system with
an Intel i7-6700K CPU at $4.0$ GHz
and $16$ GB RAM.

\begin{figure}[!b]
	\centering
	\begin{subfigure}[t]{0.49\textwidth}
		\includegraphics[width=\textwidth]{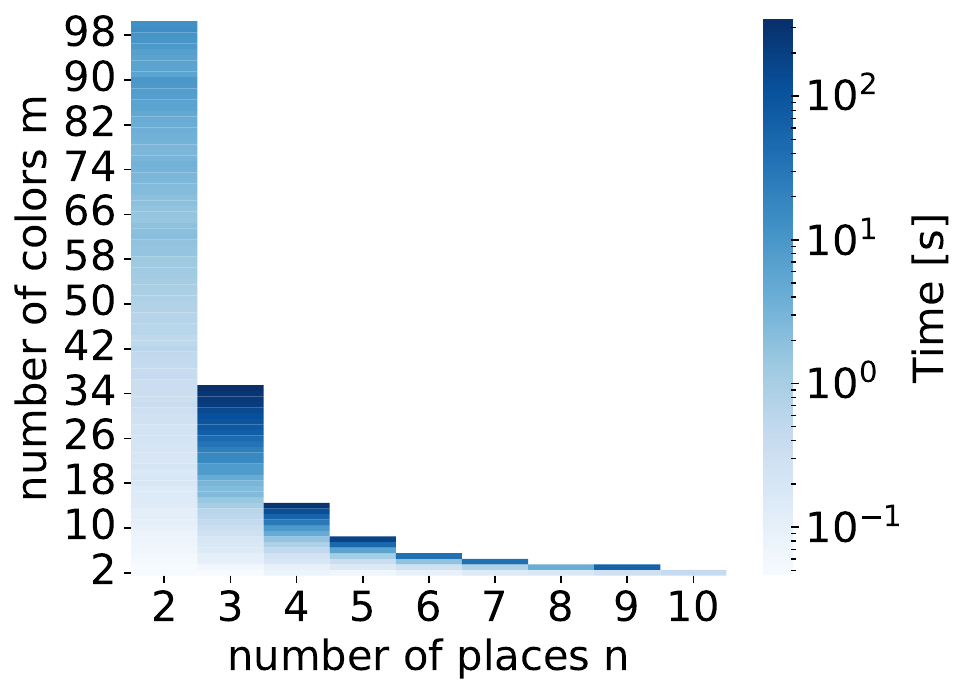}
		\captionsetup{width=.9\linewidth}
		\caption{
			Low-level,
			for growing number of places on x-axis,
			color class on y-axis
			and time as heatmap intensity.
		}\label{fig:independent-diamond-result-ll}
	\end{subfigure}\hspace*{0.01\textwidth}
	\begin{subfigure}[t]{0.49\textwidth}
		\includegraphics[width=\textwidth]{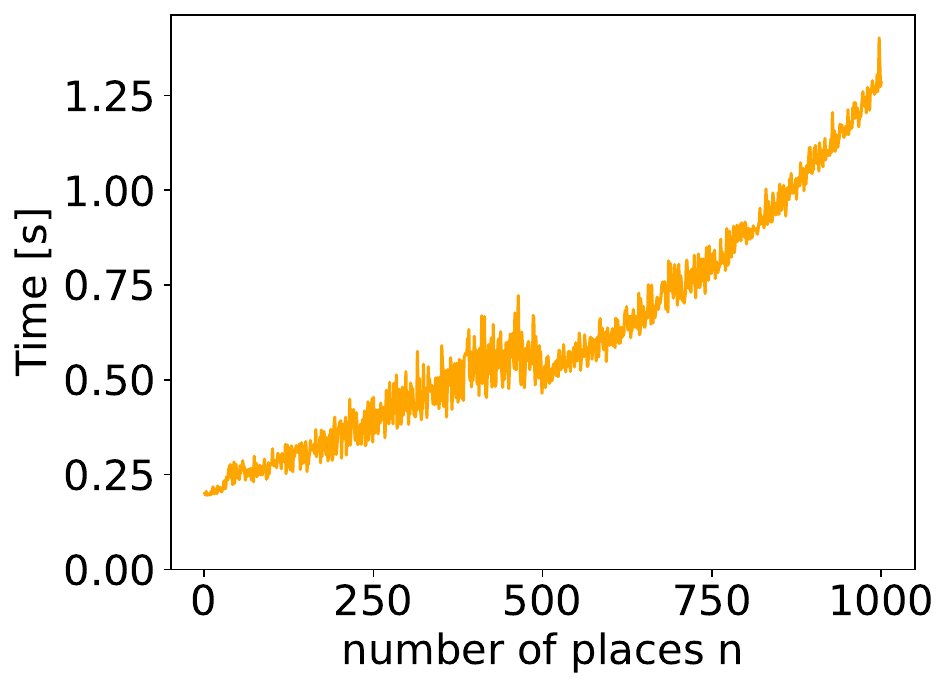}
		\captionsetup{width=.9\linewidth}
		\caption{
			Symbolic,
			for growing number of places on x-axis
			and time on y-axis.
			Behavior is independent of the color class.
			This graph is for $\hltoks = \N$.
		}\label{fig:independent-diamond-result-hl}
	\end{subfigure}
	\caption{Time needed to calculate unfolding \subref{fig:independent-diamond-result-ll} and symbolic unfolding \subref{fig:independent-diamond-result-hl} of Fork \& Join.}\label{fig:independent-diamond-result}
\end{figure}

\paragraph{Fork \& Join.}
The results for Fork \& Join are shown in
Fig.~\ref{fig:independent-diamond-result}.

In Fig.~\ref{fig:independent-diamond-result-ll},
the elapsed time for calculating the low-level unfolding
for the instance with $ n $ places (x-axis) and $ m $ colors (y-axis)
is indicated by the heatmap intensity of the respective cell.
Empty cells exceeded a 5-minute timeout.
We see that the low-level approach is only viable
if at least one parameter is very small.
It takes exponentially more time with a growing color class
and with growing number of places.

Since the symbolic approach
outperforms the low-level approach by a wide margin,
we present it in its own figure, Fig.~\ref{fig:independent-diamond-result-hl}.
For all the cases from Fig.~\ref{fig:independent-diamond-result-ll},
the symbolic approach takes around 200 ms to complete the symbolic unfolding.
This is independent of the color class $ \{ 0,\dots, m \} $.
This corresponds to the fact that the symbolic unfolding of Fork \& Join
is, for a fixed number of places $ n $,
independent of the color class $ \hltoks=\{ 0,\dots,m \} $.
This even does not change for $ \hltoks=\N $.
Since cvc5 (the tool used for checking satisfiability, cp.\ Sec.~\ref{sec:impl-det})
can handle such infinite color domains,
we fix $ \hltoks=\N $, and Fig.~\ref{fig:independent-diamond-result-hl}
shows the elapsed time (y-axis) for calculating the symbolic unfolding
of the Fork \& Join instance with $ n $ places (x-axis).
This approach is very fast (but not linear in the number of places).

\medskip
The symbolic approach is faster for all choices of the parameters.
Only for the smallest choices are both approaches equally fast within the margin of error.
This behavior was expected since the Petri nets are ``highly non mode-deterministic'',
as explained above, and the low-level unfolding is much broader than the symbolic unfolding.

\paragraph{Water Pouring Puzzle.}
For the Water Pouring Puzzle we cannot get interesting results
by varying one parameter while holing the others fixed,
because the complexity of the solution is highly volatile
and dependent on the combination of all parameters.
Since the nets are mode-deterministic,
the low-level and symbolic unfolding, as well as their prefixes, are isomorphic.

\begin{table}[!h]
\centering
\caption{Results of the Water Pouring Puzzle benchmark.}\label{tbl:buckets-result}
\begin{tabular}{cclrrcc}
	buckets $ n $ & target $ m $ & solvable & $|B|$ & $|E|$ & time low-level & time symbolic\\\hline
	$3 , 5 $     & $4$    & yes, $6$ steps  & $ 90$   & $ 75$   & $\boldsymbol{ 95 \,\mathrm{ms}}$ & $1.2 \,\mathrm{s}$ \\
	$15, 17$     & $10$   & yes, $18$ steps & $258$   & $195$   & $\boldsymbol{220 \,\mathrm{ms}}$ & $36  \,\mathrm{s}$ \\
	$57, 73$     & $51$   & yes, $92$ steps & $1294$  & $935$   & $\boldsymbol{1.7 \,\mathrm{s}}$  & $>1  \,\mathrm{h}$ \\
	$9 , 12$     & $4$    & no              & $106$   & $ 74$   & $\boldsymbol{130 \,\mathrm{ms}}$ & $1.8 \,\mathrm{s}$ \\
	$10, 16$     & $5$    & no              & $190$   & $134$   & $\boldsymbol{140 \,\mathrm{ms}}$ & $11  \,\mathrm{s}$ \\
	$8, 14, 17$  & $2$    & yes, $4$ steps  & $411$   & $642$   & $\boldsymbol{300 \,\mathrm{ms}}$ & $26  \,\mathrm{s}$ \\
	$14, 26, 27$ & $14$   & yes, $15$ steps & $21635$ & $15279$ & $\boldsymbol{7.7 \,\mathrm{s}}$  & $>1  \,\mathrm{h}$ \\
	$12, 15, 18$ & $10$   & no              & $2391$  & $1442$  & $\boldsymbol{550 \,\mathrm{ms}}$ & $20  \,\mathrm{min}$ \\
	$12, 21, 27$ & $8$    & no              & $4029$  & $2444$  & $\boldsymbol{1.0 \,\mathrm{s}}$  & $88  \,\mathrm{min}$ \\
\end{tabular}
\end{table}

Table~\ref{tbl:buckets-result} presents the results for this benchmark.
The original puzzle with $ n=[3,5] $ and $ m=4 $ is included in the first line.
Additionally, we randomly selected eight parameter sets such that we consider four scenarios involving 2 buckets and four scenarios involving 3 buckets.

\medskip
We indicate in each scenario whether the puzzle is solvable,
and in the positive case, how many steps a minimal solution has.
We compare the time to check the reachability of a goal state (fireability of the $ \mathit{goal} $ transition, cp.\ above) of the low-level and symbolic approach.
The faster approach is highlighted with bold font.
Additionally, we notate the number of conditions ($ |\hlconds| $) and events ($ |\hlevnts| $)
of the generated prefix. Each of these two numbers coincides between the two approaches.

\medskip
When we pick the parameters at random the low-level approach generally is much faster.
We can, however, construct examples in which the symbolic approach outspeeds the low-level approach.
In these constructed cases we have a large color domain but can reach the goal quite quickly,
e.g., with two buckets of capacity $ 10^6 $ and $ 10^6+1 $, and a target of $ m=1 $.

\paragraph{Hobbits And Orcs.}

\begin{figure}[!h]
\vspace*{-2mm}
	\centering
	\begin{subfigure}{0.32\textwidth}
		\includegraphics[width=\textwidth]{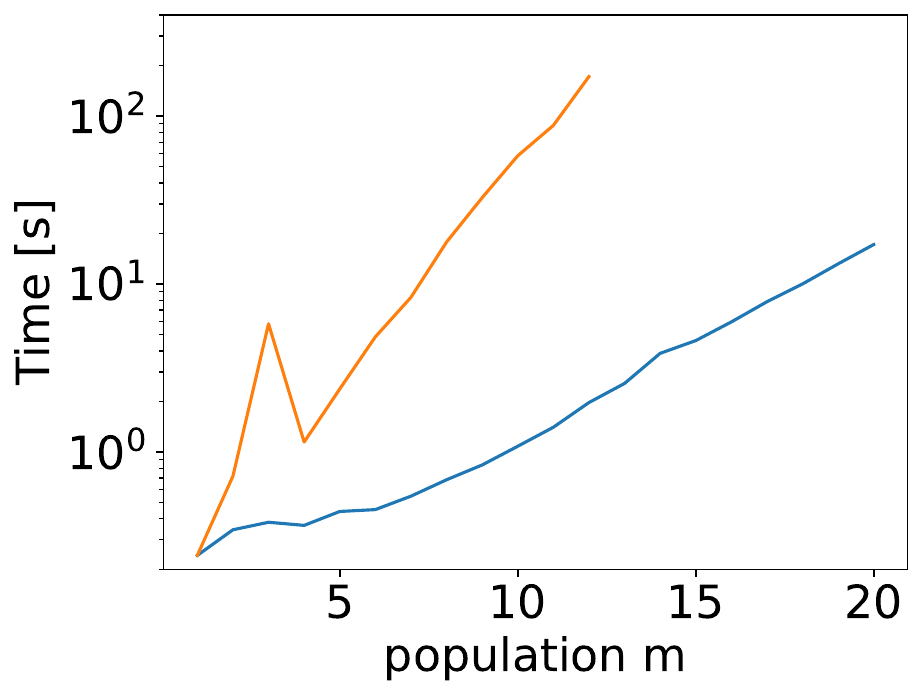}
		\caption{$n=2$}\label{fig:hob2}
	\end{subfigure}
	\hspace*{0.005\textwidth}
	\begin{subfigure}{0.32\textwidth}
		\includegraphics[width=\textwidth]{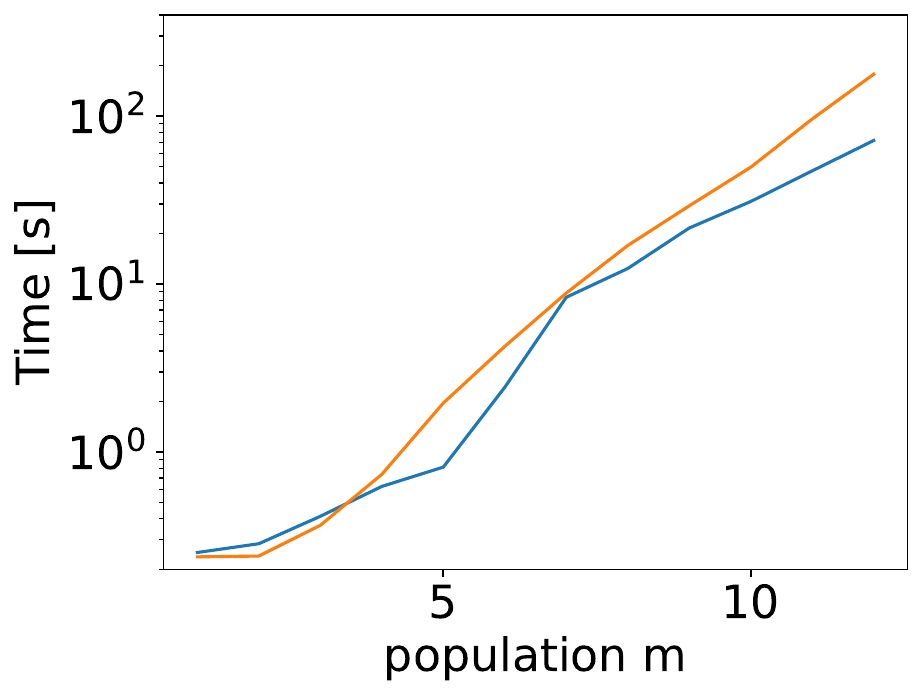}
		\caption{$n=4$}\label{fig:hob4}
	\end{subfigure}
	\hspace*{0.005\textwidth}
	\begin{subfigure}{0.32\textwidth}
		\includegraphics[width=\textwidth]{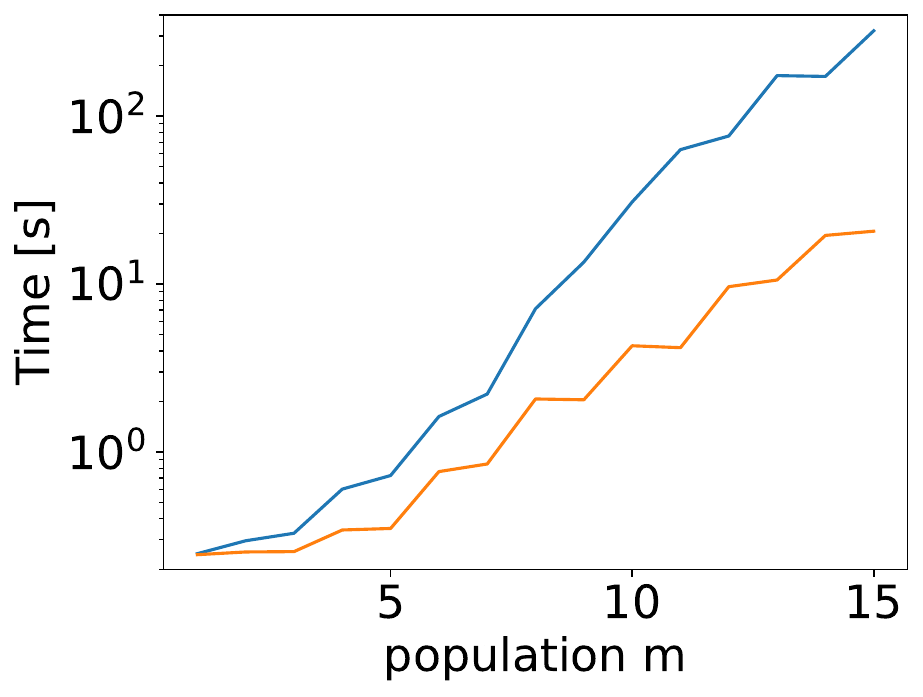}
		\caption{$n=6$}\label{fig:hob6}
	\end{subfigure}
	\caption{Results for Hobbits And Orcs problem
		for a fixed boat capacity~$n$ in \subref{fig:hob2}, \subref{fig:hob4}, \subref{fig:hob6},  and population with $ m $ of each Hobbits and Orcs on the
         x-axis. 		The orange line shows the time needed by the symbolic approach 	and the blue line shows the low-level approach.
	}\label{fig:hobbits-and-orcs-benchmark-result}
\end{figure}
\bgroup
\def\arraystretch{0.95}
\begin{table}[!ht]\small
	\centering
	\caption{Results of the Mastermind benchmark.}\label{tbl:mastermind-result}
	\begin{tabular}{rrr|rrr|rrr}
		&      &        & \multicolumn{3}{c}{Low-level}                             & \multicolumn{3}{c}{Symbolic}                               \\
		$k$   & $n$  & $m$    & time                              & $|\conds|$     & $|\evnts|$     & time                              & $|\hlconds|$      & $|\hlevnts|$     \\\hline
		$12$   & $4$ & $6$    & $            >2  \,\mathrm{min} $ &  -        &  -        & $\boldsymbol{1   \,\mathrm{min}}$ & $149$      & $36$      \\
		$1$   & $3$  & $3$    & $            86  \,\mathrm{ms}  $ & $219$     & $48$      & $\boldsymbol{62  \,\mathrm{ms}} $ & $14$       & $3$       \\
		$1$   & $3$  & $5$    & $            1.8 \,\mathrm{s}   $ & $18363$   & $3720$    & $\boldsymbol{54  \,\mathrm{ms}} $ & $14$       & $3$       \\
		$1$   & $3$  & $6$    & $            18  \,\mathrm{s}   $ & $72723$   & $14640$   & $\boldsymbol{61  \,\mathrm{ms}} $ & $14$       & $3$       \\
		$1$   & $3$  & $7$    & $            >2  \,\mathrm{min} $ & -         & -         & $\boldsymbol{54  \,\mathrm{ms}} $ & $14$       & $3$       \\
		$1$   & $3$  & $1000$ & $            >2  \,\mathrm{min} $ & -         & -         & $\boldsymbol{53  \,\mathrm{ms}} $ & $14$       & $3$       \\
		$1$   & $4$  & $4$    & $            1   \,\mathrm{s}   $ & $3651$    & $624$     & $\boldsymbol{60  \,\mathrm{ms}} $ & $17$       & $3$       \\
		$1$   & $4$  & $5$    & $            >2  \,\mathrm{min} $ & -         & -         & $\boldsymbol{67  \,\mathrm{ms}} $ & $17$       & $3$       \\
		$1$   & $5$  & $1000$ & $            >2  \,\mathrm{min} $ & -         & -         & $\boldsymbol{77  \,\mathrm{ms}} $ & $20$       & $3$       \\
		$1$   & $28$ & $1000$ & $            >2  \,\mathrm{min} $ & -         & -         & $\boldsymbol{1   \,\mathrm{s}}  $ & $89$       & $3$       \\
		$2$   & $3$  & $3$    & $            215 \,\mathrm{ms}  $ & $1719$    & $438$     & $\boldsymbol{206 \,\mathrm{ms}} $ & $24$       & $6$       \\
		$2$   & $3$  & $4$    & $            3.3 \,\mathrm{s}   $ & $110115$  & $27672$   & $\boldsymbol{120 \,\mathrm{ms}} $ & $24$       & $6$       \\
		$2$   & $3$  & $5$    & $            111 \,\mathrm{s}   $ & $1726443$ & $432060$  & $\boldsymbol{124 \,\mathrm{ms}} $ & $24$       & $6$       \\
		$2$   & $3$  & $6$    & $            >2  \,\mathrm{min} $ & -         & -         & $\boldsymbol{119 \,\mathrm{ms}} $ & $24$       & $6$       \\
		$2$   & $3$  & $1000$ & $            >2  \,\mathrm{min} $ & -         & -         & $\boldsymbol{127 \,\mathrm{ms}} $ & $29$       & $6$       \\
		$2$   & $4$  & $4$    & $            5.5 \,\mathrm{s}   $ & $137235$  & $27672$   & $\boldsymbol{1.3 \,\mathrm{s}}  $ & $29$       & $6$       \\
		$2$   & $4$  & $5$    & $            >2  \,\mathrm{min} $ & -         & -         & $\boldsymbol{116 \,\mathrm{ms}} $ & $29$       & $6$       \\
		$2$   & $4$  & $1000$ & $            >2  \,\mathrm{min} $ & -         & -         & $\boldsymbol{162 \,\mathrm{ms}} $ & $29$       & $6$       \\
		$2$   & $10$ & $1000$ & $            >2  \,\mathrm{min} $ & -         & -         & $\boldsymbol{1   \,\mathrm{s}}  $ & $59$       & $6$       \\
		$3$   & $4$ & $4$     & $            >2  \,\mathrm{min} $ & -         & -         & $\boldsymbol{35  \,\mathrm{s}}  $ & $41$       & $9$       \\
		$4$   & $3$ & $3$     & $            2.6 \,\mathrm{s}   $ & $46719$   & $12138$   & $\boldsymbol{803 \,\mathrm{ms}} $ & $44$       & $12$      \\
		$5$   & $3$ & $3$     & $            39  \,\mathrm{s}   $ & $234219$  & $60888$   & $\boldsymbol{1.3 \,\mathrm{s}}  $ & $54$       & $15$      \\
		$6$   & $3$ & $3$     & $            >2  \,\mathrm{min} $ & -         & -         & $\boldsymbol{1.8 \,\mathrm{s}}  $ & $64$       & $18$      \\
	\end{tabular}
\end{table} 
\egroup

For the Hobbits And Orcs benchmark we vary the
and the capacity $ n $ of the boat fixed in Figures \ref{fig:hob2}, \ref{fig:hob4}, \ref{fig:hob6},
and number the of each Hobbits and Orcs, $ m $, plotted on the x-axis.
We indicate the time needed by the low-level approach in blue and by the symbolic approach in orange,
with a timeout of 3 minutes.

\medskip
We find that the low-level approach performs better
when the boat capacity is smaller.
For $ n=2 $, the symbolic approach gets a timeout at $ m=13 $ while the low-level approach can calculate the prefix up to $ m=33 $.
We expected this behavior, because a smaller the boat capacity yields a higher degree of mode-determinism in the net.
When the capacity increases, the net is `less' mode-deterministic.
For $n \geq 6$ the symbolic approach is faster for all population sizes.
Independent of the boat size $ n $, both approaches take exponentially more time with growing population size.

\paragraph{Mastermind.}
We present some results for the Mastermind benchmark in Table~\ref{tbl:mastermind-result}.
We compare the time needed by the low-level and symbolic approach to generate the unfolding,
with a timeout of two minutes.
The table also shows the number of nodes (conditions $ |\conds| $ resp.\ $ |\hlconds| $, and events $ |\evnts| $ resp.\ $ |\hlevnts| $) in the unfolding.
In each row the faster time is highlighted in bold font.

\medskip
The first line shows the original problem with $ k=12 $ attempts, a code length of $ n=4 $, and $ m=6 $ colors.
We observe that in the low-level approach,
time and size grow exponentially with respect to the number $ m $ of colors,
whereas the high-level approach remains constant.
When we increase the parameter~$n$ controlling the length of the code,
both approaches grow exponentially in time.
Interestingly the high-level approach only grows linear in \emph{size} (number of nodes).
This is due to the size of the guard of the $\mathit{evaluate}$ transition
increasing exponentially ($n$ choose $2$).

Overall, the symbolic approach is the clear winner of the Mastermind benchmark.
The low-level approach can only barely compete for the smallest instances of the problem.

We also observe that, in comparison to other parameters,
the performance of the high-level approach drops when $n=m$.
We currently have no explanation for this.
This could be a quirk of the SMT solvers, or alternatively, the formulas may be inherently more challenging to solve in this particular case.
The low-level approach is not impacted negatively by this case, but still much slower than the symbolic approach, as seen in lines with parameters $ (1, 3, 3) $, $ (1, 4, 4 ) $, $ (2, 3,3) $, $ (2, 4, 4) $,
and in the last four lines.

\section{Conclusions and outlook}\label{sec:concl}

We introduced the notion of complete finite prefixes of symbolic unfoldings of high-level Petri nets.
We identified a class of 1-safe high-level nets generalizing 1-safe P/T nets,
for which we generalized the well-known algorithm by Esparza et al.\ to compute such a finite and complete prefix.
This constitutes a consolidation and generalization of the concepts of \cite{EsparzaRV02,Chatain06,ChatainF10,ChatainJ04}.
While the resulting symbolic prefix has the same depth as
a finite and complete prefix of the unfolding of the represented P/T net,
it can be significantly smaller due to less branching.
In the case of infinitely many reachable markings
(where the original algorithm is not applicable)
we identified the class of so-called \emph{symbolically compact} nets
for which an adapted version of the generalized algorithm effectively computes a finite complete prefix of the symbolic unfolding.
For that, we showed how to check an adapted cut-off criterion by symbolically describing sets of markings.
We implemented the generalized algorithm and tested it against four novel benchmark families. This experimentation validated an indicator for whether the symbolic approach is expected to outperform the low-level approach. This indicator relies on the concept of a net property we call ``mode-determinism''.

Future works include
the construction of a symbolic reachability graph for symbolically compact nets
and a comparison with the complete finite prefix, as outlined in Sec.~\ref{sec:feasibility}.
Additionally, a generalization for $ k $-bounded high-level Petri nets seems possible.
Furthermore, we want to apply the results to \emph{high-level Petri games} \cite{Gieseking2021,Gieseking2020}
to find ``symbolic strategies'' with techniques similar to \cite{Finkbeiner2015} or \cite{Chatain2014,AguirreSamboni2022},
employing a successively increasing bound on size of the considered prefix of the symbolic unfolding.

We sincerely thank the anonymous reviewers for their valuable feedback, which significantly improved this paper,
and Paul Hannibal for the discussions about the insufficiency of the cut-off criterion for symbolically compact nets.

\bibliographystyle{fundam}
\bibliography{bib}

\appendix

\section{Appendix}

\subsection{Examples of adequate orders}\label{app:adequateOrders}

We show that the adequate order used in \cite{McMillan95}, as well as the orders $ \ado_E $ and $ \ado_F $ treated in \cite{EsparzaRV02}, when lifted to the symbolic unfolding, are still adequate orders.
In particular we show that $ \ado_F $ is a total adequate order on the symbolic unfolding,
limiting the size of the later constructed finite prefix.
The definition of these orders does not change,
so we take most of the following notation directly from \cite{EsparzaRV02}.

\paragraph{The orders \texorpdfstring{$ \ado_M $}{} and \texorpdfstring{$ \ado_E $}{}.}
The order $ \ado_M $ used in \cite{McMillan95} is defined by
$ \hlconf_1\ado_M\hlconf_2:\Leftrightarrow |\hlconf_1|<|\hlconf_2| $.
It is trivial to see that $ \ado_M $ satisfies i) and ii) from Def.~\ref{def:adequateOrder}.
Since $ \extmonom_{1,\hlext}^2 $ is a  injective,
we have $  |\extmonom_{1,\hlext}^2(\hlext)|=|\hlext| $, which yields iii).

For a high-level Petri net $ \hlpn $, let $ \ll $ be an arbitrary total order on the transitions of $ \hlpn $. Given a set $ \hlevnts' $ of events in the unfolding of $ \hlpn $, let $ \pparikh(\hlevnts') $ be that sequence of transitions which is ordered according to $ \ll $, and contains each transition $ \hltrans $ as often as there are events in $ \hlevnts' $ with label $ \hltrans $. We say $ \pparikh(\hlevnts_1)\ll\pparikh(\hlevnts_2) $ if $ \pparikh(\hlevnts_1) $ is lexicographically smaller than $ \pparikh(\hlevnts_2) $ with respect to the order $ \ll $.

The order $ \ado_E $ is then defined as follows:
let $ \hlconf_1,\hlconf_2 $ be two configurations of the symbolic unfoldings of a high-level Petri net. $ \hlconf_1\ado_E\hlconf_2 $ holds if either $ |\hlconf_1|<|\hlconf_2| $, or
$ |\hlconf_1|=|\hlconf_2| $ and $ \pparikh(\hlconf_1)\ll\pparikh(\hlconf_2) $.
The proof that $ \ado_E $ is an adequate order works exactly as in \cite{EsparzaRV02}:

It is easy to show that $ \ado_E $ is a well-founded partial order implied by inclusion.
We now show that $ \ado_E $ is preserved by finite extensions.
As already mentioned above, $ |\hlext|=|\extmonom_{1,\hlext}^2(\hlext)| $.
Additionally, we have $ \pparikh(\hlext)=\pparikh(\extmonom_{1,\hlext}^2(\hlext)) $,
since $ \extmonom_{1,\hlext}^2  $ preserves the labeling of events.

Assume $ \hlconf_1\ado_E\hlconf_2 $.
If $ |\hlconf_1|<|\hlconf_2| $, then $ |\hlconf_1{\oplus}\hlext|<|\hlconf_2{\oplus}\extmonom_{1,\hlext}^2(\hlext)| $.
If $  |\hlconf_1|=|\hlconf_2| $ and $ \pparikh(\hlconf_1)\ll\pparikh(\hlconf_2) $, then $ |\hlconf_1{\oplus}\hlext|=|\hlconf_2{\oplus}\extmonom_{1,\hlext}^2(\hlext)| $ and, by the properties of the lexicographic order,
$ \pparikh(\hlconf_1{\oplus}\hlext)\ll\pparikh(\hlconf_2{\oplus}\extmonom_{1,\hlext}^2(\hlext)) $.

\paragraph{The Total Adequate Order \texorpdfstring{$ \ado_F $}{}.}
The \emph{Foata normal form} $ \hlfoata $ of a configuration $ \hlconf $ is obtained by starting with $ \hlfoata $ empty, and iteratively deleting the set $ \Min(\hlconf) $ from $ \hlconf $ and appending it to $ \hlfoata $, until $ \hlconf $ is empty.

Given two configurations $ \hlconf_1,\hlconf_2 $,
we can compare their Foata normal forms $ \hlfoata_1=\hlconf_{11}\dots\hlconf_{1n_1} $ and $  \hlfoata_2=\hlconf_{21}\dots\hlconf_{2n_2} $
with respect to the order $ \ll $ by
saying $ \hlfoata_1\ll\hlfoata_2 $ if there exists
$ i\leq i\leq n_1 $ such that
$ \pparikh(\hlconf_{1j})=\pparikh(\hlconf_{2j}) $ for every $ 1\leq j< i $, and $ \pparikh(\hlconf_{1i})\ll\pparikh(\hlconf_{2i}) $.
\begin{definition}[Order $ \ado_F $]
	let $ \hlconf_1 $ and $ \hlconf_2 $ be two configurations of the symbolic unfolding of a high-level Petri net.
	$ \hlconf_1\ado_F\hlconf_2 $ holds if
	\begin{itemize}
		\item $ |\hlconf_1|<|\hlconf_2| $, or
		\item $ |\hlconf_1|=|\hlconf_2| $ and $ \pparikh(\hlconf_1)\ll\pparikh(\hlconf_2) $, or
		\item $\pparikh(\hlconf_1)=\pparikh(\hlconf_2) $ and
		$ \hlfoata_1\ll\hlfoata_2 $.
	\end{itemize}
\end{definition}

We prove that $ \ado_F $ is a total adequate order.
In the proof, (a) -- (c) are taken directly from \cite{EsparzaRV02},
with small adaptations due to the high-level formalism.
While the ideas from (d) also come directly from \cite{EsparzaRV02},
we have work with the monomorphism $ \extmonom_{1,\hlext}^2 $ instead of the isomorphism $ I_1^2 $, and the new definition of adequate order.
This is where the only deviation from \cite{EsparzaRV02} happens.

\medskip
Let $ \beta=\tup{\hlonet,\hlhomom} $ be the symbolic unfolding of $ \hlpn=\tup{\hlns,\hlmarkings_0} $.
\begin{enumerate}[label=(\alph*)]
	\item $ \ado_F $ is a well-founded partial order.

	This follows immediately from the fact that $ \ado_E $ is a well-founded partial order as is the lexicographic order on transition sequences of some fixed length.

	\item $ \hlconf_1\subset \hlconf_2 $ implies $ \hlconf_1\ado_F\hlconf_2 $.

	This is obvious, since $ \hlconf_1\subset \hlconf_2 $ implies $ |\hlconf_1|< |\hlconf_2| $.

	\item $ \ado_F $ is total.

	Assume that $ \hlconf_1 $ and $ \hlconf_2 $ are two incomparable configurations under $ \ado_F $, i.e.,
	$ |\hlconf_1|=|\hlconf_2| $, $ \pparikh(\hlconf_1)=\pparikh(\hlconf_2) $, and $ \pparikh(\hlfoata_1)=\pparikh(\hlfoata_2) $.
	We prove $ \hlconf_1=\hlconf_2 $ by induction on the common size $ k=|\hlconf_1|=|\hlconf_2| $.

	The base case $ k=0 $ gives $ \hlconf_1=\hlconf_2=\emptyset $, so assume $ k>0 $.

	We first prove $ \Min(\hlconf_1)=\Min(\hlconf_2) $.
	Aiming a contradiction, assume w.l.o.g.\ that $ \hlevnt_1\in\Min(\hlconf_1)\setminus\Min(\hlconf_2) $.
	Since $ \pparikh(\Min(\hlconf_1))=\pparikh(\Min(\hlconf_2)) $, $ \Min(\hlconf_2) $ contains an event $ \hlevnt_2 $ s.t.\ $ \hlhomom(\hlevnt_1)=\hlhomom(\hlevnt_2) $.
	Since $ {\rightarrow}\Min(\hlconf_1) $ and $ {\rightarrow}\Min(\hlconf_2) $ are subsets of $ \hlconds_0 $,
	and all conditions of $ \hlconds_0 $ carry different labels,
	we have ${\rightarrow}{\hlevnt_1}={\rightarrow}{\hlevnt_2} $,
	and thus, $ \preset{\hlevnt_1}=\preset{\hlevnt_2} $.
	This contradicts the definition of symbolic branching processes.

	Since $ \Min(\hlconf_1)=\Min(\hlconf_2) $,
	both $ \hlconf_1\setminus\Min(\hlconf_1) $ and $ \hlconf_2\setminus\Min(\hlconf_1) $ are configurations of the branching process $ {\Uparrow}\Min(\hlconf_1) $
	of $ \tup{\hlns,\confmarks(\Min(\hlconf_1))} $,
	and they are incomparable under $ \ado_F $ by construction.
	Since the common size of $ \hlconf_1\setminus\Min(\hlconf_1) $ and $ \hlconf_2\setminus\Min(\hlconf_1) $ is strictly smaller than $ k $, we can apply the induction hypothesis and conclude $ \hlconf_1=\hlconf_2 $.

	\item $ \ado_F $ is preserved by finite extensions.

	Take two finite configurations $ \hlconf_1 $ and $ \hlconf_2 $, let $ \hlext $ be a finite suffix of $ \hlconf_1 $, and let $ \hlmarking\in\confmarks(\hlconf_1)\cap\confmarks(\hlconf_2) $ such that $ \hlconf_1\llbracket\hlmarking\rrbracket\hlext $.
	We have to show that
	$ \hlconf_1\ado\hlconf_2 $
	implies $ \hlconf_1{\oplus}\hlext\ado\hlconf_2{\oplus} \extmonom_{1,\hlext}^{2}(\hlext) $.

	First, notice that we can assume $ \hlext=\{\hlevnt\} $:
	For $ \hlevnt\in\Min(\hlext) $ we have from $ \hlconf_1\llbracket\hlmarking\rrbracket\hlext $ that  $\exists \insta\in\instas(\hlconf_1{\oplus}\hlext): \confmark(\hlconf_1.\insta|_{\hlvars_{\hlconf_1\cup\{\bot\}}})=\hlmarking $.
	Thus, for $ \hlmarking' $ s.t.\ $ \hlmarking[\hlhomom(\hlevnt).\hlmode\rangle\hlmarking' $ with $ \hlmode=\insta\circ[\hlvar\mapsto\hlvar_\hlevnt]_{\hlvar\in\hlvars(\hlevnt)} $,
	we have that
	$ \hlmarking'\in\hlmarkings(\hlconf_1{\oplus}\{\hlevnt \})\cap\hlmarkings(\hlconf_2{\oplus}\{ \extmonom_{1,\hlext}^2(\hlevnt) \}) $ and $ (\hlconf_1{\oplus}\{\hlevnt \})\llbracket\hlmarking'\rrbracket(\hlext\setminus\{\hlevnt\})$.

	Second, the cases $ |\hlconf_1|<|\hlconf_2| $ and $ \hlconf_1 \ado_E \hlconf_2 $ in (i),
	(and the respective cases
	$ |\hlconf_2|<|\hlconf_1| $ and $ \hlconf_2\ado_E\hlconf_1 $ in (ii))
	are easy (shown above).
	Hence, assume $ |\hlconf_1|=|\hlconf_2| $
	and $ \pparikh(\hlconf_1)=\pparikh(\hlconf_2) $.

	Third, we show that under these two assumptions
	$ \hlevnt $ is a minimal event of $ \hlconf_1':=\hlconf_1\cup\{ \hlevnt \} $
	if and only if
	$ \extmonom_{1,\hlext}^2(\hlevnt) $ is a minimal event of $ \hlconf_2':=\hlconf_2\cup\{ \extmonom_{1,\hlext}^2(\hlevnt) \} $.
	Let $ \hlevnt $ be minimal in $ \hlconf_1' $,
	i.e., the transition $ \hlhomom(\hlevnt) $ can be fired in a mode in one initial marking.
	Let $ \hlplace\in{\rightarrow}\hlhomom(\hlevnt) $;
	then no condition in $ {\rightarrow}\hlconf\cup\hlconf{\rightarrow} $ is labeled $ \hlplace $,
	since these conditions would be concurrent to the $ \hlplace $-labeled condition in $ {\rightarrow}\hlevnt $,
	contradicting that $ \tup{\hlns,\hlmarkings_0} $ is safe.
	Thus, $ \hlconf_1 $ contains no event $ \hlevnt' $ with
	$ \hlplace\in{\rightarrow}\hlhomom(\hlevnt') $,
	and the same holds for $ \hlconf_2 $, since $ \pparikh(\hlconf_1)=\pparikh(\hlconf_2) $.
	Therefore, the conditions in $ \confcut(\hlconf_2) $ with label in $ {\rightarrow}\hlhomom(\hlevnt) $ are minimal conditions of~$ \beta $,
	and $ \extmonom_{1,\hlext}^2(\hlevnt)=\hlevnt $ is minimal in $ \hlconf_2' $.
	The reverse implication holds analogously, since about $ \hlconf_1 $ and $ \hlconf_2 $ we have only used the hypothesis $ \pparikh(\hlconf_1)=\pparikh(\hlconf_2) $.

	With this knowledge, we now show the implication.
	Assume $ \hlconf_1\ado_F\hlconf_2 $.
	We show $ \hlconf_1'\ado_F\hlconf_2' $.

	If $ \Min(\hlconf_1)\ado_E\Min(\hlconf_2) $,
	then we now see $ \Min(\hlconf_1')\ado_E\Min(\hlconf_2') $,
	hence $ \pparikh(\hlfoata_1')\ll\pparikh(\hlfoata_2') $
	and so we are done.
	If $ \pparikh(\Min(\hlconf_1))=\pparikh(\Min(\hlconf_2)) $ and $ \hlevnt\in\Min(\hlconf_1') $,
	then
	\begin{equation*}
		\hlconf_1'\setminus\Min(\hlconf_1')
		=\hlconf_1\setminus\Min(\hlconf_1)
		\ado_F\hlconf_2\setminus\Min(\hlconf_2)
		=\hlconf_2'\setminus\Min(\hlconf_2'),
	\end{equation*}
	hence $ \hlconf_1'\ado_F\hlconf_2' $.
	Finally, if $ \pparikh(\Min(\hlconf_1))=\pparikh(\Min(\hlconf_2)) $ and $ \hlevnt\notin\Min(\hlconf_1') $,
	we again argue that $ \Min(\hlconf_1)=\Min(\hlconf_2) $
	and that, hence, $ \hlconf\setminus\Min(\hlconf_1) $ and $ \hlconf_2\setminus\Min(\hlconf_1) $ are configurations of the branching process $ {\Uparrow}\Min(\hlconf_1) $ of $ \tup{\hlns,\confmarks(\Min(\hlconf_1))} $.
	With an inductive argument we get $ \hlconf_1'\setminus\Min(\hlconf_1')
	\ado_F\hlconf_2'\setminus\Min(\hlconf_2') $
	and are also done in this case.
\end{enumerate}

\subsection{More symbolic and low-level unfoldings}\label{app:symb-unfs}

In this appendix we present unfoldings omitted in the main body of the paper.

\paragraph{Low-level Unfolding of Fork And Join.}
Since the symbolic unfolding of any Fork And Join is structurally equal to the net itself we do not display it here.
Fig.~\ref{fig:diamond-unf} shows the low-level unfolding with $ (n+2)(m+1)^n+1 $ nodes for a Fork And Join.

\begin{figure}[!htb]
	\centering
	\includegraphics[width=0.97\textwidth]{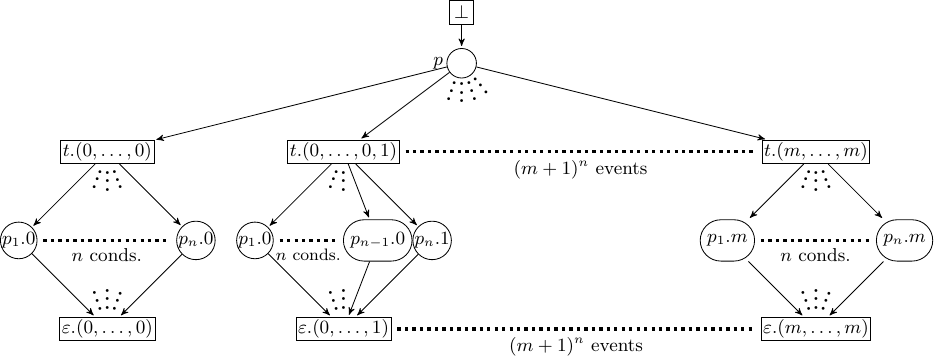}\vspace{-1mm}
	\caption{The symbolic unfolding of a Fork And Join, depending on the parameters $ m $ and $ n $.\label{fig:diamond-unf}}

 \vspace*{5mm}
	\centering
	\includegraphics[width=0.85\textwidth]{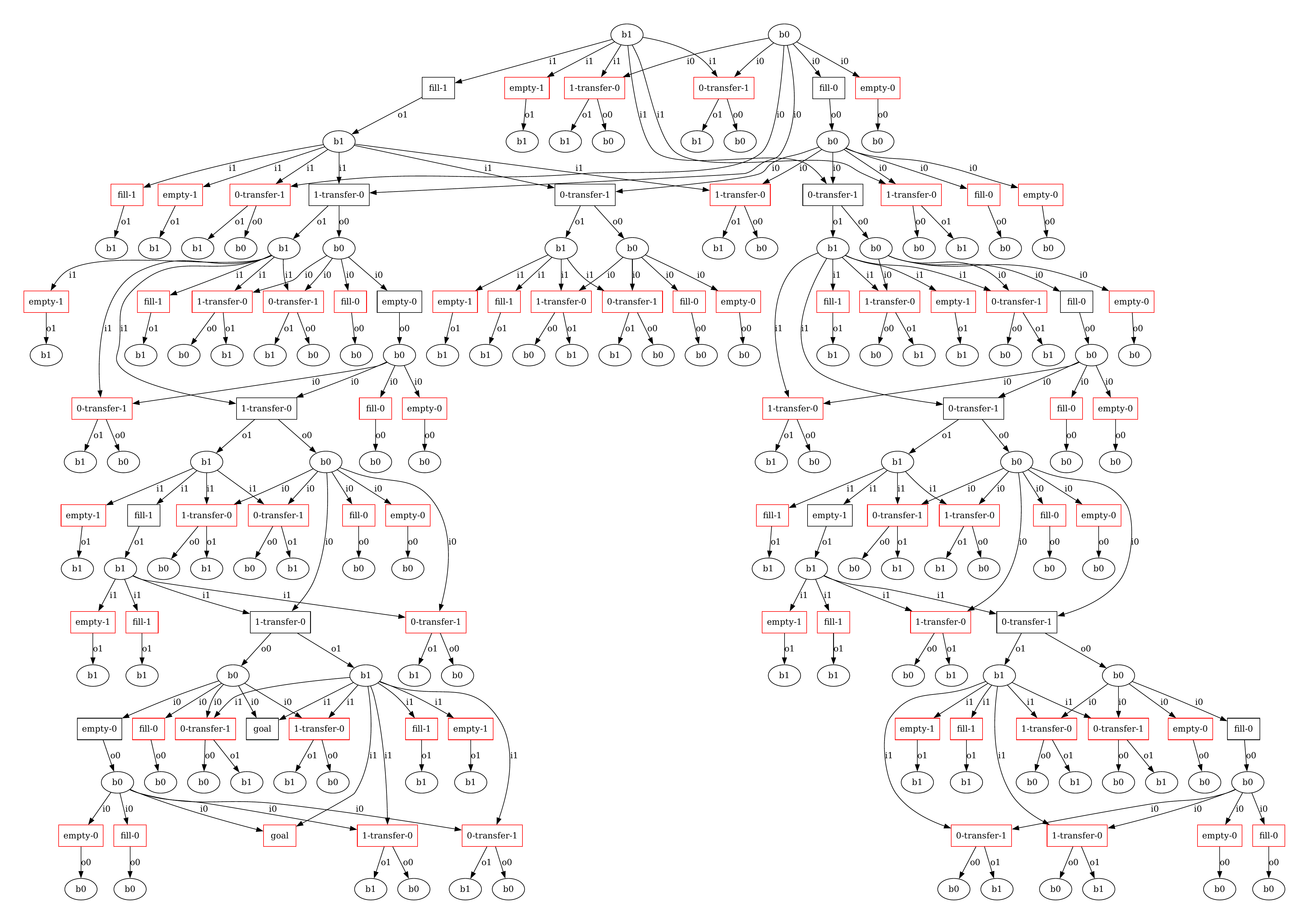}\vspace{-1mm}
	\caption{A prefix of the symbolic unfolding of the Water Pouring Puzzle for $ n=[3,5] $ and $ m=4 $, produced by \textsc{ColorUnfolder}.\label{fig:buckets-unf}}\vspace*{-5mm}
\end{figure}

\paragraph{Symbolic Unfolding of Water Pouring Puzzle.}
Figure~\ref{fig:buckets-unf}
shows the complete finite prefix of the symbolic unfolding of the Water Pouring Puzzle from Fig.~\ref{fig:buckets} with $ n=[5,3] $ and $ m=4 $,
calculated by our implementation \textsc{ColorUnfolder} \cite{ColorUnfolder}
of the generalized ERV-algorithm, Alg.~\ref{alg:complpref}.
The figure is automatically produced from the output of \textsc{ColorUnfolder}.
Cut-off events are marked by a red border.
Additionally marked by a red border are the events representing the added $\mathit{goal}$ transition
that checks the target states.
Since the net is mode-deterministic, the symbolic unfolding's skeleton coincides with the low-level unfolding.
Thus, the skeleton of the complete finite prefix in Figure~\ref{fig:buckets-unf} coincides with
the complete finite prefix of the low-level unfolding generated by the original ERV-algorithm from \cite{EsparzaRV02}.

\begin{figure}[!ht]
	\centering
	\includegraphics[width=0.92\textwidth]{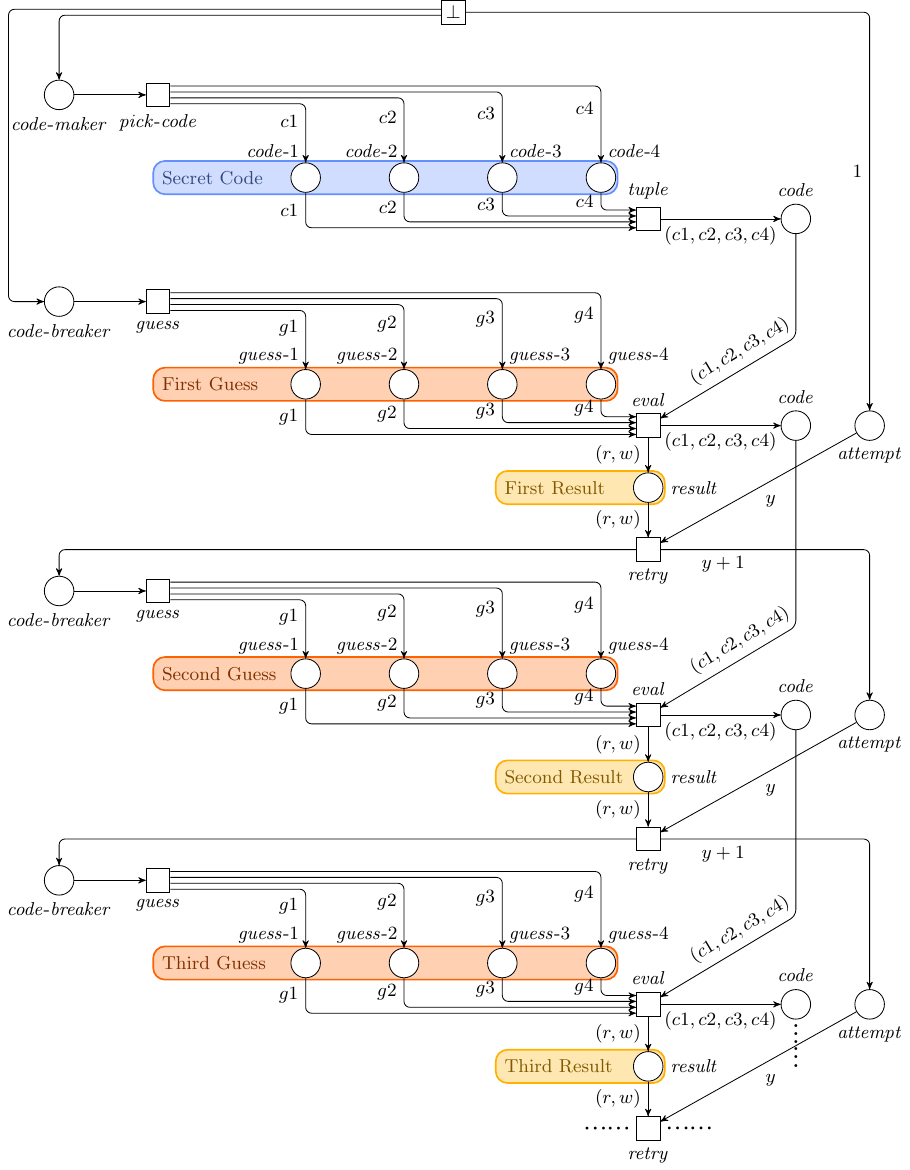}
	\caption{A prefix of the symbolic unfolding of the Mastermind net from Fig.~\ref{fig:mastermind}.\label{fig:mastermind-unf}}
\end{figure}

\paragraph{Symbolic Unfolding of Mastermind.}
Figure~\ref{fig:mastermind-unf} shows a prefix of the symbolic unfolding of the Mastermind net from Fig.~\ref{fig:mastermind} with 6 available colors, a code of length 4, and 12 possible attempts.
Combinatorial arguments give that the low-level unfolding has more than $ 10^{37} $ nodes, so we do not present it here.

\end{document}